\documentclass[a4paper,11pt,]{article}
\usepackage{graphicx}
 \usepackage{amssymb}
 \usepackage{amsmath}
\usepackage{multirow}
\usepackage{supertabular}

\newtheorem{theorem}{Theorem}[section]

\newtheorem{LMA}[theorem]{Lemma}

\global \count10 = 0

\global \count11 = 1

\global \count12 = 1

\global \count13 = 1

\global \count14 = 0

\def\squarebox#1{\hbox to #1{\hfill\vbox to #1{\vfill}}}
\def\qed{\hspace*{\fill}%
        \vbox{\hrule\hbox{\vrule\squarebox{.667em}\vrule}\hrule}\smallskip}
\newenvironment{proof}{\begin{trivlist}
\item[\hspace{\labelsep}{\bf\em\noindent Proof.~}]}{\qed\end{trivlist}}

\begin{document}

\title{
	Improved Bounds for Online Dominating Sets of Trees
}

\author{
	Koji M. Kobayashi
}

\date{}

\setlength{\baselineskip}{4.85mm}
\maketitle

\begin{abstract}
	\ifnum \count10 > 0
	%
	%
	%
	%
	$V$
	$E$
	%
	$G$
	$\{ u, v \} \in E$
	%
	%
	%
	%
	%
	%
	%
	%
	%
	Eidenbenz (Technical report, Institute of Theoretical Computer Science, ETH Z\"{u}rich 2002) and 
	Boyar
	%
	
	%
	Boyar
	%
	
	%
	\fi
	\ifnum \count11 > 0
	%
	%
	The online dominating set problem is an online variant of the minimum dominating set problem, 
	which is one of the most important NP-hard problems on graphs. 
	This problem is defined as follows: 
	Given an undirected graph $G = (V, E)$, 
	in which $V$ is a set of vertices and 
	$E$ is a set of edges. 
	We say that 
	a set $D \subseteq V$ of vertices is a {\em dominating set} of $G$ 
	if for each $v \in V \setminus D$, 
	there exists a vertex $u \in D$ such that $\{ u, v \} \in E$. 
	The vertices are revealed to an online algorithm one by one over time. 
	When a vertex is revealed, 
	edges between the vertex and vertices revealed in the past are also revealed. 
	A revealed subtree is connected at any time. 
	Immediately after the revelation of each vertex, 
	an online algorithm can irrevocably choose vertices which were already revealed 
	and must maintain a dominating set of a graph revealed so far. 
	The cost of an algorithm on a given tree is the number of vertices chosen by it, 
	and its objective is to minimize the cost. 
	Eidenbenz (Technical report, Institute of Theoretical Computer Science, ETH Z\"{u}rich, 2002) and 
	Boyar et~al.\ (SWAT 2016) studied the case 
	in which given graphs are trees. 
	They designed a deterministic online algorithm whose competitive ratio is at most three, 
	and proved that 
	a lower bound on the competitive ratio of any deterministic algorithm is two. 
	In this paper, 
	we also focus on trees.
	We establish a matching lower bound for any deterministic algorithm. 
	Moreover, 
	we design a randomized online algorithm whose competitive ratio is exactly $5/2 = 2.5$, 
	and show that the competitive ratio of any randomized algorithm is at least $4/3 \approx 1.333$. 
	\fi
\end{abstract}
%


\section{Introduction} \label{Intro}
\ifnum \count10 > 0
%
%
%
$V$
$E$
%
$\{ u, v \} \in E$
%
%
%
%
%

%
%
%
%
%
%
%
$U$
%
%
{\em 
%
$\sigma$
%
%
$C_{ON}(\sigma) \leq c C_{OPT}(\sigma)$
$ON$
%
%
%
$ON$
$ON$
\fi
\ifnum \count11 > 0
%
%
The dominating set problem is one of the most important NP-hard problems on graphs.
This problem is defined as follows: 
Given an undirected graph $G = (V, E)$, 
in which $V$ is a set of vertices and $E$ is a set of edges. 
We say that a set $D \subseteq V$ of vertices is a {\em dominating set} of $G$ 
if for each vertex $v \in V \setminus D$, there exists a vertex $u \in D$ such that $\{ u, v \} \in E$. 
The objective of the problem is to construct a minimum dominating set. 
This problem has been extensively studied for many applications, 
such as communication in ad-hoc networks (see e.g., \cite{YP2005}) and  facility location on networks (e.g., \cite{FP2001}).
The dominating set problem has also been studied in online settings \cite{KT1997,ES2002,BEFKL2016}. 
In one of the settings \cite{ES2002,BEFKL2016}, 
vertices are revealed to an online algorithm one by one, and 
edges between a revealed vertex and vertices revealed in the past are also revealed. 
The {\em input} of this setting is an undirected graph and a sequence consisting of all the vertices of the graph. 
(This sequence represents an order of the vertices revealed to an online algorithm.) 
An online algorithm holds the empty set $U$ at the beginning. 
When a new vertex is revealed, 
the algorithm can add vertices revealed so far to $U$, 
which means that an added vertex is not necessarily the newly revealed one. 
The algorithm must not remove a vertex from $U$. 
The total number of vertices is not known to an online algorithm before the final vertex is revealed. 
Thus, 
$U$ must be a dominating set immediately after the revelation of each vertex. 
The performance of online algorithms is evaluated using {\em competitive analysis} \cite{AB98,DS85}. 
The {\em cost} of an algorithm $ALG$ for an input $\sigma$ is the size of a dominating set constructed by $ALG$ for $\sigma$,  
which is denoted as $C_{ALG}(\sigma)$. 
We say that the (strict) competitive ratio of an online algorithm $ON$ is at most $c$ or 
$ON$ is {\em $c$-competitive} 
if for any input $\sigma$, 
$C_{ON}(\sigma) \leq c C_{OPT}(\sigma)$, 
in which $OPT$ is an optimal offline algorithm for $\sigma$. 
If $ON$ uses randomization, 
the expected cost of $ON$ is used. 
\fi
\ifnum \count10 > 0
%
%
\noindent
{\bf 
Eidenbenz\cite{ES2002}
Boyar
%

%
%
(i)
%
(ii)
%
%
%
(iii)
%
%
(i)
%

%
\fi
\ifnum \count11 > 0
%
%
\noindent
{\bf Previous Results and Our Results.}~
For trees, 
Eidenbenz \cite{ES2002} and Boyar et~al.\ \cite{BEFKL2016} designed a 3-competitive deterministic algorithm, 
and proved that the competitive ratio of any deterministic online algorithm is at least two 
(Boyar et~al.\ showed their results in terms of asymptotic competitive ratios, 
but the results can hold for strict competitive ratios as well). 
In this paper, 
we show the following three results for trees: 
(i) 
We prove that a lower bound on the competitive ratio of any deterministic algorithm is three. 
This bound matches the above upper bound. 
(ii)
We establish a randomized online algorithm whose competitive ratio is exactly $5/2 = 2.5$. 
This algorithm is the first non-trivial randomized algorithm for the online dominating set problem for any graph class. 
%
(iii)
We show that the competitive ratio of any randomized algorithm is at least $4/3 \approx 1.333$. 
The above results are shown with respect to the {\em strict} competitive ratio. 
However, 
it is easy to see that 
the same results for the {\em asymptotic} competitive ratios as (i) and (iii) can be shown in a quite similar way to their proofs. 
(Note that any upper bound on the strict competitive ratio is an upper bound on that on the asymptotic competitive ratio.
That is, (ii) holds for the asymptotic competitive ratio.)
\fi
\ifnum \count10 > 0
%
%
\noindent
{\bf 
Eidenbenz\cite{ES2002}
Boyar
%
%
%
%
Boyar
%
%
%

%
King
1
%
%
%

%
%
%
%
%
%
%

%
%

%

%
\fi
\ifnum \count11 > 0
%
%
\noindent
{\bf Related Results.}~
For several graph classes, 
Eidenbenz~\cite{ES2002} and Boyar et~al.\ \cite{BEFKL2016} studied online algorithms of a few variants of dominating sets, 
namely, connected dominating sets, total dominating sets and independent dominating sets. 
Their results are summarized 
in the table in Sec.~6 of \cite{ES2002} and Table~2 in Sec.~1 of \cite{BEFKL2016}. 
For example, 
they proved that the optimal competitive ratios on a bipartite graph and a planar graph are $n-1$, 
in which $n$ is the number of given vertices. 
Boyar et~al.\ \cite{BEFKL2016} defined an {\em incremental} algorithm as an algorithm 
which maintains a dominating set immediately after a new vertex is revealed. 
An online algorithm is incremental, 
but an optimal incremental algorithm knows the whole input and can perform better than any online algorithm. 
They measured the performance of online algorithms compared with an optimal {\em incremental} algorithm in addition to an optimal offline algorithm. 
Moreover, 
they compared the performance of an optimal incremental algorithm with that of an optimal offline algorithm for several graph classes, 
which is also summarized in Table~1 in Sec.~1 of \cite{BEFKL2016}. 
King and Tzeng~\cite{KT1997} studied two different variants of online dominating sets on general graphs. 
One variant is the same as the one studied in this paper, 
except that immediately after a new vertex is revealed, 
an online algorithm can choose the new one 
but cannot choose vertices revealed previously. 
In this setting, they designed a deterministic algorithm whose competitive ratio is at most $n-1$, 
and proved that the algorithm is the best possible. 
In the other variant,  
an online algorithm knows all vertices in advance, 
and at a time $i$, all the edges between the $i$-th vertex $v_{i}$ and the other vertices are revealed. 
They showed an upper bound of $3 \sqrt{n}/2$ and a lower bound of of $\sqrt{n}$ for this variant. 
For the offline setting, 
the minimum dominating set problem is one of the most significant $NP$-hard problems on graphs and has been widely studied. 
One of the most important open problems is to develop exact (exponential) algorithms 
(see, e.g. \cite{FKW2004,FGK2005,GF2006,SI2008,RB2011,IY2011}). 
The current fastest algorithm solves this problem in $O(1.4864^n)$ time and polynomial space \cite{IY2011}. 
Moreover, 
many variants have been proposed by putting additional constraints on the original dominating set problem and 
have been extensively studied: 
for example, 
connected domination, independent domination and total domination
(see, .e.g. \cite{DW2013},\cite{GH2013} and \cite{HY2013}, respectively). 
\fi
%

\section{Preliminaries} \label{sec:2}
\subsection{Model Description} \label{sec:model}
\ifnum \count10 > 0
%
%
%
%
$i$
$v_i$
%
%
%
$V$
$E$
$S$
%
%
%
%
%
%
%
%
%
%
$ALG$
$C_{ON}(\sigma) \leq c C_{OPT}(\sigma)$
(${\mathbb E}[C_{ON}(\sigma)] \leq c C_{OPT}(\sigma)$)
%

%
%

%
\fi
\ifnum \count11 > 0
%
%
We are given an undirected tree and 
its vertices are revealed to an online algorithm one by one over time. 
The total number of the vertices is not known to the online algorithm up to the end of the input. 
When the $i$-th vertex $v_i$ is revealed to the online algorithm, 
all the edges between $v_i$ and $v_j$ such that $j < i$ are also revealed. 
Except for the first revealed vertex, 
a newly revealed vertex has exactly one edge to a vertex revealed previously. 
That is, 
a revealed subtree is connected at any time. 
An {\em input} of the problem is a three-tuple of the form $(V, E, S)$, 
in which 
$V$ is the set of all the vertices of a given tree, 
$E$ is the set of all the undirected edges of the tree, and 
$S$ is a sequence consisting of all the vertices in $V$. 
$S$ represents an order of the vertices revealed to an online algorithm. 
An algorithm has the empty set $U$ before the first vertex is revealed. 
The algorithm can add vertices into $U$ immediately after the revelation of each vertex, 
and it is necessary for $U$ to be a dominating set of the given tree at the end of the input.
If the algorithm is online, 
it does not know when the input has ended, 
and thus 
$U$ must be a dominating set immediately after each vertex is revealed. 
Once a vertex is added into $U$, 
it must not be removed from $U$ later. 
The {\em cost} of the algorithm for an input $\sigma$ is 
the number of vertices in $U$ at the end of $\sigma$, 
and the objective of the problem is to minimize the cost. 
We evaluate the performance of an online algorithm using competitive analysis. 
We say that 
the {\em competitive ratio} of a deterministic online algorithm $ON$ is at most $c$ 
if for any input $\sigma$, 
$C_{ON}(\sigma) \leq c C_{OPT}(\sigma)$. 
If $ON$ is a randomized online algorithm, 
then the expected cost of $ON$ is used, 
which is denoted by ${\mathbb E}[C_{ON}(\sigma)]$. 
If for any input $\sigma$, 
${\mathbb E}[C_{ON}(\sigma)] \leq c C_{OPT}(\sigma)$, 
then we say that the competitive ratio of a randomized online algorithm $ON$ is at most $c$ against any oblivious adversary. 
If the number of vertices in a given tree is one, 
the cost ratio of any algorithm is clearly one. 
Thus, we assume that this number is at least two. 
\fi
%

\subsection{Notation and Definitions}\label{sec:notation}

\ifnum \count10 > 0
%
%
%
%
%
%
%
%
$\{ v, u \} \in E$%
2
%
%
$v$
$v$
%
%
$D_{ALG}(v)$
$v$
$ALG$
%
%
%
$D(\sigma)_{ALG}$
%
$\sigma$
%
%
$A$
%
%
${deg}_{u}(v)$
%
%
${deg}(v)$
%
%
$u$
$u$
%
%
$U$
%
%
%
{\em $U$
\fi
\ifnum \count11 > 0
%
%
In this section, 
we give some definitions and notation used throughout this paper. 
For any $i (= 1, 2, \ldots)$, 
we use $v_{i}$ to denote the $i$-th revealed vertex to an online algorithm 
(the first revealed vertex $v_{1}$ appears frequently in this paper). 
We say that vertices $v$ and $u$ are {\em adjacent} 
if $\{ v, u \} \in E$, 
in which $E$ is the set of all the edges of a given graph. 
When a vertex $v$ is revealed such that 
$v$ is adjacent to a vertex $u$ which was revealed before $v$, 
then we say that $v$ {\em arrives} at $u$. 
For any vertex $v$ and any online algorithm $ON$, 
$D_{ON}(v)$ denotes a dominating set constructed by $ON$ of a revealed graph up to the time of the revelation of $v$. 
We will omit $ON$ from the notation when it is clear from the context. 
For an algorithm $ALG$ including an offline algorithm, 
$D_{ALG}(\sigma)$ denotes a dominating set constructed by $ALG$ after the end of the input $\sigma$. 
We will omit $\sigma$ from the notation when it is clear from the context. 
For a vertex $v$, 
we say that $ALG$ {\em selects} $v$ 
if $v \in D_{ALG}$. 
For vertices $u$ and $v$ such that $u$ is revealed after $v$, 
${deg}_{u}(v)$ denotes the degree of $v$ immediately after $u$ is revealed. 
${deg}(v)$ denotes the degree of $v$ after the end of the input. 
For a vertex $v$ and a vertex $u$ revealed after $v$, 
we say that $u$ is a {\em descendant} of $v$ 
if any vertex on the simple path from $v$ to $u$ is revealed after $v$. 
The {\em cost} of a deterministic algorithm $ALG$ for a vertex set $U$ is the number of vertices selected by $ALG$ in $U$. 
That is, it is the number of vertices in $U \cap D_{ALG}$. 
Moreover, if $U$ contains only one vertex, 
then we simply say the {\em cost for the vertex}. 
In the same way, 
we use the term ``the expected cost of $ALG$ for $U$ (or a vertex)''
if $ALG$ is a randomized algorithm. 
\fi
%

\section{Deterministic Lower Bound} \label{sec:det}
\subsection{Overview of Proof} 
\ifnum \count10 > 0
%
%
%
%
%
%
%
%
%
%
%
%
$T$-set
$T$-set
%
%
2
%
%
$OPT$
$OPT$
%
$ON$
%

%
%
%
$\ell$ modulo $3 = i$
%
%
%
(1)
2
(2)
1
(3)
1
(4)
%
%
(1)
(3)
$ON$
(5) 
$T_{1}$-set
$T_{2}$-set
%
%
%
(1)
%
%
%
%
%
%
$u$
%
%
4
%
%
%

%
\fi
\ifnum \count11 > 0
%
%
We first outline an input to obtain our lower bound. 
The tree of the input is constructed according to two routines. 
The tree can be divided into several subtrees satisfying some properties and 
we evaluate the competitive ratio for each set of some subtrees. 
One of the routines appoints a vertex as the root to construct a subtree, 
which is called a {\em base vertex}. 
The other routine constructs several subtrees with at most two leaves, 
each of which arises from the base vertex. 
The set of all the vertices excluding the root in each of the subtrees is called a {\em $T$-set}. 
It depends on the behavior of an online algorithm $ON$ 
how many $T$-sets are constructed and 
how many leaves and inner vertices composing $T$-sets are. 
If a $T$-set contains two leaves, 
the leaves share the adjacent vertex. 
For each $T$-set, 
$OPT$ selects one vertex for every consecutive three vertices starting with the parent of a leaf in it. 
If the degree of a vertex selected by $OPT$ is two, 
$ON$ selects the vertex and the two adjacent vertices. 
Otherwise, that is its degree is at least three,
$ON$ selects at least three vertices from the vertex and all the adjacent vertices. 
Let us explain the proof more in detail. 
If a $T$-set contains sufficiently many inner vertices, 
it is called a {\em $T_{3}$-set}. 
Otherwise, 
a $T$-set such that $\ell$ modulo $3 = i$ is called a {\em $T_{i}$-set}, 
in which $\ell$ is the length from the base vertex to a leaf in the $T$-set. 
One of the routines tries to force $ON$ to construct one of the following four sets of $T$-sets from a base vertex (Fig.~\ref{fig:tsetex}): 
(1)
a set of two $T_{1}$-sets and at least zero $T_{0}$-set, 
(2)
a set of one $T_{2}$-set and at least zero $T_{0}$-set, 
(3)
a set of one $T_{3}$-set, at most one $T_{1}$-set and at least zero $T_{0}$-set, 
and 
(4)
a set of sufficiently many $T_{0}$-sets and at most one $T_{1}$-set. 
The cost ratios of these $T$-sets are three for (1) or (2) and approximately three for (3) or (4), respectively. 
$ON$ can construct none of these sets. 
Namely, 
(5) 
$ON$ constructs one $T_{1}$-set and then 
does one $T_{2}$-set 
(Further, 
$ON$ may also construct $T_{0}$-sets). 
In this case, 
the routine partitions the $T_{2}$-set into a vertex $u$, a $T_{1}$-set and a $T_{0}$-set 
(Fig.~\ref{fig:dlow}). 
This $T_{0}$-set and all the $T$-sets in (5) except for the partitioned $T_{2}$-set compose a set of $T$-sets of (1). 
Then, 
the routine finishes constructing a subtree from the current base vertex, 
whose cost ratio is three, 
and appoints $u$ as a new base vertex. 
One $T_{0}$-set, which is constructed from the above partition of the $T_{2}$-set, 
belongs to the new base vertex $u$. 
Since the set of $T$-sets of $u$ is not classified into any of the above four categories, 
the routine continues to construct subtrees for $u$. 
This is how the routine tries to construct one of the four sets of $T$-sets for all base vertices 
and to achieve a lower bound of (approximately) three. 
Therefore, 
we have the following theorem: 
\fi
\begin{theorem}\label{thm:det_low}
	\ifnum \count10 > 0
	%
	%
	%
	
	%
	\fi
	\ifnum \count11 > 0
	%
	%
	For any $\varepsilon > 0$, 
	the competitive ratio of any deterministic online algorithm is at least $3 - \varepsilon$. 
	\fi
\end{theorem}
\ifnum \count12 > 0
\begin{figure*}
	 \begin{center}
	  \includegraphics[width=160mm]{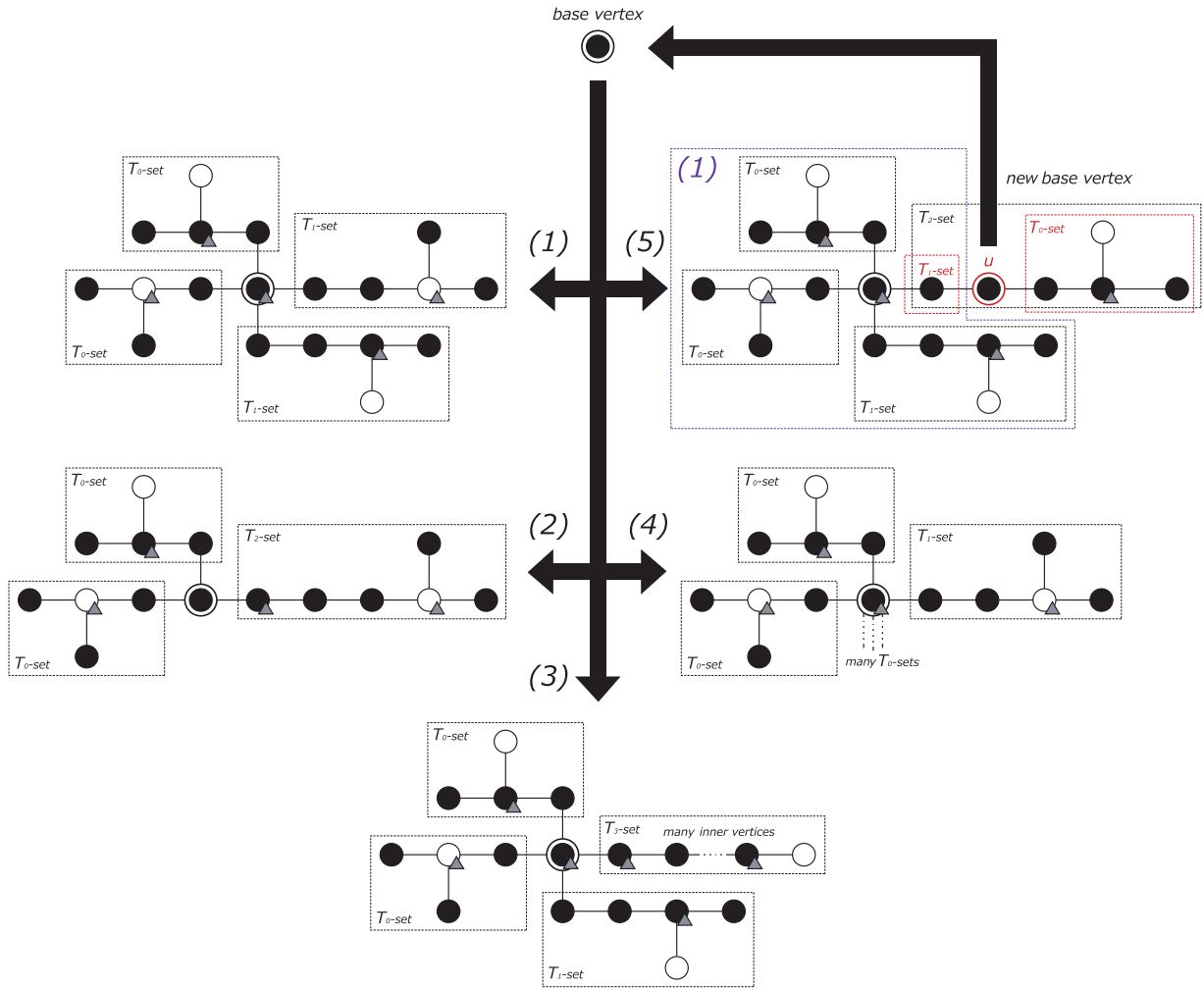}
	 \end{center}
	 \caption{
\ifnum \count10 > 0
%
%
(1)
1
%
(5)
$u$
$u$
\fi
\ifnum \count11 > 0
%
%
An example of the five sets of $T$-sets from (1) to (5). 
Highlighted vertices denote base vertices. 
Black vertices denote vertices selected by $ON$. 
Vertices with a gray triangle denote vertices selected by $OPT$. 
One of the five sets of $T$-sets is constructed for a base vertex. 
If (5) is constructed, 
the $T_{2}$-set in the set is partitioned into a new base vertex $u$, 
a $T_{1}$-set and a $T_{0}$-set. 
After that, 
the routines force $ON$ to construct one of the five sets of $T$-sets for $u$ recursively. 
$u$ is not dominated by $OPT$ yet, but is dominated later. 
\fi
			}
	\label{fig:tsetex}
\end{figure*}
\fi

\subsection{Proof of Theorem~\ref{thm:det_low}} 
\ifnum \count10 > 0
%
%
%
%
$ON$
%
$ON$
%
%
%
%
%
%
%
%
%
$ON$
%
%
%

%
\fi
\ifnum \count11 > 0
%
%
Consider a deterministic online algorithm $ON$ for the input $\sigma$, 
which will be defined below. 
Without loss of generality, 
we assume that 
$ON$ does not select any vertex 
when a vertex arrives at a selected vertex. 
Moreover, 
we assume that 
when a vertex arrives at a vertex which is not selected, 
$ON$ selects only one of the two vertices. 
First we define a routine to construct a subinput. 
When this routine is executed, 
it takes a vertex which has already been revealed to $ON$ as a parameter. 
We call such a vertex a {\em base vertex}
(whose formal definition will be given later). 
Let us outline this routine. 
The routine constructs a subtree with at most two leaves from a base vertex, 
which is the root of this subtree. 
The routine continues increasing inner vertices of the subtree by the time $ON$ selects a newly revealed vertex. 
When $ON$ selects it, 
the routine sets it to be a leaf and reveals another leaf. 
{\tt MaxLength} is a sufficiently large positive integer 
such that {\tt MaxLength} modulo $3 = 2$, 
which is used in the routine. 
\fi
\ifnum \count10 > 0
%
%
\noindent\vspace{-1mm}\rule{\textwidth}{0.5mm} 
\vspace{-3mm}
{{\sc SubtreeRoutine}(
\rule{\textwidth}{0.1mm}
%
	%
	%
	{\bf\boldmath Step~T1:} 
		%
		$u_{0} := v, j := 0$. \\
	{\bf\boldmath Step~T2:} 
		%
		$j =$ {\tt MaxLength}
		$j := j + 1$. \\
%
	{\bf\boldmath Step~T3:} 
		%
		$u_{j-1}$
%
%
\hspace*{2mm}
		{\bf\boldmath Case~T3.1 ($ON$
		%
%
				$u_{j-1}$
\hspace*{2mm}
		{\bf\boldmath Case~T3.2 ($ON$
		%
%
			Step~T2
\noindent\vspace{-1mm}\rule{\textwidth}{0.5mm} 
\fi
\ifnum \count11 > 0
%
%
%
\noindent\vspace{-1mm}\rule{\textwidth}{0.5mm} 
\vspace{-3mm}
{{\sc SubtreeRoutine}(vertex $v$)}\\
\rule{\textwidth}{0.1mm}
%
	%
	%
	{\bf\boldmath Step~T1:} 
		%
		$u_{0} := v$ and $j := 0$. \\
	{\bf\boldmath Step~T2:} 
		%
		If $j =$ {\tt MaxLength}, then finish. 
		$j := j + 1$. 
%
	{\bf\boldmath Step~T3:} 
		%
		A vertex $u_{j}$ arrives at a vertex $u_{j-1}$. 
		Execute one of the following two cases. \\
\hspace*{2mm}
		{\bf\boldmath Case~T3.1 ($ON$ selects $u_{j}$):} 
		%
%
				A vertex $u'_{j-1}$ arrives at $u_{j-1}$, and finish. \\
\hspace*{2mm}
		{\bf\boldmath Case~T3.2 ($ON$ does not select $u_{j}$):} 
		%
%
			Go to Step~T2. \\
\noindent\vspace{-1mm}\rule{\textwidth}{0.5mm} 
\fi
\ifnum \count10 > 0
%
%
$ON$
Case~T3.1
%
Case~T3.2
$ON$
%
%
%
$v$
$u_{i}$
$u_{t-1}$
(i) 
${g}(v, a) = \{ u_{1}, u_{2}, \ldots, u_{t}, u'_{t-1} \}$
%
(ii) 
$u_{x}$
$x \leq t-2$
${g}(v, a) = \{ u_{1}, u_{2}, \ldots, u_{x-1} \}$
$x \geq t-1$
${g}(v, a) = \{ u_{1}, u_{2}, \ldots, u_{t-1}, u'_{t-1} \}$
%
%
$u'_{t-1}$
${g}(v, a)$
%
(ii)
(ii)
%
%
{\em $v$
%
%
$\ell(v, a) = \max\{ i \mid u_{i} \in {g}(v, a) \}$
$\ell(v, a)$
$g(v, a)$
(i) $\ell(v, a) = {\tt MaxLength}$
(ii) $\ell(v, a) < {\tt MaxLength}$
(iii) $\ell(v, a) < {\tt MaxLength}$
(vi) $\ell(v, a) < {\tt MaxLength}$
%
%
${g}(v, a)$
%

%
\noindent
{\bf $T_{3}$-set
$\ell(v, a) = {\tt MaxLength}$
$u_{i} \hspace*{1mm} (i = 1, 2, \ldots, \ell(v, a)-2)$
$u_{\ell(v, a)-1}$
$D_{ON}(u_{\ell(v, a)})$
%

%
\noindent
{\bf $T_{0}$-set, $T_{1}$-set, $T_{2}$-set
$u_{i} \hspace*{1mm} (i = 1, 2, \ldots, \ell(v, a)-2, \ell(v, a))$
%

%
%
$OFF$
%
$T$-set
$OFF$
%

%
\noindent
{\bf $T_{0}$-set
$OFF$
\noindent
{\bf $T_{1}$-set
$OFF$
\noindent
{\bf $T_{2}$-set
$OFF$
\noindent
{\bf $T_{3}$-set
$OFF$
%

%
%
%
$OFF$%
$T_{1}$-set
$T_{0}$-set
%
%
%
%
%
%
{\sc TreeRoutine}
$v_{1}$
{\sc SubtreeRoutine}
%
Case~3.3.2
$T$-set
%
%
%
%
%
(1)$T_{1}$-set
(2)$T_{2}$-set
(3)$T_{3}$-set
(4)
%
%
$OFF$
%

%
%
{\tt MaxT}$_{0}$
\fi
\ifnum \count11 > 0
%
%
Note that 
$ON$ selects either $u_{j-1}$ or $u'_{j-1}$ in Case~T3.1 
because any online algorithm must select at least one of two consecutive vertices.
Also, 
note that $ON$ selects $u_{j-1}$ 
if it is not selected immediately before the execution of Case~T3.2 
because $D_{ON}(u_{j})$ must contain either $u_{j-1}$ or $u_{j}$. 
($u_{0}$ is already selected at the beginning of the routine, which is shown later.)
Suppose that the routine is executed with a vertex $v$ as the parameter 
and $j = t$ holds when the routine finishes. 
Also, suppose that 
the routine gives vertices $u_{i} \hspace*{1mm} (i \in [1, t])$ and $u'_{t-1}$ (if any) 
such that $u_{1}$ is the $a$-th vertex arriving at $v$, 
$u_{i+1}$ arrives at $u_{i}$, 
and $u'_{t-1}$ arrives at $u_{t-1}$, 
during the execution of the routine. 
Then, 
let us define the vertex set ${g}(v, a)$ as follows: 
(i) 
If for all $i \in [1, t]$, $u_{i}$ is not a base vertex, 
then 
${g}(v, a) = \{ u_{1}, u_{2}, \ldots, u_{t}, u'_{t-1} \}$. 
(ii) 
Suppose that for any $i \in [1, x-1]$, $u_{i}$ is not a base vertex
but $u_{x}$ is a base vertex. 
If $x \leq t-2$, 
${g}(v, a) = \{ u_{1}, u_{2}, \ldots, u_{x-1} \}$. 
Otherwise, that is, 
if $x \geq t-1$, 
${g}(v, a) = \{ u_{1}, u_{2}, \ldots, u_{t-1}, u'_{t-1} \}$. 
($u'_{t-1}$ is not a base vertex.
Thus, 
${g}(v, a)$ does not contain any base vertex. 
Immediately after the execution of the routine, 
(i) always holds. 
$u_{2}$ in (i) can be set to be a base vertex later 
and then (ii) can hold. 
Namely, 
both $x = 2$ and $t$ modulo $3 = 2$ in (ii). 
These facts will be shown later.) 
%
%
We call this vertex set ${g}(v, a)$ the $a$-th {\em $T$-set} of $v$, and 
say that {\em $v$ has the $T$-set}. 
Also, 
we define $\ell(v, a) = \max\{ i \mid u_{i} \in {g}(v, a) \}$, 
which is called the {\em length} of ${g}(v, a)$. 
Let us classify ${g}(v, a)$ into the following four categories: 
(i) ${g}(v, a)$ such that $\ell(v, a) = {\tt MaxLength}$ is called a {\em $T_{3}$-set}, 
(ii) ${g}(v, a)$ such that $\ell(v, a) < {\tt MaxLength}$ and $\ell(v, a)$ modulo $3 = 0$ is called a {\em $T_{0}$-set},
(iii) ${g}(v, a)$ such that $\ell(v, a) < {\tt MaxLength}$ and $\ell(v, a)$ modulo $3 = 1$  is called a {\em $T_{1}$-set}, and 
(iv) ${g}(v, a)$ such that $\ell(v, a) < {\tt MaxLength}$ and $\ell(v, a)$ modulo $3 = 2$ is called a {\em $T_{2}$-set}. 
Next, 
we evaluate the cost of $ON$ for ${g}(v, a)$, 
that is, 
the number of vertices in $D_{ON} \cap {g}(v, a)$. 
\noindent
{\bf Cost for a $T_{3}$-set:}
$\ell(v, a) = {\tt MaxLength}$ and 
$D_{ON}(u_{t})$ (namely, $D_{ON}$) contains 
$u_{i} \hspace*{1mm} (i = 1, 2, \ldots, \ell(v, a)-2)$ 
and 
either $u_{\ell(v, a)-1}$ or $u_{\ell(v, a)}$. 
Thus, 
the number of vertices selected by $ON$ in ${g}(v, a)$ is ${\tt MaxLength} - 1$. 
\noindent
{\bf Costs for a $T_{0}$-set, a $T_{1}$-set, and a $T_{2}$-set:}
$D_{ON}$ contains $u_{i} \hspace*{1mm} (i = 1, 2, \ldots, \ell(v, a)-2, \ell(v, a))$. 
$D_{ON}$ also contains either $u_{\ell(v, a)-1}$ or $u'_{\ell(v, a)-1}$. 
Thus, 
the number of vertices selected by $ON$ in ${g}(v, a)$ is $\ell(v, a)$. 
We define the following offline algorithm $OFF$ to give an upper bound on the cost of $OPT$. 
Roughly speaking, 
$OFF$ selects one vertex every consecutive three vertices 
starting with $u_{\ell(v, a)-1}$ to deal with $u'_{\ell(v, a)-1}$. 
Moreover, 
a $T$-set does not contain a base vertex $v (= u_{0})$, 
but $OFF$ selects $v$ to cover $u_{1}$ for a $T_{1}$-set. 
$OFF$'s cost for a $T_{1}$-set includes the cost for $v$.   
\noindent
{\bf Cost for a $T_{0}$-set:}
$OFF$ selects $\ell(v, a)/3$ vertices of $u_{2}, u_{5}, u_{8}, \ldots, u_{\ell(v, a)-1}$. 
\noindent
{\bf Cost for a $T_{1}$-set:}
$OFF$ selects $(\ell(v, a)+2)/3$ vertices of $u_{0}, u_{3}, u_{6}, \ldots, u_{\ell(v, a)-1}$. 
\noindent
{\bf Cost for a $T_{2}$-set:}
$OFF$ selects $(\ell(v, a)+1)/3$ vertices of $u_{1}, u_{4}, \ldots, u_{\ell(v, a)-1}$. 
\noindent
{\bf Cost for a $T_{3}$-set:}
$OFF$ selects $(\ell(v, a)+1)/3$ vertices of $u_{1}, u_{4}, \cdots, u_{\ell(v, a)-1}$
($OFF$'s selection begins with $u_{1}$ for $u_{0}$, and $\ell(v, a) = {\tt MaxLength}$ such that ${\tt MaxLength}$ modulo 3 = 2 by definition). 
$ON$ selects all base vertices, 
which is shown later, 
while $OFF$ selects a base vertex 
if its $T$-sets contain a $T_{1}$-set or are all $T_{0}$-sets. 
Then, 
we introduce another routine {\sc TreeRoutine} to construct the input $\sigma$. 
This routine uses {\sc SubtreeRoutine} as a subroutine. 
Now we formally define a base vertex. 
A {\em base vertex} denotes a vertex set to be the value {\tt BaseVtx} in the routine. 
(This is because this definition also applies to a vertex $u_{2}$ when the routine finishes in Case~3.3.2.)
At the beginning, 
the routine sets $v_{1}$ to be the first base vertex and 
calls {\sc SubtreeRoutine}. 
Obtaining a $T$-set by the call, 
the routine sets a vertex in the $T$-set to be {\tt BaseVtx} 
if it executes Case~3.3.2. 
Thus, 
more than one vertices can be base vertices. 
After the definition of the routine, 
we will prove that 
all the base vertices except at most one are classified into the following four categories: 
(1) a base vertex with two $T_{1}$-sets and at least zero $T_{0}$-sets, 
(2) a base vertex with one $T_{2}$-set and at least zero $T_{0}$-sets,
(3) a base vertex with one $T_{3}$-set, at most one $T_{1}$-set and at least zero $T_{0}$-sets, and
(4) a base vertex with sufficiently many $T_{0}$-sets and at most one $T_{1}$-set. 
Moreover, 
we will evaluate the cost for each base vertex with its $T$-sets 
and prove that each cost ratio of $ON$ to $OFF$ is at least three. 
Now we define the routine, 
which uses {\tt MaxT}$_{0}$ and {\tt MaxT}$_{1}$ as sufficiently large positive integers. 
\fi
\ifnum \count10 > 0
%
%
\noindent\vspace{-1mm}\rule{\textwidth}{0.5mm} 
\vspace{-3mm}
{{\sc TreeRoutine}}\\
\rule{\textwidth}{0.1mm}
%
	%
	%
	{\bf\boldmath 
		%
		{\tt Count} $:= 1$
		$v_{1}$
		{\tt BaseVtx}$ := v_{1}$
	{\bf\boldmath Step~1:} 
		%
		{\tt CntT$_{0}$} $:= 0,$ {\tt CntT$_{1}$} $:= 0$
	{\bf\boldmath Step~2:} 
		%
		{\sc SubtreeRoutine}({\tt BaseVtx})
		\\
	{\bf\boldmath Step~3:} 
		%
		%
%
\hspace*{2mm}
		{\bf\boldmath Case~3.1 (${g}(v, a)$
		%
		{\tt CntT}$_{0}$ $:=$ {\tt CntT}$_{0}$ $+ 1$
\hspace*{4mm}
			{\bf\boldmath Case~3.1.1 ({\tt CntT}$_{0}$ $<$ {\tt MaxT}$_{0}$
		%
%
				Step~2
\hspace*{4mm}
			{\bf\boldmath Case~3.1.2 ({\tt CntT}$_{0}$ $=$ {\tt MaxT}$_{0}$
		%
%
		%
%
\hspace*{2mm}
		{\bf\boldmath Case~3.2 (${g}(v, a)$
		%
		{\tt CntT}$_{0}$ $:=$ {\tt CntT}$_{0}$ $+ 1$
\hspace*{4mm}
			{\bf\boldmath Case~3.2.1 ({\tt CntT}$_{1}$ $= 0$
		%
%
				{\tt CntT}$_{1}$ $:= 1$
\hspace*{4mm}
			{\bf\boldmath Case~3.2.2 ({\tt CntT}$_{1}$ $= 1$
		%
%
		%
%
\hspace*{2mm}
		{\bf\boldmath Case~3.3 (${g}(v, a)$
\hspace*{4mm}
			{\bf\boldmath Case~3.3.1 ({\tt CntT}$_{1}$ $= 0$
		%
%
		%
%
\hspace*{4mm}
			{\bf\boldmath Case~3.3.2 ({\tt CntT}$_{1}$ $= 1$
\hspace*{6mm}
				$u_{1}$
				${g}(v, a)$
				%
				{\tt BaseVtx} $:= u_{2}$
				%
				{\tt Count} $:=$ {\tt Count}$+1$
				Step~1
				\\
\hspace*{2mm}
		{\bf\boldmath Case~3.4 (${g}(v, a)$
		%
		%
%
\noindent\vspace{-1mm}\rule{\textwidth}{0.5mm} 
\fi
\ifnum \count11 > 0
%
%
\noindent\vspace{-1mm}\rule{\textwidth}{0.5mm} 
\vspace{-3mm}
{{\sc TreeRoutine}}\\
\rule{\textwidth}{0.1mm}
%
	%
	%
	{\bf\boldmath Initialize:} 
		%
		{\tt Count} $:= 1$. 
		$v_{1}$ is revealed and {\tt BaseVtx} $:= v_{1}$. \\
	{\bf\boldmath Step~1:}
		%
		{\tt CntT$_{0}$} $:= 0$ and {\tt CntT$_{1}$} $:= 0$. \\
	{\bf\boldmath Step~2:}\\
\hspace*{2mm}
		Call {\sc SubtreeRoutine}({\tt BaseVtx}). 
		Suppose that it gives a $T$-set ${g}(v, a)$, 
		in which $v = $ {\tt BaseVtx}. 
		\\
	{\bf\boldmath Step~3:} 
		%
		Execute one of the following four cases. \\
\hspace*{2mm}
		{\bf\boldmath Case~3.1 (${g}(v, a)$ is a $T_{0}$-set):} 
		%
		{\tt CntT}$_{0}$ $:=$ {\tt CntT}$_{0}$ $+ 1$. \\
\hspace*{4mm}
			{\bf\boldmath Case~3.1.1 ({\tt CntT}$_{0}$ $<$ {\tt MaxT}$_{0}$):} 
		%
%
				Go to Step~2. \\
\hspace*{4mm}
			{\bf\boldmath Case~3.1.2 ({\tt CntT}$_{0}$ $=$ {\tt MaxT}$_{0}$):} 
		%
%
				Finish. \\
\hspace*{2mm}
		{\bf\boldmath Case~3.2 (${g}(v, a)$ is a $T_{1}$-set):} 
		%
		{\tt CntT}$_{0}$ $:=$ {\tt CntT}$_{0}$ $+ 1$. \\
\hspace*{4mm}
			{\bf\boldmath Case~3.2.1 ({\tt CntT}$_{1}$ $= 0$):} 
		%
%
				{\tt CntT}$_{1}$ $:= 1$ and go to Step~2. \\
\hspace*{4mm}
			{\bf\boldmath Case~3.2.2 ({\tt CntT}$_{1}$ $= 1$):} 
		%
%
				Finish. \\
\hspace*{2mm}
		{\bf\boldmath Case~3.3 (${g}(v, a)$ is a $T_{2}$-set):} \\
\hspace*{4mm}
			{\bf\boldmath Case~3.3.1 ({\tt CntT}$_{1}$ $= 0$):} 
		%
%
				Finish. \\
\hspace*{4mm}
			{\bf\boldmath Case~3.3.2 ({\tt CntT}$_{1}$ $= 1$):} \\
\hspace*{6mm}
				Let $u_{1} \in {g}(v, a)$ be the vertex arriving at {\tt BaseVtx} and 
				let $u_{2} \in {g}(v, a)$ be the vertex which first\\ 
\hspace*{6mm}	arrived at $u_{1}$. 
				${g}(v, a)$ is a $T_{2}$-set and hence $u_{2}$ exists (Fig.~\ref{fig:dlow}). 
				{\tt BaseVtx} $:= u_{2}$. 
				If {\tt Count} $=$\\
\hspace*{6mm}	
				{\tt MaxT}$_{1}$, 
				then finish. 
				Otherwise, 
				{\tt Count} $:=$ {\tt Count}$+1$ and go to Step~1. 
				\\
\hspace*{2mm}
		{\bf\boldmath Case~3.4 (${g}(v, a)$ is a $T_{3}$-set):} 
		%
			Finish. \\
\noindent\vspace{-1mm}\rule{\textwidth}{0.5mm} 
\fi
\ifnum \count12 > 0
\begin{figure*}[ht]
	 \begin{center}
	  \includegraphics[width=130mm]{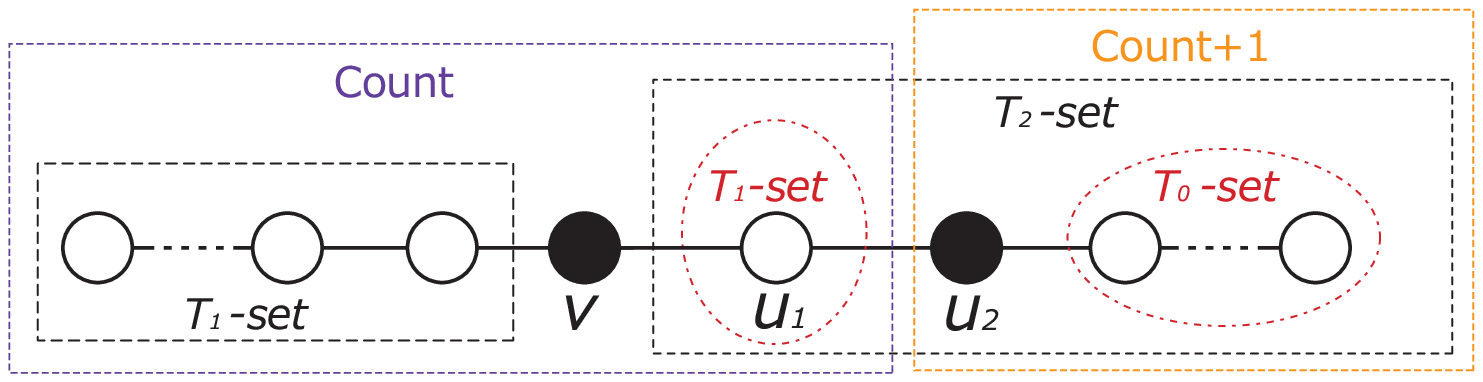}
	 \end{center}
	 \caption{
\ifnum \count10 > 0
%
%
Case~3.3.2
{\sc SubtreeRoutine}
$T_{2}$-set ${g}(v, a)$
Case~3.3.2
$T_{0}$-set ${g}(u_{2}, 1)$
%
$v$
$u_{2}$
$u_{2}$
\fi
\ifnum \count11 > 0
%
%
Execution example of Case~3.2.2. 
$v$ is a base vertex when Case~3.2.2 is executed. 
{\sc SubtreeRoutine} is executed with $v$ and 
constructs a $T_{2}$-set ${g}(v, a)$. 
Since $u_{2}$ is set to be a new base vertex in Case~3.3.2, 
the original $T_{2}$-set is divided into one $T_{1}$-set ${g}(v, a)$ whose length is one and 
one $T_{0}$-set ${g}(u_{2}, 1)$. 
Thus, 
if $v$ is the {\tt Count}-th base vertex, 
then $u_{2}$ is the {\tt Count}$+1$-st base vertex. 
Also, 
$u_{1}$ becomes a $T$-set of $v$ and 
all the descendants of $u_{2}$ become a $T_{0}$-set of $u_{2}$. 
\fi
			}
	\label{fig:dlow}
\end{figure*}
\fi
\ifnum \count10 > 0
%
%
$v_{1}$
${\tt BaseVtx} = v_{1}$
%
%
%
%
%
%
%
Case~3.1.2, 3.2.2, 3.3.1, 3.4
$y < ${\tt MaxT}$_{1}$
$y$
%
Case~3.3.2
$y = ${\tt MaxT}$_{1}$
$y+1$
%
%
$z$
\noindent
{\bf\boldmath (1):} 
%
%
%
Case~3.3.2
Case~3.1.1
%
%
%
Case~3.3.2
${g}(v, a)$
$u_{1}$
$u_{1},u_{2},u'_{1}$
%
%
%
%
$ON$
$u_{2}$
%

%
$v$
Case~3.1.1
1
%
Case~3.2.1
$v$
Case~3.3.2
$v$
$T_{1}$-set
%
%
$v$
$x-2 (\geq 0)$
%
%
%
%
2
$x-2$
$\sum_{i = 1}^{x} t_{i} + 1$
1
%
2
($v$
$T_{0}$-set
%
$(\sum_{i = 1}^{x} t_{i} + 1)/3$
%

%
\noindent
{\bf\boldmath (2):} 
%
%
$z = y$
Case~3.1.2, 3.2.2, 3.3.1, 3.3.2, 3.4
\noindent
{\bf (2-1) Case~3.2.2
%
$v$
Case~3.2.2
$v$
%
Case~3.1.1
%
(1)
$z \geq 2$
$z = y-1$
$v$
%
%
$v$
%
%
%

%
\noindent
{\bf (2-2) Case~3.3.1
$v$
%
%
$T_{2}$-set
$x-1 (\geq 0)$
$\sum_{i = 1}^{x} t_{i} + 1$
%
%
$v$
$T_{0}$-set
%
$(\sum_{i = 1}^{x} t_{i} + 1)/3$
%

%
\noindent
{\bf (2-3) Case~3.4
%
$v$
$x-1 (\geq 0)$
$v$
({\tt CntT}$_{1}$
%
%
%
$T_{3}$-set
$T_{0}$-set
$1 + t_{1} - 1 + \sum_{i = 2}^{x} t_{i} = \sum_{i = 1}^{x} t_{i}$
%
$(t_{1} + 1)/3 + \sum_{i = 2}^{x} t_{i}/3$
$+1$
%
%
%
$T_{1}$
%
$\sum_{i = 0}^{x} t_{i}$
%
$(t_{1} + 1)/3 + \sum_{i = 2}^{x} t_{i}/3 + (t_{0}-1)/3$
%

%
\noindent
{\bf (2-4) Case~3.1.2
$v$
1
%
%
$T_{0}$-set
$\sum_{i = 1}^{x} t_{i} + 1$
%
$\sum_{i = 1}^{x} t_{i}/3 + 1$
%
%
$T_{1}$-set
%
$\sum_{i = 0}^{x} t_{i} + 1$
%
$\sum_{i = 1}^{x} t_{i}/3 + 1 + (t_{0}-1)/3$
%

%
\noindent
{\bf (2-5) Case~3.3.2
(1)
%
%
(1)
$z +1 = y + 1$
%
(2-4)
%
%
$u_{2}$
%
$T_{1}$-set
%
%

%
%
%
Case~3.1.2, 3.2.2, 3.3.1, 3.3.2, 3.4
Case~3.1.2, 3.2.2, 3.3.1, 3.4
$y$
Case~3.3.2
$y+1$
%
%
%

%
\noindent
{\bf (2-1) Case~3.2.2
$C_{ON}(\sigma) / C_{OFF}(\sigma) = 3$
%

%
\noindent
{\bf (2-2) Case~3.3.1
$C_{ON}(\sigma) / C_{OFF}(\sigma) = 3$
%

%
\noindent
{\bf (2-3) Case~3.4
$C_{ON}(\sigma) / C_{OFF}(\sigma) = 3 - \epsilon$
%
$\epsilon > 0$%
0
\noindent
{\bf (2-4) Case~3.1.2
$C_{ON}(\sigma) / C_{OFF}(\sigma) = 3 - \epsilon'$
%
$\epsilon' > 0$%
0
\noindent
{\bf (2-5) Case~3.3.2
$C_{ON}(\sigma) / C_{OFF}(\sigma) = 3 - \epsilon''$
%
$\epsilon'' > 0$%
0
%

%
%
%

%
\fi
\ifnum \count11 > 0
%
%
$ON$ must select $v_{1}$. 
Thus, 
when {\sc SubtreeRoutine} is called with ${\tt BaseVtx} = v_{1}$ in Step~2, 
this subroutine executes Case~T3.2 for the first execution of Step~T3. 
Now we analyze the routine. 
In Case~3.3.2, 
a new base vertex is set and the value of {\tt Count} increases. 
Let $y$ be the value of {\tt Count} when the routine finishes. 
Then, 
if the routine finishes in Case~3.1.2, 3.2.2, 3.3.1, or 3.4, 
then $y < ${\tt MaxT}$_{1}$ and 
there exist $y$ base vertices. 
Otherwise, that is, if it finishes in Case~3.3.2, 
then $y = ${\tt MaxT}$_{1}$ and 
there exist $y+1$ base vertices. 
Let us evaluate the costs for the $z$-th base vertex $v$ and $T$-sets of $v$. 
\noindent
{\bf\boldmath (1):}
First, 
suppose that $z \in [1, y-1]$. 
Vertices composing $T$-sets of $v$ are as follows: 
at most two vertices in a $T$-set ${g}(v, a)$, which is a $T_{2}$-set, 
constructed immediately before the execution of Case~3.3.2, 
and 
all the vertices in $T$-sets constructed at the executions of Cases 3.1.1 and 3.2.1 
before the execution of Case~3.3.2. 
Cases 3.1.1 and 3.2.1 are executed at least zero times and once, respectively. 
In Case~3.3.2, 
the routine sets a vertex $u_{2}$ in ${g}(v, a)$ to be a new base vertex. 
The original $T$-set ${g}(v, a)$ is divided into two $T$-sets: 
a $T_{1}$-set ${g}(v, a)$ of $u_{1}$ ($u'_{1}$, if any) and 
a $T_{0}$-set ${g}(u_{2}, 1)$ of the rest of the vertices in the original ${g}(v, a)$ 
(see Fig.~\ref{fig:dlow}).
$ON$ must select the first revealed vertex $v_{1}$ and 
also selects $u_{2}$ by the definition of $T_{2}$-sets. 
Hence, 
$ON$ selects all base vertices. 
By organizing the above argument, 
$T$-sets of $v$ are as follows: 
one $T_{0}$-set is constructed every time Case~3.1.1 is executed, 
one $T_{1}$-set is constructed when Case~3.2.1 is executed, 
and 
one $T_{1}$-set is constructed when Case~3.3.2 is executed. 
That is, 
$v$ has two $T_{1}$-sets and $x-2 (\geq 0)$ $T_{0}$-sets. 
Thus, 
by the argument about the costs of $T$-sets, 
which were discussed before the routine, 
the total cost for these $T$-sets and $v$ is as follows: 
Suppose that 
the lengths of the two $T_{1}$-sets are $t_{1}$ and $t_{2}$, respectively and 
those of the $x-2$ $T_{0}$-sets are $t_{3}, t_{4}, \ldots, t_{x}$, respectively. 
Then, 
the cost of $ON$ is $\sum_{i = 1}^{x} t_{i} + 1$, 
in which the second term follows from the base vertex $v$. 
On the other hand, 
the total cost of $OFF$ for the two $T_{1}$-sets and $v$ is $(t_{1}+2)/3 + (t_{2}-1)/3$
($v$ is selected only once and the former term is greater than the latter by one).
The cost for $x$ $T_{0}$-sets is $\sum_{i = 3}^{x} t_{i}/3$. 
Hence, 
the total cost of $OFF$ for these $T$-sets is $(\sum_{i = 1}^{x} t_{i} + 1)/3$. 
\noindent
{\bf\boldmath (2):} 
Next, 
suppose that $z = y$. 
We consider the cases finishing the routine: Cases 3.1.2, 3.2.2, 3.3.1, 3.3.2, and 3.4. 
\noindent
{\bf (2-1) Case~3.2.2:} 
{\tt CntT$_{1}$} used in the routine denotes the number of $T_{1}$-sets of $v$, 
and thus 
$v$ has two $T_{1}$-sets 
when the routine finished in Case~3.2.2. 
The routine can execute Case~3.1.1 several times before the finish. 
Moreover, 
as mentioned in (1), 
if $z \geq 2$, 
the routine executes Case~3.2.2 when $z = y-1$, and 
one $T_{0}$-set of $v$, which is called $u_{2}$ in Case~3.2.2, 
is constructed. 
By the above argument, 
$v$ has two $T_{1}$-sets and 
can have several $T_{0}$-sets. 
Therefore, 
the cost for theses $T$-sets and $v$ is obtained in the same way as that in (1). 
\noindent
{\bf (2-2) Case~3.3.1:} 
$v$ has one $T_{2}$-set and several $T_{0}$-sets similarly to (2-1). 
Suppose that the length of the $T_{2}$-set is $t_{1}$ and 
those of the $x-1 (\geq 0)$ $T_{0}$-sets are $t_{2}, \ldots, t_{x}$, respectively. 
Then, 
the cost of $ON$ is $\sum_{i = 1}^{x} t_{i} + 1$. 
On the other hand, 
the cost of $OFF$ for $v$ and the $T_{2}$-set is $(t_{1}+1)/3$, 
and 
that for the $x-1$ $T_{0}$-sets is $\sum_{i = 3}^{x} t_{i}/3$. 
Hence, 
the cost of $OFF$ is $(\sum_{i = 1}^{x} t_{i} + 1)/3$. 
\noindent
{\bf (2-3) Case~3.4:} 
$v$ has two $T_{3}$-sets and $x-1 (\geq 0)$ $T_{0}$-sets. 
In addition, 
$v$ has one $T_{1}$-set 
when {\tt CntT}$_{1}$ $= 1$ at the finish of the routine
({\tt CntT}$_{1}$ is set to be zero or one).
First, 
we consider the case in which there does not exist such $T_{1}$-set, 
that is, 
{\tt CntT}$_{1}$ $= 0$ at the finish of the routine. 
Suppose that the length of the $T_{3}$-set is $t_{1} (= {\tt MaxLength})$ and 
those of the $x-1$ $T_{0}$-sets are $t_{2}, \ldots, t_{x}$, respectively. 
Then, 
the cost of $ON$ is $1 + t_{1} - 1 + \sum_{i = 2}^{x} t_{i} = \sum_{i = 1}^{x} t_{i}$. 
The cost of $OFF$ is at most $(t_{1} + 1)/3 + \sum_{i = 2}^{x} t_{i}/3$, 
in which $+1$ follows from $v$. 
Next, 
we consider the case in which there exists one $T_{1}$-set. 
Suppose that the length of the $T_{1}$-set is $t_{0}$. 
Then, 
the costs of $ON$ and $OFF$ are $\sum_{i = 0}^{x} t_{i}$ and 
at most $(t_{1} + 1)/3 + \sum_{i = 2}^{x} t_{i}/3 + (t_{0}-1)/3$, respectively. 
\noindent
{\bf (2-4) Case~3.1.2:} 
$v$ has {\tt MaxT}$_{0}$ $T_{0}$-sets. 
Also, 
$v$ has one $T_{1}$-set 
if {\tt CntT}$_{1}$ $= 1$ at the finish of the routine. 
We consider the case in which {\tt CntT}$_{1}$ $= 0$ at the finish. 
Suppose that $x = \mbox{{\tt MaxT}}_{0}$ and the lengths of the $x$ $T_{0}$-sets are $t_{1}, \ldots, t_{x}$, respectively. 
Then, 
the costs of $ON$ and $OFF$ are $\sum_{i = 1}^{x} t_{i} + 1$ and at most $\sum_{i = 1}^{x} t_{i}/3 + 1$, respectively. 
We next consider the case in which {\tt CntT}$_{1}$ $= 1$ at the finish. 
Suppose that the length of the $T_{1}$-set is $t_{0}$. 
Then, 
the costs of $ON$ and $OFF$ are $\sum_{i = 0}^{x} t_{i} + 1$ and at most $\sum_{i = 1}^{x} t_{i}/3 + 1 + (t_{0}-1)/3$, respectively. 
\noindent
{\bf (2-5) Case~3.3.2:} 
As mentioned in (1), 
$v$ has two $T_{1}$-sets and can have several $T_{0}$-sets. 
The costs of $ON$ and $OFF$ for this case are obtained in the same way as those in (1) and (2-1). 
Moreover, 
we must consider 
the $z + 1 = y + 1$-st base vertex $u_{2}$ and its $T$-sets. 
$u_{2}$ has one $T_{0}$-set. 
Then, 
the costs of $ON$ and $OFF$ are obtained in the same way as those in (2-4). 
However, 
$u_{2}$ has just one $T_{0}$-set. 
The cost ratio of $ON$ to $OFF$ for $u_{2}$ and its $T$-sets can be quite smaller than three.
However, 
note that there exist {\tt MaxT}$_{1}$ base vertices with two $T_{1}$-sets, and 
{\tt MaxT}$_{1}$ is a sufficiently large integer. 
Now, we are ready to evaluate the competitive ratio of $ON$. 
As mentioned above, 
the routine executes Case~3.3.2 $y-1$ times and 
finishes in Case~3.1.2, 3.2.2, 3.3.1, 3.3.2 or 3.4. 
If it finishes in Case~3.1.2, 3.2.2, 3.3.1 or 3.4, 
then there exist $y$ base vertices. 
Otherwise, that is, 
if it finished in Case~3.3.2, 
then there exist $y+1$ base vertices. 
By the above argument, 
we obtain the following results. 
\noindent
{\bf (2-1) Case~3.2.2:} 
$C_{ON}(\sigma) / C_{OFF}(\sigma) = 3$. 
\noindent
{\bf (2-2) Case~3.3.1:} 
$C_{ON}(\sigma) / C_{OFF}(\sigma) = 3$. 
\noindent
{\bf (2-3) Case~3.4:} 
$C_{ON}(\sigma) / C_{OFF}(\sigma) = 3 - \epsilon$, 
in which 
$\epsilon > 0$ approaches zero 
as the value {\tt MaxLength} in the subroutine becomes sufficiently large. 
\noindent
{\bf (2-4) Case~3.1.2:} 
$C_{ON}(\sigma) / C_{OFF}(\sigma) = 3 - \epsilon'$, 
in which 
$\epsilon' > 0$ approaches zero as the value {\tt MaxT}$_{0}$ in the routine becomes sufficiently large. 
\noindent
{\bf (2-5) Case~3.3.2:} 
$C_{ON}(\sigma) / C_{OFF}(\sigma) = 3 - \epsilon''$, 
in which 
$\epsilon'' > 0$ approaches zero 
as the value {\tt MaxT}$_{1}$ in the routine becomes sufficiently large. 
We have shown that the statement of the theorem is true. 
\fi
%

\section{Randomized Upper Bound} \label{sec:rand_up}
\ifnum \count10 > 0
%
%

%

%
\fi
\ifnum \count11 > 0
%
%

%
\fi
%
\subsection{Algorithm}\label{sec:algorithm}
\ifnum \count10 > 0
%
%
%
$RA$
$1/2$
%
%
$p(v)$
%
%
$A$
$A$
$B$
%
${deg}(v) \leq 2$
$v \notin D_{A} \cap D_{B}$
%

%
$A$($B$)
%
%
%

%
\fi
\ifnum \count11 > 0
%
%
First, we define our algorithm $RA$. 
Before the first vertex is revealed, 
$RA$ chooses to start running one of two deterministic online algorithms $A$ and $B$, 
which are defined later, with the probability of $1/2$ and thereafter keeps running it up to the end of the input. 
For a vertex $v$, 
$p(v)$ denotes the length of the simple path from $v_{1}$ to $v$. 
Roughly speaking, 
the difference between $A$ and $B$ is that for a vertex $v$, 
$A$ selects $v$ if $p(v)$ is odd, and 
$B$ selects $v$ if $p(v)$ is even.
Then, 
$A$ and $B$ try to establish the property that for any vertex $u$ of degree at most two, 
$u \notin D_{A} \cap D_{B}$. 
$A$ ($B$) can select a vertex which $A$ ($B$) selected previously in the following definition. 
It means that $A$ ($B$) does nothing at that time. 
First, 
we give the definition of $A$ as follows. 
\fi
\ifnum \count10 > 0
%
%
\noindent\vspace{-1mm}\rule{\textwidth}{0.5mm} 
\vspace{-3mm}
{\bf Algorithm $A$}\\
\rule{\textwidth}{0.1mm}
	$i$
	{\bf\boldmath Case~1 ($i = 1$
		%
		$v_1$
	{\bf\boldmath Case~2 ($i \geq 2$
		$v_{i}$
\hspace*{2mm}
		{\bf\boldmath Case~2.1 (${deg}_{v_{i}}(u) \geq 3$
			$u$
\hspace*{2mm}
		{\bf\boldmath Case~2.2 (${deg}_{v_{i}}(u) \leq 2$
\hspace*{4mm}
			{\bf\boldmath Case~2.2.1 ($p(v_{i})$ modulo $2 = 0$
				$u$
\hspace*{4mm}
			{\bf\boldmath Case~2.2.2 ($p(v_{i})$ modulo $2 = 1$
				$v_{i}$
\noindent\vspace{-1mm}\rule{\textwidth}{0.5mm} 
\fi
\ifnum \count11 > 0
%
%
\noindent\vspace{-1mm}\rule{\textwidth}{0.5mm} 
\vspace{-3mm}
{\bf Algorithm $A$}\\
\rule{\textwidth}{0.1mm}
	Suppose that the $i$-th vertex $v_{i}$ is revealed. \\
	{\bf\boldmath Case~1 ($i = 1$):}
		%
		Select $v_1$. \\
	{\bf\boldmath Case~2 ($i \geq 2$):}
		Suppose that $v_{i}$ arrives at a vertex $u$. \\
\hspace*{2mm}
		{\bf\boldmath Case~2.1 (${deg}_{v_{i}}(u) \geq 3$):}
			Select $u$. \\
\hspace*{2mm}
		{\bf\boldmath Case~2.2 (${deg}_{v_{i}}(u) \leq 2$):}\\
\hspace*{4mm}
			{\bf\boldmath Case~2.2.1 ($p(v_{i})$ modulo $2 = 0$):}
				Select $u$. \\
\hspace*{4mm}
			{\bf\boldmath Case~2.2.2 ($p(v_{i})$ modulo $2 = 1$):}
				Select $v_{i}$. \\
\noindent\vspace{-1mm}\rule{\textwidth}{0.5mm} 
\fi
\ifnum \count10 > 0
%
%
$A$
$A$
$B$
$B$
%
%
%

%
\fi
\ifnum \count11 > 0
%
%
Since $A$ selects either a revealed vertex $v_{i}$ or the vertex adjacent to $v_{i}$, 
the set of vertices selected by $A$ is a dominating set of a revealed graph 
immediately after each of $A$'s selections. 
The definition of $B$ is quite the same as that of $A$ except for Case~2.2. 
The process of $B$ in Case~2.2.1 (2.2.2) is the same as that of $A$ in Case~2.2.2 (2.2.1). 
Thus, 
the set of vertices selected by $B$ is also a dominating set at any time. 
\fi
\ifnum \count10 > 0
%
%
\noindent\vspace{-1mm}\rule{\textwidth}{0.5mm} 
\vspace{-3mm}
{\bf Algorithm $B$}\\
\rule{\textwidth}{0.1mm}
	$i$
	{\bf\boldmath Case~1 ($i = 1$
		%
		$v_1$
	{\bf\boldmath Case~2 ($i \geq 2$
		$v_i$
\hspace*{2mm}
		{\bf\boldmath Case~2.1 (${deg}_{v_i}(u) \geq 3$
			$u$
\hspace*{2mm}
		{\bf\boldmath Case~2.2 (${deg}_{v_i}(u) \leq 2$
\hspace*{4mm}
			{\bf\boldmath Case~2.2.1 ($p(v_{i})$ modulo $2 = 0$
				$v_{i}$
\hspace*{4mm}
			{\bf\boldmath Case~2.2.2 ($p(v_{i})$ modulo $2 = 1$
				$u$
\noindent\vspace{-1mm}\rule{\textwidth}{0.5mm} 
\fi
\ifnum \count11 > 0
%
%
\noindent\vspace{-1mm}\rule{\textwidth}{0.5mm} 
\vspace{-3mm}
{\bf Algorithm $B$}\\
\rule{\textwidth}{0.1mm}
	Suppose that the $i$-th vertex $v_{i}$ is revealed. \\
	{\bf\boldmath Case~1 ($i = 1$):}
		%
		Select $v_1$. \\
	{\bf\boldmath Case~2 ($i \geq 2$):}
		Suppose that $v_{i}$ arrives at a vertex $u$. \\
\hspace*{2mm}
		{\bf\boldmath Case~2.1 (${deg}_{v_i}(u) \geq 3$):}
			Select $u$. \\
\hspace*{2mm}
		{\bf\boldmath Case~2.2 (${deg}_{v_i}(u) \leq 2$):}\\
\hspace*{4mm}
			{\bf\boldmath Case~2.2.1 ($p(v_{i})$ modulo $2 = 0$):}
				Select $v_{i}$. \\
\hspace*{4mm}
			{\bf\boldmath Case~2.2.2 ($p(v_{i})$ modulo $2 = 1$):}
				Select $u$. \\
\noindent\vspace{-1mm}\rule{\textwidth}{0.5mm} 
\fi
%

\subsection{Basic Properties of $RA$}\label{sec:properties}
\ifnum \count10 > 0
%
%
$A$
\fi
\ifnum \count11 > 0
%
%
In this section, 
we show several basic properties of dominating sets by $A$ and $B$. 
\fi
\begin{LMA} \label{LMA:rand_up.ppab}
	\ifnum \count10 > 0
	%
	%
	%
%
	%
	\hspace*{1mm}
	(1) 
	$v = v_{1}$
	$v \in D_{A}$
	%
%
	%
	%
%
	%
	\hspace*{1mm}
	(2) 
	${deg}(v) \geq 3$
	$v \in D_{A}$
	\hspace*{1mm}
	(3)
	${deg}(v) = 2$
	\hspace*{2mm}
	(3-e)
	$p(v)$ modulo $2 = 0$
	$v \notin D_{A}$
	\hspace*{2mm}
	(3-o)
	$p(v)$ modulo $2 = 1$
	$v \in D_{A}$
	\hspace*{1mm}
	(4)
	${deg}(v) = 1$
	$v$
	\hspace*{2mm}
	(4-1)
	${deg}(\tilde{u}) \geq 3$
	${deg}_{v}(\tilde{u}) \leq 2$
	\hspace*{3mm}
	(4-1-e)	
	$p(v)$ modulo $2 = 0$
	$v \notin D_{A}$
	\hspace*{3mm}
	(4-1-o)	
	$p(v)$ modulo $2 = 1$
	$v \in D_{A}$
	\hspace*{2mm}
	(4-2)
	${deg}(\tilde{u}) \geq 3$
	${deg}_{v}(\tilde{u}) \geq 3$
	$v \notin D_{A}$
	\hspace*{2mm}
	(4-3)
	${deg}(\tilde{u}) \leq 2$
	\hspace*{3mm}
	(4-3-e)	
	$p(v)$ modulo $2 = 0$
	$v \notin D_{A}$
	\hspace*{3mm}
	(4-3-o)	
	$p(v)$ modulo $2 = 1$
	$v \in D_{A}$
	\fi
	\ifnum \count11 > 0
	%
	%
	The following properties hold for a vertex $v$: 
	\\
	\hspace*{1mm}
	(1) 
	If $v = v_{1}$, 
	$v \in D_{A}$ and $v \in D_{B}$. \\
	Suppose that $v \ne v_{1}$. \\
	\hspace*{1mm}
	(2) 
	If ${deg}(v) \geq 3$, 
	$v \in D_{A}$ and $v \in D_{B}$. \\
	\hspace*{1mm}
	(3)
	Suppose that ${deg}(v) = 2$. \\
	\hspace*{2mm}
	(3-e)
	If $p(v)$ modulo $2 = 0$, 
	$v \notin D_{A}$ and $v \in D_{B}$. \\
	\hspace*{2mm}
	(3-o)
	If $p(v)$ modulo $2 = 1$, 
	$v \in D_{A}$ and $v \notin D_{B}$. \\
	\hspace*{1mm}
	(4)
	Suppose that ${deg}(v) = 1$ and 
	let $\tilde{u}$ be the vertex adjacent to $v$. \\
	\hspace*{2mm}
	(4-1)
	Suppose that ${deg}(\tilde{u}) \geq 3$ and 
	${deg}_{v}(\tilde{u}) \leq 2$. \\
	\hspace*{3mm}
	(4-1-e)	
	If $p(v)$ modulo $2 = 0$, 
	$v \notin D_{A}$ and $v \in D_{B}$. \\
	\hspace*{3mm}
	(4-1-o)	
	If $p(v)$ modulo $2 = 1$, 
	$v \in D_{A}$ and $v \notin D_{B}$. \\
	\hspace*{2mm}
	(4-2)
	If ${deg}(\tilde{u}) \geq 3$ and 
	${deg}_{v}(\tilde{u}) \geq 3$, 
	then $v \notin D_{A}$ and $v \notin D_{B}$. \\
	\hspace*{2mm}
	(4-3)
	Suppose that ${deg}(\tilde{u}) \leq 2$. \\
	\hspace*{3mm}
	(4-3-e)	
	If $p(v)$ modulo $2 = 0$, 
	$v \notin D_{A}$ and $v \in D_{B}$. \\
	\hspace*{3mm}
	(4-3-o)	
	If $p(v)$ modulo $2 = 1$, 
	$v \in D_{A}$ and $v \notin D_{B}$. 
	\fi
\end{LMA}
\begin{proof}
	\ifnum \count10 > 0
	%
	%
	%
	%
	$i$
	$v_{i}$
	%
	
	%
	\noindent
	\hspace*{1mm}
	(1) 
	$j = 1$
	$v_{1} \in D_{A}(v_{i})$
	%
%
	%
	%
%
	%
	\hspace*{1mm}
	(2) 
	${deg}_{v_{i}}(v_{j}) \geq 3$
	$v_{j} \in D_{A}(v_{i})$
	\hspace*{1mm}
	(3)
	${deg}_{v_{i}}(v_{j}) = 2$
	\hspace*{2mm}
	(3-e)
	$p(v_{j})$ modulo $2 = 0$
	$v_{j} \notin D_{A}(v_{i})$
	\hspace*{2mm}
	(3-o)
	$p(v_{j})$ modulo $2 = 1$
	$v_{j} \in D_{A}(v_{i})$
	\hspace*{1mm}
	(4)
	${deg}_{v_{i}}(v_{j}) = 1$
	$v_{j}$
	\hspace*{2mm}
	(4-1)
	${deg}_{v_{i}}(\hat{u}) \geq 3$
	${deg}_{v_{j}}(\hat{u}) \leq 2$
	\hspace*{3mm}
	(4-1-e)	
	$p(v_{j})$ modulo $2 = 0$
	$v_{j} \notin D_{A}(v_{i})$
	\hspace*{3mm}
	(4-1-o)	
	$p(v_{j})$ modulo $2 = 1$
	$v_{j} \in D_{A}(v_{i})$
	\hspace*{2mm}
	(4-2)
	${deg}_{v_{i}}(\hat{u}) \geq 3$
	${deg}_{v_{j}}(\hat{u}) \geq 3$
	$v_{j} \notin D_{A}(v_{i})$
	\hspace*{2mm}
	(4-3)
	${deg}_{v_{i}}(\hat{u}) \leq 2$
	\hspace*{3mm}
	(4-3-e)	
	$p(v_{j})$ modulo $2 = 0$
	$v_{j} \notin D_{A}(v_{i})$
	\hspace*{3mm}
	(4-3-o)	
	$p(v_{j})$ modulo $2 = 1$
	$v_{j} \in D_{A}(v_{i})$
	$i = 1$
	$v_{1}$
	$A$
	$v_{1}$
	$v_{1}$
	$i \leq m-1 \hspace*{1mm} (m \geq 2)$
	$v_{m-1}$
	$i = m$
	%
	%
	%
	$v_{m}$
	%
	$u$
	$i = m$
	%
	%
	%
	%
	$u$
	$u'' \ne v_{1}$
	$u''$
	%
	%
	$i = m$
	%
	%
	$u \ne v_{1}$
	%
	${deg}_{v_{m}}(u) \geq 2$
	%
	
	%
	${deg}_{v_{m}}(u) \geq 4$
	$v_{m}$
	$A$
	%
	$u \in D_{A}(v_{m})$
	%
	${deg}_{v_{m}}(u) \geq 3$
	$u$
	%
	$v_{m}$
	$v_{m} \notin D_{A}(v_{m})$
	%
	${deg}_{v_{m}}(v_{m}) = 1$
	$v_{m}$
	%
	%
	$u''$
	$v_{m}$
	(4-)
	$v_{m}$
	$A$
	$v_{m}$
	(4-1)
	${deg}_{v_{m}}(u) = 3$
	$A$
	$u \in D_{A}(v_{m})$
	${deg}_{v_{m}}(u) \geq 4$
	$u$
	$v_{m}$
	%
	%
	$u''$
	$u''$
	%
	${deg}_{u''}(u) \leq 2$
	$p(v_{m})$ modulo $2 = 0$
	$u''$
	$p(u'')$ modulo $2 = 0$
	%
	$u'' \notin D_{A}(v_{m-1})$
	$u'' \in D_{B}(v_{m-1})$
	%
	$v_{m}$
	$A$
	$u'' \notin D_{A}(v_{m})$
	$u'' \in D_{B}(v_{m})$
	%
	${deg}_{v_{m}}(u) = 3$
	${deg}_{v_{m}}(u'') = 1$
	%
	%
	$p(v_{m})$ modulo $2 = 1$
	%
	
	%
	${deg}_{v_{m}}(u) = 2$
	%
	${deg}_{v_{m-1}}(u) = 1$
	$u''$
	$u$
	$u''$
	%
	$u''$
	%
	%
	%
	$u$
	%
	%
	${deg}_{v_{m}}(u') \leq 2$
	$p(v_{m})$ modulo $2 = 0$
	%
	%
	$v_{m}$
	$A$
	%
	%
	$A$
	%
	$v_{m} \notin D_{A}(v_{m})$
	$u \in D_{A}(v_{m})$
	%
	$p(v_{m})$ modulo $2 = 0$
	$p(u)$ modulo $2 = 1$
	$u \notin D_{B}(v_{m})$
	%
	$u \notin D_{B}(v_{m})$
	%
	%
	${deg}_{v_{m}}(v_{m}) = 1$
	$v_{m}$
	%
	${deg}_{v_{m}}(u) = 2$
	$u$
	${deg}_{v_{m}}(u') \leq 2$
	$p(v_{m})$ modulo $2 = 1$
	%
	$A$
	$v_{m}$
	$u$
	%
	
	%
	${deg}_{v_{m}}(u) = 2$
	${deg}_{v_{m}}(u') \geq 3$
	${deg}_{u}(u') \geq 3$
	%
	%
	$p(v_{m})$ modulo $2 = 0$
	$v_{m}$
	$A$
	$A$
	%
	$v_{m} \notin D_{A}(v_{m})$
	$u \in D_{A}(v_{m})$
	%
	%
	${deg}_{u}(u') \geq 3$
	$u \notin D_{B}(v_{m-1})$
	$u \notin D_{B}(v_{m})$
	%
	%
	${deg}_{v_{m}}(v_{m}) = 1$
	$v_{m}$
	%
	${deg}_{v_{m}}(u) = 2$
	$u$
	$p(v_{m})$ modulo $2 = 1$
	%
	
	%
	${deg}_{v_{m}}(u) = 2$
	${deg}_{v_{m}}(u') \geq 3$
	${deg}_{u}(u') \leq 2$
	%
	%
	$v_{m}$
	$A$
	%
	$p(v_{m})$ modulo $2 = 0$
	$u \notin D_{B}(v_{m-1})$
	%
	%
	${deg}_{u}(u') \geq 3$
	$p(v_{m})$ modulo $2 = 0$
	%
	%
	$p(v_{m})$ modulo $2 = 1$
	${deg}_{u}(u') \geq 3$
	$p(v_{m})$ modulo $2 = 1$
	$u = v_{1}$
	$u$
	$u \in D_{A}(v_{m-1})$
	$u \in D_{B}(v_{m-1})$
	%
	$u \in D_{A}(v_{m})$
	$u \in D_{B}(v_{m})$
	(1)
	$v_{m}$
	$u \ne v_{m}$
	$u''$
	${deg}_{v_{m}}(u) = 2$
	$u \ne v_{m}$
	${deg}_{v_{m}}(u) = 2$
	%
	${deg}_{v_{m-1}}(u) = 1$
	$u''$
	$v_{m}$
	(4-3-o)
	$u'' \in D_{A}(v_{m-1})$
	$u'' \notin D_{B}(v_{m-1})$
	%
	$v_{m}$
	$A$
	$u''$
	%
	$u'' \in D_{A}(v_{m})$
	$u'' \notin D_{B}(v_{m})$
	%
	%
	$u''$
	\fi
	\ifnum \count11 > 0
	%
	To show each property in the statement of this lemma, 
	we prove the following properties for each $j \leq i$ by induction on the number $i$ of revealed vertices. 
	%
	Note that 
	if we have proven them, 
	then each property in the statement is satisfied 
	when $i$ is the number of the vertices of the tree in an input, 
	that is, 
	$v_{i}$ is the final revealed vertex. 
	\noindent
	\hspace*{1mm}
	(1) 
	If $j = 1$, 
	$v_{1} \in D_{A}(v_{i})$ and $v_{1} \in D_{B}(v_{i})$. \\
	Suppose that $j \geq 2$. \\
	\hspace*{1mm}
	(2) 
	If ${deg}_{v_{i}}(v_{j}) \geq 3$, 
	$v_{j} \in D_{A}(v_{i})$ and $v_{j} \in D_{B}(v_{i})$. \\
	\hspace*{1mm}
	(3)
	Suppose that ${deg}_{v_{i}}(v_{j}) = 2$. \\
	\hspace*{2mm}
	(3-e)
	If $p(v_{j})$ modulo $2 = 0$, 
	$v_{j} \notin D_{A}(v_{i})$ and $v_{j} \in D_{B}(v_{i})$. \\
	\hspace*{2mm}
	(3-o)
	If $p(v_{j})$ modulo $2 = 1$, 
	$v_{j} \in D_{A}(v_{i})$ and $v_{j} \notin D_{B}(v_{i})$. \\
	\hspace*{1mm}
	(4)
	Suppose that ${deg}_{v_{i}}(v_{j}) = 1$. 
	Let $v_{j}$ be the vertex adjacent to $\hat{u}$. \\
	\hspace*{2mm}
	(4-1)
	Suppose that 
	${deg}_{v_{i}}(\hat{u}) \geq 3$ and 
	${deg}_{v_{j}}(\hat{u}) \leq 2$. \\
	\hspace*{3mm}
	(4-1-e)	
	If $p(v_{j})$ modulo $2 = 0$, 
	$v_{j} \notin D_{A}(v_{i})$ and $v_{j} \in D_{B}(v_{i})$. \\
	\hspace*{3mm}
	(4-1-o)	
	If $p(v_{j})$ modulo $2 = 1$, 
	$v_{j} \in D_{A}(v_{i})$ and $v_{j} \notin D_{B}(v_{i})$. \\
	\hspace*{2mm}
	(4-2)
	If 
	${deg}_{v_{i}}(\hat{u}) \geq 3$ and 
	${deg}_{v_{j}}(\hat{u}) \geq 3$,  
	$v_{j} \notin D_{A}(v_{i})$ and $v_{j} \notin D_{B}(v_{i})$. \\
	\hspace*{2mm}
	(4-3)
	Suppose that ${deg}_{v_{i}}(\hat{u}) \leq 2$. \\
	\hspace*{3mm}
	(4-3-e)	
	If $p(v_{j})$ modulo $2 = 0$, 
	$v_{j} \notin D_{A}(v_{i})$ and $v_{j} \in D_{B}(v_{i})$. \\
	\hspace*{3mm}
	(4-3-o)	
	If $p(v_{j})$ modulo $2 = 1$, 
	$v_{j} \in D_{A}(v_{i})$ and $v_{j} \notin D_{B}(v_{i})$. 
	If $i = 1$, 
	$v_{1}$ is revealed. 
	Both $A$ and $B$ execute Case~1 and select $v_{1}$. 
	Thus, 
	$v_{1}$ satisfies (1). 
	We assume that 
	the above properties are satisfied when $i \leq m-1 \hspace*{1mm} (m \geq 2)$, 
	that is, 
	before $v_{m}$ is revealed, 
	and we show that 
	they are also satisfied when $i = m$, 
	that is, 
	immediately after $v_{m}$ is revealed. 
	Each of the above properties depends on the degrees of revealed vertices. 
	When $v_{m}$ is revealed, 
	only the degree of an adjacent vertex $u$ to $v_{m}$ changes.  
	Hence, 
	all the revealed vertices except for $u$ and vertices adjacent to $u$ also satisfy the above properties 
	when $i = m$ by the induction hypothesis. 
	Then, 
	we in what follows show that 
	$u$ and a vertex adjacent to $u$ including $v_{m}$ satisfy the properties. 
	Note that the degree of $u$ may affect an adjacent vertex $u'' (\ne v_{m})$ at only the case (4) of the above properties, 
	that is, 
	the case in which $u'' \ne v_{1}$ and the degree of $u''$ is one. 
	Hence, 
	$u''$ also satisfies the above properties 
	when $i = m$ in the other cases by the induction hypothesis. 
	%
	%
	First, 
	we consider the case in which $u \ne v_{1}$. 
	Note that 
	${deg}_{v_{m}}(u) \geq 2$.  
	Suppose that ${deg}_{v_{m}}(u) \geq 4$. 
	When $v_{m}$ is revealed, 
	both $A$ and $B$ execute Case~2.1 and select $u$. 
	That is, 
	$u \in D_{A}(v_{m})$ and $u \in D_{B}(v_{m})$. 
	Thus, 
	$u$ satisfies the property (2)
	because ${deg}_{v_{m}}(u) \geq 3$. 
	On the other hand, 
	$A$ and $B$ do nothing for $v_{m}$, 
	and 
	thus 
	$v_{m} \notin D_{A}(v_{m})$ and $v_{m} \notin D_{B}(v_{m})$. 
	Then, 
	${deg}_{v_{m}}(v_{m}) = 1$ and 
	$v_{m}$ satisfies (4-2). 
	Since the degree of $u''$ is one as mentioned above, 
	$u''$ satisfies one of the properties (4-1) and (4-2)
	before $v_{m}$ is revealed by the induction hypothesis. 
	Analogously to $v_{m}$, 
	$A$ and $B$ do nothing $u''$, and 
	$v_{m}$ keeps (4-1) or (4-2). 
	Suppose that ${deg}_{v_{m}}(u) = 3$. 
	In this case, $A$ and $B$ execute Case~2.1, and 
	$u \in D_{A}(v_{m})$ and $u \in D_{B}(v_{m})$. 
	Similarly to the case in which ${deg}_{v_{m}}(u) \geq 4$, 
	$u$ and $v_{m}$ satisfy (2) and (4-2), respectively. 
	Let us consider $u''$ next. 
	Since the degree of $u''$ is one and 
	$u''$ is not $v_{1}$, 
	${deg}_{u''}(u) \leq 2$. 
	Suppose that $p(v_{m})$ modulo $2 = 0$. 
	$u'' \notin D_{A}(v_{m-1})$
	and 
	$u'' \in D_{B}(v_{m-1})$ by (4-3-e) in the induction hypothesis. 
	When $v_{m}$ is revealed, 
	$A$ and $B$ do nothing for $u''$. 
	Thus,  
	$u'' \notin D_{A}(v_{m})$
	and 
	$u'' \in D_{B}(v_{m})$, 
	which means that 
	$u''$ satisfies (4-1-e)
	since ${deg}_{v_{m}}(u) = 3$ and ${deg}_{v_{m}}(u'') = 1$. 
	If $p(v_{m})$ modulo $2 = 1$, 
	we can show that $u''$ satisfies (4-1-o) similarly. 
	Suppose that ${deg}_{v_{m}}(u) = 2$. 
	Thus, 
	${deg}_{v_{m-1}}(u) = 1$. 
	We discuss $u''$ in advance. 
	$u''$ is revealed before $u$ by the condition of this case 
	because $u$ is not $v_{1}$. 
	However, 
	the degree of $u''$ is one, 
	which contradicts that $u''$ is not $v_{1}$. 
	Hence, 
	there does not exist $u''$ in this case. 
	Let $u' (\ne v_{m})$ be a vertex adjacent to $u$. 
	Suppose that ${deg}_{v_{m}}(u') \leq 2$
	and 
	$p(v_{m})$ modulo $2 = 0$. 
	$A$ and $B$ execute Case~2.2.1 by definition 
	when $v_{m}$ is revealed. 
	Hence, 
	$A$ and $B$ select $u$ and $v_{m}$, respectively, 
	which indicates that 
	$v_{m} \notin D_{A}(v_{m})$, 
	$v_{m} \in D_{B}(v_{m})$ and 
	$u \in D_{A}(v_{m})$. 
	Also, 
	$p(u)$ modulo $2 = 1$ because 
	$p(v_{m})$ modulo $2 = 0$. 
	Hence, 
	$u \notin D_{B}(v_{m})$
	by (4-3-o) in the induction hypothesis. 
	Thus, 
	$u \notin D_{B}(v_{m})$. 
	Then, 
	$v_{m}$ satisfies (4-3-e) 
	since ${deg}_{v_{m}}(v_{m}) = 1$. 
	In addition, 
	$u$ satisfies (3-o) 
	since ${deg}_{v_{m}}(u) = 2$. 
	The case in which 
	${deg}_{v_{m}}(u') \leq 2$
	and 
	$p(v_{m})$ modulo $2 = 1$ can be shown similarly.  
	Roughly speaking, 
	$A$ and $B$ execute Case~2.2.2 in this case, 
	and $v_{m}$ and $u$ satisfy (4-3-o) and (3-e), respectively. 
	Suppose that 
	${deg}_{v_{m}}(u) = 2$, 
	${deg}_{v_{m}}(u') \geq 3$ and 
	${deg}_{u}(u') \geq 3$. 
	Also, 
	suppose that $p(v_{m})$ modulo $2 = 0$. 
	Both $A$ and $B$ execute Case~2.2.1 and 
	$A$ and $B$ select $u$ and $v_{m}$, respectively, 
	when $v_{m}$ is revealed. 
	Thus, 
	$v_{m} \notin D_{A}(v_{m})$, 
	$v_{m} \in D_{B}(v_{m})$ and 
	$u \in D_{A}(v_{m})$. 
	Since ${deg}_{u}(u') \geq 3$, 
	$u \notin D_{B}(v_{m-1})$ by (4-2) in the induction hypothesis, 
	which directly implies that 
	$u \notin D_{B}(v_{m})$. 
	Hence, 
	$v_{m}$ satisfies (4-3-e) 
	because ${deg}_{v_{m}}(v_{m}) = 1$. 
	$u$ satisfies (3-o)
	because ${deg}_{v_{m}}(u) = 2$. 
	We can prove the case in which $p(v_{m})$ modulo $2 = 1$, similarly, 
	and hence we omit it. 
	Suppose that ${deg}_{v_{m}}(u) = 2$, 
	${deg}_{v_{m}}(u') \geq 3$ and 
	${deg}_{u}(u') \leq 2$. 
	Then, 
	$A$ and $B$ execute Case~2.2.1. 
	Moreover, 
	$u \notin D_{B}(v_{m-1})$ by (4-1-o) in the induction hypothesis 
	if $p(v_{m})$ modulo $2 = 0$. 
	We omit the proof of this case 
	because the situation in this case is quite the same as that in the case in which 
	${deg}_{u}(u') \geq 3$
	and 
	$p(v_{m})$ modulo $2 = 0$. 
	In addition, 
	the case in which $p(v_{m})$ modulo $2 = 1$ can be proven 
	in the same way as the case in which 
	${deg}_{u}(u') \geq 3$
	and 
	$p(v_{m})$ modulo $2 = 1$. 
	Suppose that $u = v_{1}$. 
	$u \in D_{A}(v_{m-1})$ 
	and 
	$u \in D_{B}(v_{m-1})$ 
	by (1) in the induction hypothesis. 
	Thus, 
	$u \in D_{A}(v_{m})$
	and 
	$u \in D_{B}(v_{m})$, 
	which means that $u$ satisfies (1). 
	We omit the proof with respect to $v_{m}$
	because it can be shown in the same way as the case of $u \ne v_{m}$. 
	Finally, let us discuss $u''$. 
	If ${deg}_{v_{m}}(u) \ne 2$,	
	we can show the proof in the same way as that in the case in which $u \ne v_{m}$. 
	Then, we consider the case in which ${deg}_{v_{m}}(u) = 2$. 
	That is, 
	${deg}_{v_{m-1}}(u) = 1$. 
	Since the degree of $u''$ is one, 
	both 
	$u'' \in D_{A}(v_{m-1})$
	and 
	$u'' \notin D_{B}(v_{m-1})$
	by (4-3-o) in the induction hypothesis. 
	Also, 
	$A$ and $B$ execute Case~2.2.2 when $v_{m}$ is revealed, 
	and they do not select $u''$. 
	Thus, 
	$u'' \in D_{A}(v_{m})$
	and 
	$u'' \notin D_{B}(v_{m})$. 
	Therefore, 
	$u''$ satisfies (4-3-o)
	immediately after the revelation of $v_{m}$. 
	\fi
\end{proof}
\ifnum \count10 > 0
%
%

%

%
\fi
\ifnum \count11 > 0
%
%

%
\fi
%
\begin{LMA} \label{LMA:rand_up.vtx_cost}
	\ifnum \count10 > 0
	%
	%
	$RA$
	\hspace*{1mm}
	(1) 
	$v = v_{1}$
	%
%
	%
	%
%
	%
	\hspace*{1mm}
	(2) 
	${deg}(v) \geq 3$
	$1$
	\hspace*{1mm}
	(3)
	${deg}(v) = 2$
	$1/2$
	\hspace*{1mm}
	(4)
	${deg}(v) = 1$
	$v$
	\hspace*{2mm}
	(4-1)
	${deg}(u) \geq 3$
	${deg}_{v}(u) \leq 2$
	$1/2$
	\hspace*{2mm}
	(4-2)
	${deg}(u) \geq 3$
	${deg}_{v}(u) \geq 3$
	$0$
	\hspace*{2mm}
	(4-3)
	${deg}(u) \leq 2$
	$1/2$
	\fi
	\ifnum \count11 > 0
	%
	%
	The expected cost of $RA$ for $v$ is as follows: \\
	\hspace*{1mm}
	(1) 
	If $v = v_{1}$, it is one. \\
	Suppose that $v \ne v_{1}$. \\
	\hspace*{1mm}
	(2) 
	If ${deg}(v) \geq 3$, 
	it is one. \\
	\hspace*{1mm}
	(3)
	If ${deg}(v) = 2$, 
	it is $1/2$. \\
	\hspace*{1mm}
	(4)
	Suppose that ${deg}(v) = 1$ and 
	$v$ is adjacent to a vertex $u$. \\
	\hspace*{2mm}
	(4-1)
	If ${deg}(u) \geq 3$ and 
	${deg}_{v}(u) \leq 2$, 
	it is $1/2$. \\
	\hspace*{2mm}
	(4-2)
	If ${deg}(u) \geq 3$ and 
	${deg}_{v}(u) \geq 3$, 
	it is zero. \\
	\hspace*{2mm}
	(4-3)
	If ${deg}(u) \leq 2$, 
	it is $1/2$. \\
	\fi
\end{LMA}
\begin{proof}
	\ifnum \count10 > 0
	%
	%
	$RA$
	%
	$v \notin D_{A}$
	$v \in D_{A}$
	$v$
	%
	%
	$v \in D_{A}$
	$v$
	%
	$v \notin D_{A}$
	%
	%
	%
	
	%
	\fi
	\ifnum \count11 > 0
	%
	%
	$RA$ chooses $A$ ($B$) with the probability $1/2$ at the beginning. 
	Thus, 
	the expected cost of $RA$ for a vertex $v$ is $1/2$ 
	if either $v \notin D_{A}$ and $v \in D_{B}$ or 
	$v \in D_{A}$ and $v \notin D_{B}$. 
	It is easy to see that  
	the expected cost for $v$ is one 
	if both $v \in D_{A}$ and $v \in D_{B}$. 
	Also, 
	the expected cost for $v$ is zero 
	if both $v \notin D_{A}$ and $v \notin D_{B}$. 
	The above argument with each property in Lemma~\ref{LMA:rand_up.ppab} gives each cost in the statement of this lemma. 
	\fi
\end{proof}
\ifnum \count10 > 0
%
%
$v$
%
$v$
%
%
$v$
%
%
$u$
%
$u$
$v$
%
%
%
{\em 
(i) $v \ne v_{1}$, 
(ii) ${deg}(u) \geq 3$, 
(iii)
${deg}(v) = 3$%
${deg}(v) = 2$
${deg}(v') \geq 3$
${deg}_{v}(v') \geq 3$
%
%
3
3
(i) $u_{1}$
(ii) ${deg}(u_{1}) = {deg}(u_{2}) = {deg}(u_{3}) = 3$
(iii) $OPT$
%

%
%

%
%
\ifnum \count11 > 0
%
%
We say that a vertex $v$ {\em dominates} vertices adjacent to $v$ if $OPT$ selects $v$.  
We also say that $v$ {\em dominates} $v$ itself. 
If a vertex $u$ arrives at a vertex $v$, 
$(v, u)$ denotes the edge between $v$ and $u$. 
Suppose that a vertex $u$ arrives at a vertex $v$. 
Also, 
suppose that $u$ is dominated by a vertex in $U$ and 
$v$ is dominated by a vertex not in $U$, 
in which $U$ is the set of $u$ and all the descendants of  $u$. 
%
Then, 
we say that the edge $(v, u)$ is {\em free}.
We say that a free edge $(v, u)$ is {\em fixed}
if this edge satisfies the following three conditions: 
(i) $v \ne v_{1}$, 
(ii) ${deg}(u) \geq 3$, and 
(iii)
either 
${deg}(v) = 3$ or 
${deg}(v) = 2$,  
${deg}(v') \geq 3$ and 
${deg}_{v}(v') \geq 3$, 
in which 
$v'(\ne u)$ is the vertex at which $v$ arrives. 
%
%
We say that a vertex triplet $( u_{1}, u_{2}, u_{3} )$ is {\em good} 
if the vertices $u_{1}, u_{2}$ and $u_{3}$ satisfy the following three conditions: 
(i) both $u_{1}$ and $u_{3}$ are adjacent to $u_{2}$,
(ii) ${deg}(u_{1}) = {deg}(u_{2}) = {deg}(u_{3}) = 3$, 
and 
(iii) $OPT$ selects $u_{1}$ and $u_{3}$. 
In the rest of this section, 
we will show the following lemma. 
\fi
\begin{LMA} \label{LMA:rand_up.basic_p1_7}
	\ifnum \count10 > 0
	%
	%
	$\frac{{\mathbb E}[C_{RA}(\sigma)]}{C_{OPT}(\sigma)}$
	\begin{description}
	\itemsep=-2.0pt
	\setlength{\leftskip}{10pt}
	\item[(P1)]
	%
	\item[(P2)]
	%
	\item[(P3)]
	$OPT$
	%
	\item[(P4)]
	$OPT$
	%
	\item[(P5)]
	%
	\item[(P6)]
	$v$
	%
	\item[(P7)]
	%
	\end{description}
	\fi
	\ifnum \count11 > 0
	%
	%
	There exists an input $\sigma$ which maximizes $\frac{{\mathbb E}[C_{RA}(\sigma)]}{C_{OPT}(\sigma)}$ and satisfies the following seven properties.  
	\begin{description}
	\itemsep=-2.0pt
	\setlength{\leftskip}{10pt}
	\item[(P1)]
	Any free edge is fixed
	(Lemmas~\ref{LMA:rand_up.p1} and \ref{LMA:rand_up.p1r}). 
	\item[(P2)]
	The degree of any vertex is at most three
	(Lemma~\ref{LMA:rand_up.p2}). 
	\item[(P3)]
	The degree of any vertex selected by $OPT$ is three
	(Lemma~\ref{LMA:rand_up.p3}). 
	\item[(P4)]
	For any free edge $(v, u)$, 
	$OPT$ does not select $v$
	(Lemma~\ref{LMA:rand_up.p4}). 
	\item[(P5)]
	Good vertex triplets are not contained
	(Lemma~\ref{LMA:rand_up.p5}). 
	\item[(P6)]
	For any free edge $(v, u)$, 
	the degree of $v$ is not two
	(Lemma~\ref{LMA:rand_up.p6}). 
	\item[(P7)]
	The degree of any vertex is either one or three
	(Lemma~\ref{LMA:rand_up.degnot2}).
	\end{description}
	\fi
\end{LMA}
%
%
$RA$
(P1)
(P7)
(P2)
%
%
%
%
%
$V_{1} = V \setminus U$
$E_{1} = \{ \{ u, u' \} \in E \mid u, u' \in V_{1} \}$%
$S_{1}$
%
%
$E_{2} = \{ \{ u, u' \} \in E \mid u, u' \in U \}$%
$S_{2}$
%

%
%
%
$V_{3} = V' \cup V''$
$E_{3} = E' \cup E'' \cup \{ \{ v', u''_{1} \} \}$%
$n''$
$u''_{i}$
$S_{3} = (u'_{1}, \ldots, u'_{n'}, u''_{1}, \ldots, u''_{n''})$
$n'$
$u'_{i}$
\fi
\ifnum \count11 > 0
%
%
This lemma shows that 
we only have to consider an input satisfying the properties from (P1) to (P7) 
to evaluate the competitive ratio of $RA$. 
It is easy to see that 
if (P7) holds, both (P2) and (P6) clearly hold. 
However, 
we must prove some lemmas including ones about the both properties before showing (P7).
To prove the above lemma and the following lemmas, 
we give a few definitions about transformations of an input. 
First, 
we ``divide'' an input into two inputs.  
For an input $\sigma = (V, E, S)$ and a vertex $v \in V$, 
we define the input $f_{1}(\sigma, v) = (V_{1}, E_{1}, S_{1})$ such that 
$V_{1} = V \setminus U$,  
in which 
$U$ is the set of $v$ and all the descendants of $v$, 
$E_{1} = \{ \{ u, u' \} \in E \mid u, u' \in V_{1} \}$. 
That is, 
$(V_{1}, E_{1})$ is the subgraph of $(V, E)$ induced by $V_{1}$, and 
$S_{1}$ is the subsequence of $S$ consisting of all the vertices of $V_{1}$
(Fig.~\ref{fig:fig_inputfix}). 
Also, 
we define the input $f_{2}(\sigma, v) = (U, E_{2}, S_{2})$ such that 
$E_{2} = \{ \{ u, u' \} \in E \mid u, u' \in U \}$, 
that is, 
$(U, E_{2})$ is the subgraph of  $(V, E)$ induced by $U$, and 
$S_{2}$ is the subsequence of $S$  consisting of all the vertices of $U$. 
Moreover, 
we ``connect'' two inputs.  
For an input $\sigma' = (V', E', S')$, a vertex $v' \in V'$ and 
an input $\sigma'' = (V'', E'', S'')$, 
we define $f_{3}(\sigma', v', \sigma'') = (V_{3}, E_{3}, S_{3})$ such that 
$V_{3} = V' \cup V''$, 
$E_{3} = E' \cup E'' \cup \{ \{ v', u''_{1} \} \}$, 
in which 
$u''_{i}$ is the $i (\in [1, n''])$-th vertex in $S''$ and 
$n''$ is the number of vertices in $S''$, 
and 
$S_{3} = (u'_{1}, \ldots, u'_{n'}, u''_{1}, \ldots, u''_{n''})$, 
in which 
$n'$ is the number of vertices in $S'$ and 
$u'_{i}$ is the $i (\in [1, n'])$-th vertex in $S'$. 
\fi
\ifnum \count12 > 0
\begin{figure*}[ht]
	 \begin{center}
	  \includegraphics[width=130mm]{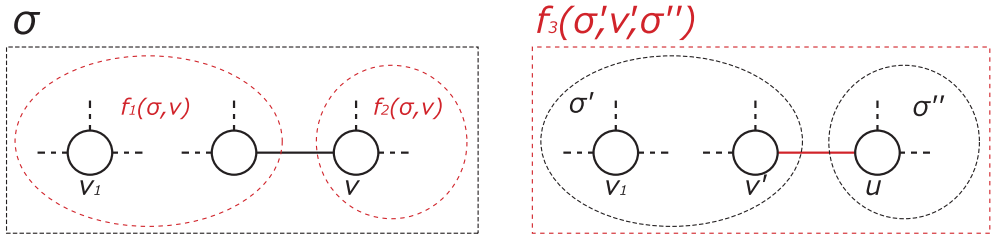}
	 \end{center}
	 \caption{
\ifnum \count10 > 0
%
%
$v$
%
$u$
\fi
\ifnum \count11 > 0
%
%
``Dividing'' and ``connecting'' inputs. 
The left figure shows that 
we construct the two inputs $f_{1}(\sigma, v)$ and $f_{2}(\sigma, v)$ in terms of a vertex $v$ from an input $\sigma$. 
The right figure shows that 
we obtain the input $f_{3}(\sigma', v', \sigma'')$ by adding an input $\sigma''$ to an input $\sigma'$ in terms of a vertex $v'$ in $\sigma'$. 
$u$ was the first revealed vertex $v_{1}$ in $\sigma''$. 
\fi
			}
	\label{fig:fig_inputfix}
\end{figure*}
\fi

\begin{LMA} \label{LMA:rand_up.input_div_opt}
	\ifnum \count10 > 0
	%
	%
	%
	$D_{OPT}(f_{1}(\sigma, u)) \cup D_{OPT}(f_{2}(\sigma, u)) = D_{OPT}(\sigma)$
	%
	
	%
	\fi
	\ifnum \count11 > 0
	%
	%
	Suppose that 
	the vertex set of an input $\sigma$ contains two vertices $v$ and $u$
	such that the edge $(v, u)$ is free. 
	Then, 
	there exists $OPT$ such that 
	$D_{OPT}(f_{1}(\sigma, u)) \cup D_{OPT}(f_{2}(\sigma, u)) = D_{OPT}(\sigma)$. 
	\fi
\end{LMA}
\begin{proof}
	\ifnum \count10 > 0
	%
	%
	%
	%
	$\sigma$
	%
	$D_{OFF}(f_{1}(\sigma, u)) \cup D_{OFF}(f_{2}(\sigma, u)) = D_{OPT}(\sigma)$
	\begin{equation} \label{LMA:rand_up.input_div_opt:eq.1}
		C_{OFF}(f_{1}(\sigma, u)) + C_{OFF}(f_{2}(\sigma, u)) = C_{OPT}(\sigma)
	\end{equation}
	%
	%
	%
	$OFF$
	\begin{equation} \label{LMA:rand_up.input_div_opt:eq.2}
		C_{OPT'}(f_{1}(\sigma, u)) + C_{OPT'}(f_{2}(\sigma, u)) 
		< C_{OFF}(f_{1}(\sigma, u)) + C_{OFF}(f_{2}(\sigma, u))
	\end{equation}
	%
	%
	$f_{1}(\sigma, u)$
	$D_{OPT'}(f_{1}(\sigma, u)) \cup D_{OPT'}(f_{2}(\sigma, u))$
	%
	$D_{OPT}(\sigma)$
	%
	%
	
	%
	\fi
	\ifnum \count11 > 0
	%
	%
	Suppose that the vertex set of an input $\sigma$ contains two vertices $v$ and $u$ 
	such that the edge $(v, u)$ is free. 
	We consider an offline algorithm $OFF$  for the input $f_{1}(\sigma, u)$ and input $f_{2}(\sigma, u)$ which selects the same vertices 
	as those selected by $OPT$ for $\sigma$. 
	That is, 
	$D_{OFF}(f_{1}(\sigma, u)) \cup D_{OFF}(f_{2}(\sigma, u)) = D_{OPT}(\sigma)$, and hence, 
	\begin{equation} \label{LMA:rand_up.input_div_opt:eq.1}
		C_{OFF}(f_{1}(\sigma, u)) + C_{OFF}(f_{2}(\sigma, u)) = C_{OPT}(\sigma). 
	\end{equation}
	We prove that $OFF$ is optimal by contradiction. 
	We assume that $OFF$ is not optimal, 
	that is, 
	there exists an offline algorithm $OPT'$ such that 
	\begin{equation} \label{LMA:rand_up.input_div_opt:eq.2}
		C_{OPT'}(f_{1}(\sigma, u)) + C_{OPT'}(f_{2}(\sigma, u)) 
		< C_{OFF}(f_{1}(\sigma, u)) + C_{OFF}(f_{2}(\sigma, u)). 
	\end{equation}
	By the definitions of $f_{1}(\sigma, u)$ and $f_{2}(\sigma, u)$, 
	$D_{OPT'}(f_{1}(\sigma, u)) \cup D_{OPT'}(f_{2}(\sigma, u))$ is a dominating set of the graph of $\sigma$. 
	However, 
	the size of this set is smaller than that of $D_{OPT}(\sigma)$  
	by Eqs.~(\ref{LMA:rand_up.input_div_opt:eq.1}) and (\ref{LMA:rand_up.input_div_opt:eq.2}), 
	which contradicts the optimality of $OPT$. 
	Therefore, 
	$OFF$ is optimal. 
	\fi
\end{proof}
\ifnum \count10 > 0
%
%

%

%
\fi
\ifnum \count11 > 0
%
%

%
\fi
\begin{LMA} \label{LMA:rand_up.input_cnct_opt}
	\ifnum \count10 > 0
	%
	%
	$OPT$
	%
	$D_{OPT}(f_{3}(\sigma, v, \hat{\sigma})) = D_{OPT}(\sigma)$
	$\hat{\sigma} = (\{ u \}, \varnothing, u)$
	$u$
	%
	
	%
	\fi
	\ifnum \count11 > 0
	%
	%
	Suppose that 
	the graph in an input $\sigma$ contains a vertex $v$ and $OPT$ selects $v$. 
	Then, 
	there exists $OPT$ such that 
	$D_{OPT}(f_{3}(\sigma, v, \hat{\sigma})) = D_{OPT}(\sigma)$, 
	in which 
	$\hat{\sigma} = (\{ u \}, \varnothing, u)$ and 
	$u$ is a vertex not in the vertex set of the graph in $\sigma$. 
	\fi
\end{LMA}
\begin{proof}
	\ifnum \count10 > 0
	%
	%
	%
	%
	$OPT$
	%
	$\hat{\sigma} = (\{ u \}, \varnothing, u)$
	$u$
	%
	%
	%
	$\sigma$
	%
	%
	$D_{OFF}(f_{3}(\sigma, v, \hat{\sigma})) = D_{OPT}(\sigma)$
	\begin{equation} \label{LMA:rand_up.input_cnct_opt:eq.1}
		C_{OFF}(f_{3}(\sigma, v, \hat{\sigma})) = C_{OPT}(\sigma)
	\end{equation}
	%
	%
	%
	$OFF$
	\begin{equation} \label{LMA:rand_up.input_cnct_opt:eq.2}
		C_{OPT'}(f_{3}(\sigma, v, \hat{\sigma})) < C_{OFF}(f_{3}(\sigma, v, \hat{\sigma}))
	\end{equation}
	%
	
	%
	$f_{3}(\sigma, v, \hat{\sigma})$
	$D_{OPT'}(f_{3}(\sigma, v, \hat{\sigma}))$%
	%
	%
	$D_{OPT}(\sigma)$
	%
	%
	
	%
	$OPT'$
	$OPT'$
	%
	%
	%
	$D_{OPT}(\sigma)$
	%
	%
	
	%
	\fi
	\ifnum \count11 > 0
	%
	%
	We prove this lemma in a similar way to the proof of Lemma~\ref{LMA:rand_up.input_div_opt}. 
	Suppose that the graph in the vertex set of the graph in an input $\sigma$ contains a vertex $v$ and $OPT$ selects $v$. 
	Let $\hat{\sigma} = (\{ u \}, \varnothing, u)$, 
	in which 
	$u$ is a vertex not in the graph in $\sigma$. 
	We consider an offline algorithm $OFF$ for the input $f_{3}(\sigma, v, \hat{\sigma})$ which selects the same vertices 
	as those selected by $OPT$. 
	That is, 
	$D_{OFF}(f_{3}(\sigma, v, \hat{\sigma})) = D_{OPT}(\sigma)$, and hence, 
	\begin{equation} \label{LMA:rand_up.input_cnct_opt:eq.1}
		C_{OFF}(f_{3}(\sigma, v, \hat{\sigma})) = C_{OPT}(\sigma). 
	\end{equation}
	We prove that $OFF$  is optimal by contradiction. 
	We assume that $OFF$ is not optimal, 
	that is, 
	there exists an offline algorithm $OPT'$ such  that 
	\begin{equation} \label{LMA:rand_up.input_cnct_opt:eq.2}
		C_{OPT'}(f_{3}(\sigma, v, \hat{\sigma})) < C_{OFF}(f_{3}(\sigma, v, \hat{\sigma})). 
	\end{equation}
	If $OPT'$ does not select $u$ for $f_{3}(\sigma, v, \hat{\sigma})$, 
	$D_{OPT'}(f_{3}(\sigma, v, \hat{\sigma}))$ is a dominating set of the graph in $\sigma$. 
	The size of this set is smaller than that of $D_{OPT}(\sigma)$
	by Eqs.~(\ref{LMA:rand_up.input_cnct_opt:eq.1}) and (\ref{LMA:rand_up.input_cnct_opt:eq.2}), 
	which contradicts the optimality of $OPT$. 
	Next, we consider the case in which $OPT'$ selects $u$. 
	Suppose that $OPT'$ selects $v$ instead of $u$. 
	Then, 
	the vertex set is also a dominating set of the graph in $f_{3}(\sigma, v, \hat{\sigma})$.  
	Moreover, 
	it is also a dominating set of the graph in $\sigma$. 
	Hence, 
	the size of this set is smaller than that of $D_{OPT}(\sigma)$ by Eqs.~(\ref{LMA:rand_up.input_cnct_opt:eq.1}) and (\ref{LMA:rand_up.input_cnct_opt:eq.2}). 
	This fact contradicts the optimality of $OPT$, and thus 
	$OFF$ is optimal. 
	\fi
\end{proof}
\ifnum \count10 > 0
%
%

%

%
\fi
\ifnum \count11 > 0
%
%

%
\fi
%
\begin{LMA} \label{LMA:rand_up.p1}
	\ifnum \count10 > 0
	%
	%
	%
	%
	%
	(a) 
	$\sigma'$
	%
	(b)
	$\frac{ {\mathbb E}[C_{RA}(\sigma)] }{ C_{OPT}(\sigma) } 
	\leq \frac{ {\mathbb E}[C_{RA}(\sigma')] }{ C_{OPT}(\sigma') }$
	%
	
	%
	\fi
	\ifnum \count11 > 0
	%
	%
	Suppose that the graph of an input $\sigma$ contains at least one free edge which is not fixed. 
	Then, 
	there exits an input $\sigma'$ such that 
	(a) 
	any free edge in the graph of $\sigma'$ is fixed, 
	that is (P1) holds, and 
	(b)
	$\frac{ {\mathbb E}[C_{RA}(\sigma)] }{ C_{OPT}(\sigma) } 
	\leq \frac{ {\mathbb E}[C_{RA}(\sigma')] }{ C_{OPT}(\sigma') }$. 
	\fi
\end{LMA}
\begin{proof}
	\ifnum \count10 > 0
	%
	%
	%
	%
	%
	(1) 
	$\sigma''$
	$\sigma$
	(2) 
	$
		\frac{ {\mathbb E}[C_{RA}(\sigma)] }{ C_{OPT}(\sigma) } 
			\leq \frac{ {\mathbb E}[C_{RA}(\sigma'')] }{ C_{OPT}(\sigma'') }
	$
	%
	%
	%
	%
	%
	
	%
	$\sigma$
	%
	$\sigma_{1} = f_{1}(\sigma, u)$%
	$\sigma_{2} = f_{2}(\sigma, u)$
	%
	%
	\begin{equation} \label{LMA:rand_up.p1:eq.1}
		D_{OPT}(\sigma_{1}) \cup D_{OPT}(\sigma_{2}) = D_{OPT}(\sigma)
	\end{equation}
	%
	$\sigma_{1}$
	$\sigma$
	$\sigma_{1}$
	$\sigma$
	$\sigma_{2}$
	$\sigma_{1}, \sigma_{2}$
	$\sigma_{1}$
	$\sigma$
	%
	$\sigma_{1}$
	%
	%
	%
	%
	$\sigma_{1}$
	\begin{equation} \label{LMA:rand_up.p1:eq.2}
		C_{OPT}(\sigma_{1}) + C_{OPT}(\sigma_{2}) = C_{OPT}(\sigma)
	\end{equation}
	%
	
	%
	$RA$
	$RA$
	%
	$\sigma$
	$v$
	%
	%
	$v$
	$\sigma$
	${deg}(v) \geq 3$
	$v''$
	$\sigma_{1}$
	${deg}_{u}(v) = 2$
	$v''$
	%
	$v''$
	%
	$v''$
	%
	%
	$u$
	$v''$
	%
	
	%
	$v$
	$v = v_{1}$
	$\sigma$
	%
	%
	$v \ne v_{1}$
	$\sigma$
	$v' (\ne u)$
	$\sigma$
	$\sigma_{1}$
	$\sigma$
	$\sigma$
	${deg}(v') \geq 3$
	${deg}_{v}(v') \leq 2$
	$\sigma_{1}$
	${deg}(v') \geq 3$
	${deg}_{v}(v') \geq 3$
	$\sigma_{1}$
	%
	
	%
	$\sigma$
	$\sigma$
	$v$
	$\sigma_{1}$
	${deg}(v) = 2$
	$v$
	$1/2$
	$\sigma$
	$\sigma$
	$v$
	$\sigma_{1}$
	${deg}(v) \geq 3$
	$v$
	%
	
	%
	$u$
	$u$
	$\sigma_{2}$
	%
	%
	$\sigma$
	$\sigma_{2}$
	$\sigma$
	$\sigma_{2}$
	%

	%
	\[
		{\mathbb E}[C_{RA}(\sigma)] 
			> {\mathbb E}[C_{RA}(\sigma_{1})] + {\mathbb E}[C_{RA}(\sigma_{2})]
	\]
	$v$
	$u$
	%
	$(v, u)$
	%
	$(v, u)$
	\[
		{\mathbb E}[C_{RA}(\sigma)] 
			\leq {\mathbb E}[C_{RA}(\sigma_{1})] + {\mathbb E}[C_{RA}(\sigma_{2})]
	\]
	%
	\begin{eqnarray*}
		\frac{ {\mathbb E}[C_{RA}(\sigma)] }{ C_{OPT}(\sigma) } 
			&\leq& \frac{ {\mathbb E}[C_{RA}(\sigma_{1})] + {\mathbb E}[C_{RA}(\sigma_{2})] }{C_{OPT}(\sigma_{1}) + C_{OPT}(\sigma_{2}) }\\
			&\leq& \max \left \{ \frac{ {\mathbb E}[C_{RA}(\sigma_{1})] }{ C_{OPT}(\sigma_{1}) }, \frac{ {\mathbb E}[C_{RA}(\sigma_{2})] }{ C_{OPT}(\sigma_{2}) } \right \}
	\end{eqnarray*} 
	$\sigma_{1}$
	\fi
	\ifnum \count11 > 0
	%
	%
	Suppose that the graph in an input $\sigma$ contains at least one free edge which is not fixed. 
	To prove the lemma, 
	we will prove that 
	we can construct an input $\sigma''$ satisfying the following two properties from $\sigma$: 
	(1) 
	the number of free edges which are not fixed in the graph of $\sigma''$ is less than that in the graph of $\sigma$, 
	and 
	(2)
	$	\frac{ {\mathbb E}[C_{RA}(\sigma)] }{ C_{OPT}(\sigma) } 
			\leq \frac{ {\mathbb E}[C_{RA}(\sigma'')] }{ C_{OPT}(\sigma'') }$. 
	Proving the existence of such $\sigma''$ means that 
	we can construct an input satisfying the conditions (a) and (b) in the statement of this lemma recursively. 
	Let $(v, u)$ be a free edge which is not fixed of the graph in $\sigma$. 
	Define 
	$\sigma_{1} = f_{1}(\sigma, u)$ and 
	$\sigma_{2} = f_{2}(\sigma, u)$. 
	By Lemma~\ref{LMA:rand_up.input_div_opt}, 
	\begin{equation} \label{LMA:rand_up.p1:eq.1}
		D_{OPT}(\sigma_{1}) \cup D_{OPT}(\sigma_{2}) = D_{OPT}(\sigma). 
	\end{equation}
	Thus, 
	if an edge of the graph in $\sigma_{1}$ is free (not free), 
	then that in $\sigma$ is also free (not free). 
	Also, 
	if a free edge of the graph in $\sigma_{1}$ is fixed (not fixed), 
	then that in $\sigma$ is also fixed (not fixed). 
	This fact is also true with edges in $\sigma_{2}$. 
	We remove the edge $(v, u)$ when constructing $\sigma_{1}$ and $\sigma_{2}$, 
	and hence 
	the number of free edges which are not fixed in $\sigma_{1}$ ($\sigma_{2}$) is less than that in $\sigma$ by at least one, 
	which implies that  
	$\sigma_{1}$ and $\sigma_{2}$ satisfy the property (1) described above. 
	In what follows, 
	we prove the property (2). 
	By Eq.~(\ref{LMA:rand_up.p1:eq.1}), 
	with respect to the costs of $OPT$ for $\sigma_{1}$ and $\sigma_{2}$, 
	\begin{equation} \label{LMA:rand_up.p1:eq.2}
		C_{OPT}(\sigma_{1}) + C_{OPT}(\sigma_{2}) = C_{OPT}(\sigma). 
	\end{equation}
	Next, 
	we evaluate the expected cost of $RA$ using Lemma~\ref{LMA:rand_up.vtx_cost}. 
	The expected cost of $RA$ for a vertex depends on its degree and the degree of an adjacent vertex, 
	and thus 
	vertices for which the difference between the expected costs in $\sigma$ and $\sigma_{1}$($\sigma_{2}$) can arise are only $v$, $u$ and their adjacent vertices. 
	First, 
	we discuss the expected cost for a vertex $v'' (\ne u)$ adjacent to $v$. 
	In $\sigma$, 
	if ${deg}(v) \geq 3$, ${deg}(v'') = 1$, 
	${deg}_{v'}(v) = 3$ and ${deg}_{u}(v) = 2$, 
	then the expected cost for $v''$ in $\sigma$ is zero. 
	Also, 
	the expected cost for $v''$ in $\sigma_{1}$ is $1/2$ 
	because ${deg}_{u}(v) = 2$. 
	Otherwise, 
	the expected cost for $v''$ does not change. 
	Thus, 
	the expected cost for $v''$ does not decrease. 
	In a similar way to the expected cost for $v''$, 
	that for a vertex adjacent to $u$ in $\sigma_{2}$ does not decrease.
	Next, 
	we evaluate the cost for $v$. 
	If $v = v_{1}$, 
	the expected cost for $v$ is one in either $\sigma$ or $\sigma_{1}$. 
	That is, it does not change. 
	In the rest of the proof, 
	we consider the case in which $v \ne v_{1}$. 
	${deg}(v) \geq 2$ in $\sigma$ holds, 
	and then 
	let $v' (\ne u)$ be the vertex adjacent to $v$. 
	We consider the case in which ${deg}(v) = 2$ in $\sigma$, 
	that is, 
	the degree of $v$ in $\sigma_{1}$ turns one. 
	The expected cost for $v$ in $\sigma$ is $1/2$. 
	In $\sigma$, 
	if either ${deg}(v') \leq 2$ or 
	${deg}(v') \geq 3$ and 
	${deg}_{v}(v') \leq 2$, 
	then 
	the expected cost for $v$ in $\sigma_{1}$ remains $1/2$. 
	If ${deg}(v') \geq 3$ and 
	${deg}_{v}(v') \geq 3$, 
	then 
	the expected cost in $\sigma_{1}$ becomes zero, 
	that is, 
	it decreases by $1/2$. 
	In the case in which ${deg}(v) = 3$ in $\sigma$, 
	the expected cost for $v$ in $\sigma$ is one. 
	Since ${deg}(v) = 2$ in $\sigma_{1}$, 
	the expected cost for $v$ is $1/2$ and decreases by $1/2$. 
	In the case in which ${deg}(v) \geq 4$ in $\sigma$, 
	the expected cost for $v$ in $\sigma$ is one. 
	${deg}(v) \geq 3$ in $\sigma_{1}$ as well, 
	and the expected cost for $v$ remains one. 
	Next, 
	we evaluate the cost for $u$. 
	By definition, 
	$u$ is the first revealed vertex in $\sigma_{2}$, 
	and thus 
	the expected cost for $u$ in $\sigma_{2}$ is one. 
	We discuss the cost for $u$ in $\sigma$. 
	If ${deg}(u) \leq 2$ in $\sigma$, 
	then the expected cost is at most $1/2$, 
	which means that 
	the expected cost for $u$ increases by at least $1/2$ in $\sigma_{2}$. 
	If ${deg}(u) \geq 3$ in $\sigma$, 
	then the expected cost is one, 
	which means that 
	the expected cost for $u$ does not change. 
	By the above argument, 
	the case in which 
	\[
		{\mathbb E}[C_{RA}(\sigma)] 
			> {\mathbb E}[C_{RA}(\sigma_{1})] + {\mathbb E}[C_{RA}(\sigma_{2})]
	\]
	can hold is only the case in which 
	the expected cost for $v$ decreases by $1/2$ and 
	that for $u$ does not change, 
	%
	that is, 
	the case in which $(v, u)$ is fixed. 
	Thus, 
	if $(v, u)$ is not fixed, 
	\[
		{\mathbb E}[C_{RA}(\sigma)] 
			\leq {\mathbb E}[C_{RA}(\sigma_{1})] + {\mathbb E}[C_{RA}(\sigma_{2})]. 
	\]
	This inequality together with Eq.~(\ref{LMA:rand_up.p1:eq.2}) shows 
	\begin{eqnarray*}
		\frac{ {\mathbb E}[C_{RA}(\sigma)] }{ C_{OPT}(\sigma) } 
			&\leq& \frac{ {\mathbb E}[C_{RA}(\sigma_{1})] + {\mathbb E}[C_{RA}(\sigma_{2})] }{C_{OPT}(\sigma_{1}) + C_{OPT}(\sigma_{2}) }\\
			&\leq& \max \left \{ \frac{ {\mathbb E}[C_{RA}(\sigma_{1})] }{ C_{OPT}(\sigma_{1}) }, \frac{ {\mathbb E}[C_{RA}(\sigma_{2})] }{ C_{OPT}(\sigma_{2}) } \right \}, 
	\end{eqnarray*} 
	which proves that 
	either $\sigma_{1}$ or $\sigma_{2}$ satisfies the property (2) described above. 
	\fi
\end{proof}
\ifnum \count10 > 0
%
%

%

%
\fi
\ifnum \count11 > 0
%
%

%
\fi
%
\begin{LMA} \label{LMA:rand_up.p2}
	\ifnum \count10 > 0
	%
	%
	(i) 
	%
	(ii) 
	(P1)
	%
	
	%
	%
	(a) 
	(P2)
	(b) 
	$\frac{ {\mathbb E}[C_{RA}(\sigma)] }{ C_{OPT}(\sigma) } 
	\leq \frac{ {\mathbb E}[C_{RA}(\sigma')] }{ C_{OPT}(\sigma') }$
	%
	
	%
	\fi
	\ifnum \count11 > 0
	%
	%
	Suppose that an input $\sigma$ satisfies the following conditions: 
	(i) 
	the graph in $\sigma$ contains at least one vertex of degree at least four, 
	and 
	(ii) 
	(P1) holds. 
	Then, 
	there exists an input $\sigma'$ such that 
	(a) 
	the degree of any vertex of the graph in $\sigma'$ is at most three, 
	that is, (P2) holds, 
	and 
	(b) 
	$\frac{ {\mathbb E}[C_{RA}(\sigma)] }{ C_{OPT}(\sigma) } 
	\leq \frac{ {\mathbb E}[C_{RA}(\sigma')] }{ C_{OPT}(\sigma') }$. 
	\fi
\end{LMA}
\begin{proof}
	\ifnum \count10 > 0
	%
	%
	%
	$\sigma$
	$\sigma$
	${deg}(v) \geq 4$
	${deg}_{u}(v) = {deg}(v)$
	$u$
	%
	%
	%
	(1) 
	$\sigma$
	$\sigma''$
	%
	$v$
	$T''$
	$T$
	%
	$T''$
	(2) 
	$	\frac{ {\mathbb E}[C_{RA}(\sigma)] }{ C_{OPT}(\sigma) } 
			\leq \frac{ {\mathbb E}[C_{RA}(\sigma'')] }{ C_{OPT}(\sigma'') }
	$
	%
	%
	%
	%
	%
	(Case~1) 
	$OPT$
	(Case~2) 
	$OPT$
	(Case~3)
	%
	
	%
	\noindent
	{\bf (Case~1):}
	$v$
	$OPT$
	${deg}(v) \geq 4$
	$v$
	%
	$u'$
	$u'$
	%
	%
	$v$
	${deg}(v) \geq 4$
	%
	$\sigma$
	%
	%
	
	%
	\noindent
	{\bf (Case~2):}
	$U$
	$U'$
	%
	%
	%
	(Case~2-1) 
	$OPT$
	(Case~2-2) 
	$OPT$
	$\sigma_{1} = f_{1}(\sigma, u)$
	$\sigma_{2} = f_{2}(\sigma, u)$
	%
	
	%
	\noindent
	{\bf (Case~2-1):}
	Case~2
	%
	${deg}(v') = 3$
	$u'$
	${deg}_{u'}(v') = 3$
	${deg}(u') = 1$
	$OPT$
	%
	%
	$\sigma_{1}$
	%
	$\sigma'_{1} = f_{1}(\sigma_{1}, u')$
	$\sigma_{3} = f_{3}(\sigma'_{1}, v', \sigma_{2})$
	%
	%
	$\sigma_{3}$
	(2)
	%
	$\sigma$
	$\sigma_{3}$
	%
	$D_{OFF}(\sigma_{3}) = D_{OPT}(\sigma)$
	%
	$\sigma_{3}$
	$D_{OPT}(\sigma)$
	%
	\begin{equation} \label{LMA:rand_up.p2:eq.2}
		C_{OFF}(\sigma_{3}) = C_{OPT}(\sigma)
	\end{equation}
	%

	%
	%
	$\sigma_{3}$
	$\sigma$
	%
	%
	%
	$v$
	%
	$\sigma$
	${deg}_{u}(v) = {deg}(v) \geq 3$
	$\sigma_{3}$
	${deg}_{u}(v') = {deg}(v') = 3$
	$u$
	%
	%
	$\sigma$
	${deg}_{u'}(v') = 3$
	$u'$
	%
	%
	\[
		{\mathbb E}[C_{RA}(\sigma_{3})] = {\mathbb E}[C_{RA}(\sigma)]
	\]
	%
	\[
		\frac{ {\mathbb E}[C_{RA}(\sigma)] }{ C_{OPT}(\sigma) } 
			= \frac{ {\mathbb E}[C_{RA}(\sigma_{3})] }{ C_{OFF}(\sigma_{3}) }
			\leq \frac{ {\mathbb E}[C_{RA}(\sigma_{3})] }{ C_{OPT}(\sigma_{3}) }
	\]
	$\sigma_{3}$
	\noindent
	{\bf (Case~2-2):}
	$v'' \in U'$
	%
	$\sigma'_{3} = f_{3}(\sigma_{1}, v'', \sigma_{2})$
	%
	$\sigma'_{3}$
	$\sigma'_{3}$
	Case~2-1
	$D_{OPT}(\sigma)$
	$\sigma'_{3}$
	$D_{OPT}(\sigma)$
	%
	\begin{equation} \label{LMA:rand_up.p2:eq.3}
		C_{OFF}(\sigma'_{3}) = C_{OPT}(\sigma)
	\end{equation}
	%
	
	%
	$RA$
	$v$
	Case~2-1
	$v''$
	$v''$
	$v''$
	%
	\[
		{\mathbb E}[C_{RA}(\sigma'_{3})] \geq {\mathbb E}[C_{RA}(\sigma_{3})]
	\]
	%
	\[
		\frac{ {\mathbb E}[C_{RA}(\sigma)] }{ C_{OPT}(\sigma) } 
			\leq \frac{ {\mathbb E}[C_{RA}(\sigma'_{3})] }{ C_{OFF}(\sigma'_{3}) }
			\leq \frac{ {\mathbb E}[C_{RA}(\sigma'_{3})] }{ C_{OPT}(\sigma'_{3}) }
	\]
	%
	$\sigma'_{3}$
	\noindent
	{\bf (Case~3):}
	$v$
	$u$
	%
	${deg}(v) \geq 4$
	$(v, u)$
	%
	
	%
	\fi
	\ifnum \count11 > 0
	%
	%
	We prove this lemma in a similar way to the proof of Lemma~\ref{LMA:rand_up.p1}. 
	Suppose that an input $\sigma$ satisfies the conditions (i) and (ii) in the statement. 
	Let $v$ and $u$ be vertices in $\sigma$ 
	such that $u$ arrives at $v$, 
	${deg}(v) \geq 4$, 
	${deg}_{u}(v) = {deg}(v)$ and 
	the degrees of all the descendants of $u$ are at most three. 
	To prove the lemma, 
	we will prove that 
	we can construct an input $\sigma''$ satisfying the following two properties from $\sigma$: 
	(1) 
	all the vertices of the graph $T''$ in $\sigma''$ are in 
	the vertices of the graph $T$ in $\sigma$, 
	the degree of each vertex except for $v$ of $T''$ is equal to that of $T$, 
	the degree of $v$ in $T''$ is less than that of $T$, 
	and 
	(2) 
	$	\frac{ {\mathbb E}[C_{RA}(\sigma)] }{ C_{OPT}(\sigma) } 
			\leq \frac{ {\mathbb E}[C_{RA}(\sigma'')] }{ C_{OPT}(\sigma'') }$. 
	We can construct an input satisfying the properties (a) and (b) in the statement recursively by proving the existence of $\sigma''$. 
	In what follows, 
	we discuss such input $\sigma''$ for each of the following three cases: 
	(Case~1) 
	$OPT$ selects $u$ and only $u$ dominates $v$, 
	(Case~2) 
	$OPT$ selects $v$ and only $v$ dominates $u$, 
	and 
	(Case~3)
	the other cases. 
	\noindent
	{\bf (Case~1):}
	$OPT$ does not select $v$ by the condition of this case. 
	Since ${deg}(v) \geq 4$, 
	there exists another vertex $u' (\ne u)$ arriving at $v$. 
	$u'$ does not dominate $v$ by the condition of this case, and 
	$OPT$ selects a descendant of $u'$ to dominate $u'$. 
	Thus, 
	the edge $(v, u')$ is free. 
	${deg}(v) \geq 4$ by the definition of $v$, and 
	hence $(v, u')$ is not fixed, 
	which contradicts that 
	$\sigma$ satisfies the property (P1) by the condition (ii) in the statement of the lemma. 
	Therefore, 
	this case does not happen. 
	\noindent
	{\bf (Case~2):}
	Let $U$ be the set of vertices consisting of $u$ and all the descendants of $u$, 
	and let $U'$ be the set of vertices of all the descendants of $v$ except for $U$. 
	%
	%
	We consider the following two cases: 
	(Case~2-1) 
	the degree of any vertex selected by $OPT$ in $U'$ is three, 
	and 
	(Case~2-2) 
	the other case, 
	that is, 
	there exists at least one vertex selected by $OPT$ in $U'$ 
	whose degree is at most two. 
	Define $\sigma_{1} = f_{1}(\sigma, u)$ and $\sigma_{2} = f_{2}(\sigma, u)$. 
	\noindent
	{\bf (Case~2-1):}
	There exist vertices $v', u' \in U'$ by the conditions of Cases 2 and 2-1 
	such that 
	${deg}(v') = 3$, 
	$u'$ arrives at $v'$, 
	${deg}_{u'}(v') = 3$, 
	${deg}(u') = 1$ and 
	$OPT$ selects $v'$. 
	Here, 
	we construct the input $\sigma'_{1}$ by removing $u'$ from $\sigma_{1}$, and 
	obtain the input $\sigma_{3}$ by adding $\sigma_{2}$ to $\sigma'_{1}$ in terms of $v'$. 
	Formally, 
	define $\sigma'_{1} = f_{1}(\sigma_{1}, u')$, 
	$\sigma_{2} = f_{2}(\sigma, u)$, 
	and 
	$\sigma_{3} = f_{3}(\sigma'_{1}, v', \sigma_{2})$
	(see Fig.~\ref{fig:fig_p2}). 
	Then, 
	$\sigma_{3}$ satisfies the above property (1). 
	To show that the property (2) is satisfied, 
	we first evaluate the cost of $OPT$ for $\sigma_{3}$. 
	Let $OFF$ be an offline algorithm for $\sigma_{3}$ which selects the same vertices as those selected by $OPT$ for $\sigma$, 
	that is, $D_{OFF}(\sigma_{3}) = D_{OPT}(\sigma)$. 
	Then, 
	$D_{OPT}(\sigma)$ is a dominating set of the graph in $\sigma_{3}$ by the definition of $\sigma_{3}$. 
	Thus, 
	\begin{equation} \label{LMA:rand_up.p2:eq.2}
		C_{OFF}(\sigma_{3}) = C_{OPT}(\sigma). 
	\end{equation}
	We next consider the expected cost of $RA$ using Lemma~\ref{LMA:rand_up.vtx_cost}. 
	The degree of each vertex except for $v$ in $\sigma_{3}$ is equal to that in $\sigma$. 
	The above input transformations do not affect the expected costs for vertices except for $v, v'$ and their adjacent vertices. 
	The expected cost for $v$ does not change 
	because ${deg}(v) \geq 3$ before and after the transformations. 
	Also, 
	${deg}_{u}(v) = {deg}(v) \geq 3$ in $\sigma$ and 
	${deg}_{u}(v') = {deg}(v') = 3$ in $\sigma_{3}$. 
	Thus, 
	the expected cost for $u$ does not change, either. 
	The expected cost for $u'$ is zero 
	because ${deg}_{u'}(v') = 3$ in $\sigma$. 
	Moreover, 
	the transformations do not affect the expected costs of all the vertices adjacent to either $v$ or $v'$ except for $u$ and $u'$. 
	Hence, 
	\[
		{\mathbb E}[C_{RA}(\sigma_{3})] = {\mathbb E}[C_{RA}(\sigma)]. 
	\]
	By this equality and Eq.~(\ref{LMA:rand_up.p2:eq.2}), 
	\[
		\frac{ {\mathbb E}[C_{RA}(\sigma)] }{ C_{OPT}(\sigma) } 
			= \frac{ {\mathbb E}[C_{RA}(\sigma_{3})] }{ C_{OFF}(\sigma_{3}) }
			\leq \frac{ {\mathbb E}[C_{RA}(\sigma_{3})] }{ C_{OPT}(\sigma_{3}) }, 
	\]
	which implies that 
	$\sigma_{3}$ satisfies the property (2). 
	\noindent
	{\bf (Case~2-2):}
	Let $v'' \in U'$ be a vertex selected by $OPT$ of degree at most two. 
	Then, 
	define $\sigma'_{3} = f_{3}(\sigma_{1}, v'', \sigma_{2})$, 
	which satisfies the property (1). 
	Let us consider the cost of $OPT$ for $\sigma'_{3}$. 
	In the same way as Case~2-1, 
	let $OFF$ be an offline algorithm which selects all the vertices in $D_{OPT}(\sigma)$. 
	$D_{OPT}(\sigma)$ is a dominating set of the graph in $\sigma_{3}$ by the definition of $\sigma'_{3}$. 
	Thus, 
	\begin{equation} \label{LMA:rand_up.p2:eq.3}
		C_{OFF}(\sigma'_{3}) = C_{OPT}(\sigma). 
	\end{equation}
	We consider the expected cost of $RA$. 
	The degrees of all the vertices except for $v$ and $v''$ do not change. 
	The expected costs for all the vertices except for $v''$ do not change by the same argument in Case~2-1. 
	Since the degree of $v''$ is at least two, 
	the expected cost for $v''$ does not decrease. 
	Hence, 
	\[
		{\mathbb E}[C_{RA}(\sigma'_{3})] \geq {\mathbb E}[C_{RA}(\sigma_{3})]. 
	\]
	By this inequality and Eq.~(\ref{LMA:rand_up.p2:eq.3}), 
	\[
		\frac{ {\mathbb E}[C_{RA}(\sigma)] }{ C_{OPT}(\sigma) } 
			\leq \frac{ {\mathbb E}[C_{RA}(\sigma'_{3})] }{ C_{OFF}(\sigma'_{3}) }
			\leq \frac{ {\mathbb E}[C_{RA}(\sigma'_{3})] }{ C_{OPT}(\sigma'_{3}) }. 
	\]
	Thus,  
	$\sigma'_{3}$ satisfies the property (2). 
	\noindent
	{\bf (Case~3):}
	$(v, u)$ is free 
	because $v$ ($u$) is dominated by a vertex which is not $u$ ($v$). 
	${deg}(v) \geq 4$, and thus $(v, u)$ is not fixed, 
	which contradicts that $\sigma$ satisfies (P1) by the condition (ii). 
	\fi
\end{proof}
\ifnum \count12 > 0
\begin{figure*}[ht]
	 \begin{center}
	  \includegraphics[width=130mm]{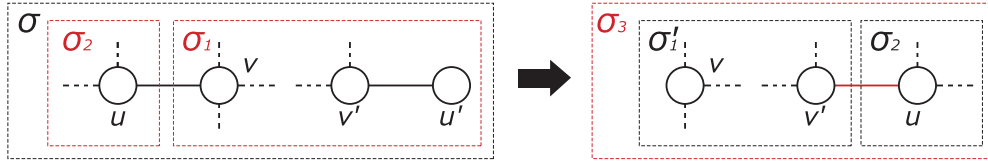}
	 \end{center}
	 \caption{
\ifnum \count10 > 0
%
%
$\sigma$
\fi
\ifnum \count11 > 0
%
%
Input transformation from $\sigma$ to $\sigma_{3}$. 
\fi
			}
	\label{fig:fig_p2}
\end{figure*}
\fi
\ifnum \count10 > 0
%
%

%

%
\fi
\ifnum \count11 > 0
%
%

%
\fi
%
\begin{LMA} \label{LMA:rand_up.p1r}
	\ifnum \count10 > 0
	%
	%
	%
	(i) $\sigma$
	(ii) (P2)
	%
	
	%
	%
	(a) 
	(P1)
	(b)
	$\frac{ {\mathbb E}[C_{RA}(\sigma)] }{ C_{OPT}(\sigma) } 
	\leq \frac{ {\mathbb E}[C_{RA}(\sigma')] }{ C_{OPT}(\sigma') }$
	%
	
	%
	\fi
	\ifnum \count11 > 0
	%
	%
	Suppose that an input $\sigma$ satisfies the following conditions: 
	(i)
	the graph in $\sigma$ contains at least one free edge which is not fixed, and 
	(ii)
	(P2) holds. 
	Then, 
	there exists an input $\sigma'$ such that 
	(a) 
	(P1) and (P2) hold, 
	and 
	(b)
	$\frac{ {\mathbb E}[C_{RA}(\sigma)] }{ C_{OPT}(\sigma) } 
	\leq \frac{ {\mathbb E}[C_{RA}(\sigma')] }{ C_{OPT}(\sigma') }$. 
	\fi
\end{LMA}
\begin{proof}
	\ifnum \count10 > 0
	%
	%
	2
	(a)
	$f_{1}$
	(ii)
	%
	
	%
	\fi
	\ifnum \count11 > 0
	%
	%
	In the same way as the proof of Lemma~\ref{LMA:rand_up.p1}, 
	we can obtain an input satisfying the property (b) and (P1) in the property (a) in the statement 
	by constructing the two inputs recursively from an input satisfying the conditions (i) and (ii).   
	The degree of any vertex in a newly constructed input does not increase 
	by the definitions of $f_{1}$ and $f_{2}$. 
	Thus, 
	the degree is at most three by (P2) in the condition (ii), 
	which means that this input also satisfies (P2) in (a). 
	\fi
\end{proof}
\ifnum \count10 > 0
%
%

%

%
\fi
\ifnum \count11 > 0
%
%

%
\fi
%
\begin{LMA} \label{LMA:rand_up.p3}
	\ifnum \count10 > 0
	%
	%
	(i) 
	$OPT$
	(ii) 
	(P1)
	%
	
	%
	%
	(a) 
	$OPT$
	%
	(b)
	(P1)
	(c) 
	$\frac{ {\mathbb E}[C_{RA}(\sigma)] }{ C_{OPT}(\sigma) } 
	\leq \frac{ {\mathbb E}[C_{RA}(\sigma')] }{ C_{OPT}(\sigma') }$
	%
	
	%
	\fi
	\ifnum \count11 > 0
	%
	%
	Suppose that an input $\sigma$ satisfies the following conditions: 
	(i) 
	$OPT$ selects at least one vertex of degree at most two, and 
	(ii) 
	(P1) and (P2) hold. 
	Then, 
	there exists an input $\sigma'$ such that 
	(a) 
	the degree of any vertex selected by $OPT$ is three, 
	that is, 
	(P3) holds, 
	(b)
	(P1) and (P2) hold, and 
	(c) 
	$\frac{ {\mathbb E}[C_{RA}(\sigma)] }{ C_{OPT}(\sigma) } 
	\leq \frac{ {\mathbb E}[C_{RA}(\sigma')] }{ C_{OPT}(\sigma') }$. 
	\fi
\end{LMA}
\begin{proof}
	\ifnum \count10 > 0
	%
	%
	%
	%
	(i)
	%
	%
	%
	(1) 
	$\sigma_{3}$
	$\sigma$
	(2)
	(P1),(P2)
	(3) 
	$	\frac{ {\mathbb E}[C_{RA}(\sigma)] }{ C_{OPT}(\sigma) } 
			\leq \frac{ {\mathbb E}[C_{RA}(\sigma_{3})] }{ C_{OPT}(\sigma_{3}) }$
	%
	%
	%
	
	%
	$\sigma$
	${deg}(v) \leq 2$
	$v$
	%
	%
	$v$
	%
	%
	${deg}(v) = 2$
	$\sigma_{3} = f_{3}(\sigma, v, \hat{\sigma})$
	$\hat{\sigma} = (\{ u \}, \varnothing, u)$
	%
	%
	${deg}(v) = 1$
	$\sigma_{3} = f_{3}( f_{3}(\sigma, v, \hat{\sigma}), v, \hat{\sigma}')$
	$\hat{\sigma}' = (\{ u' \}, \varnothing, u')$
	%
	
	%
	$\sigma$
	$v$
	$v$
	$\sigma_{3}$
	%
	%
	\begin{equation} \label{LMA:rand_up.p3:eq.2}
		D_{OPT}(\sigma_{3}) = D_{OPT}(\sigma)
	\end{equation}
	%
	$\sigma_{3}$
	%
	%
	$C_{OPT}(\sigma_{3}) = C_{OPT}(\sigma)$
	%
	$(v, u)$
	%
	$\sigma_{3}$
	%
	%
	$\sigma$
	$RA$
	\[
		{\mathbb E}[C_{RA}(\sigma_{3})] \geq {\mathbb E}[C_{RA}(\sigma)]
	\]
	%
	$\sigma_{3}$
	\fi
	\ifnum \count11 > 0
	%
	%
	We prove this lemma in a similar way to the proofs of the last three lemmas. 
	Suppose that an input $\sigma$ satisfies the conditions (i) and (ii) in the statement of this lemma. 
	To prove the lemma, 
	we will prove that 
	we can construct an input $\sigma_{3}$ satisfying the following three properties from $\sigma$: 
	(1) 
	the number of vertices selected by $OPT$ of degree at most two of the graph in $\sigma_{3}$ is less than that in $\sigma$, 
	(2)
	(P1) and (P2) hold, 
	and 
	(3) 
	$	\frac{ {\mathbb E}[C_{RA}(\sigma)] }{ C_{OPT}(\sigma) } 
			\leq \frac{ {\mathbb E}[C_{RA}(\sigma_{3})] }{ C_{OPT}(\sigma_{3}) }$.
	We can construct an input satisfying the properties (a), (b) and (c) in the statement recursively by proving the existence of $\sigma_{3}$. 
	Let $v$ be a vertex of the graph given in $\sigma$ such that 
	${deg}(v) \leq 2$ and $OPT$ selects $v$.  
	Roughly speaking, 
	we obtain $\sigma_{3}$ by connecting $v$ with a couple of vertices of degree one 
	so that the degree of $v$ is three. 
	Formally, 
	if ${deg}(v) = 2$, 
	we define $\sigma_{3} = f_{3}(\sigma, v, \hat{\sigma})$, 
	in which 
	$\hat{\sigma} = (\{ u \}, \varnothing, u)$. 
	If ${deg}(v) = 1$, 
	we define $\sigma_{3} = f_{3}( f_{3}(\sigma, v, \hat{\sigma}), v, \hat{\sigma}')$, 
	in which 
	$\hat{\sigma}' = (\{ u' \}, \varnothing, u')$. 
	(P2) holds in $\sigma$ 
	by the condition (ii) in the statement. 
	Also, 
	the input transformation by $f_{3}$ does not change 
	the degrees of the vertices except for $v$, 
	and the degree of $v$ is three in $\sigma_{3}$. 
	Hence, 
	(P2) also holds in $\sigma_{3}$. 
	By Lemma~\ref{LMA:rand_up.input_cnct_opt}, 
	\begin{equation} \label{LMA:rand_up.p3:eq.2}
		D_{OPT}(\sigma_{3}) = D_{OPT}(\sigma). 
	\end{equation}
	The above facts shows that 
	$\sigma_{3}$ satisfies the property (1) described above. 
	Also, 
	$C_{OPT}(\sigma_{3}) = C_{OPT}(\sigma)$ by Eq.~(\ref{LMA:rand_up.p3:eq.2}). 
	Since $OPT$ selects neither $u$ nor $u'$, 
	neither $(v, u)$ nor $(v, u')$ is free. 
	That is, 
	(P1) also holds in $\sigma_{3}$ by the condition (ii). 
	Hence, 
	we have shown that $\sigma_{3}$ satisfies the property (2). 
	Moreover, 
	even if a new vertex arrives, 
	the cost of $RA$ does not decrease, and hence 
	\[
		{\mathbb E}[C_{RA}(\sigma_{3})] \geq {\mathbb E}[C_{RA}(\sigma)]. 
	\]
	$\sigma_{3}$ satisfies the property (3) 
	by this inequality and Eq.~(\ref{LMA:rand_up.p3:eq.2}). 
	\fi
\end{proof}
\ifnum \count10 > 0
%
%

%

%
\fi
\ifnum \count11 > 0
%
%

%
\fi
%

\begin{LMA} \label{LMA:rand_up.p4}
	\ifnum \count10 > 0
	%
	%
	(i) 
	$OPT$
	(ii) 
	(P1), (P2), (P3)
	%
	
	%
	%
	(a) 
	$OPT$
	%
	(b)
	(P1),(P2),(P3)
	(c) 
	$\frac{ {\mathbb E}[C_{RA}(\sigma)] }{ C_{OPT}(\sigma) } 
	\leq \frac{ {\mathbb E}[C_{RA}(\sigma')] }{ C_{OPT}(\sigma') }$
	%
	
	%
	\fi
	\ifnum \count11 > 0
	%
	%
	Suppose that an input $\sigma$ satisfies the following conditions: 
	(i) 
	there exists at least one free edge $(v, u)$ such that $OPT$ selects $v$, 
	and 
	(ii) 
	(P1), (P2) and (P3) hold. 
	Then, 
	there exists an input $\sigma'$ such that 
	(a) 
	for any free edge $(v, u)$, 
	$OPT$ does not select $v$, 
	that is, 
	(P4) holds, 
	(b)
	(P1), (P2) and (P3) hold, 
	and 
	(c) 
	$\frac{ {\mathbb E}[C_{RA}(\sigma)] }{ C_{OPT}(\sigma) } 
	\leq \frac{ {\mathbb E}[C_{RA}(\sigma')] }{ C_{OPT}(\sigma') }$. 
	\fi
\end{LMA}
\begin{proof}
	\ifnum \count10 > 0
	%
	%
	%
	(i)
	$\sigma$
	$OPT$
	%
	%
	%
	(1) 
	$\sigma''$
	$(v, u)$
	$\sigma$
	(2)
	(P1),(P2),(P3)
	(3) 
	$	\frac{ {\mathbb E}[C_{RA}(\sigma)] }{ C_{OPT}(\sigma) } 
			\leq \frac{ {\mathbb E}[C_{RA}(\sigma'')] }{ C_{OPT}(\sigma'') }
	$
	%
	%
	%
	
	%
	$\sigma_{1} = f_{1}(\sigma, u)$%
	$\sigma_{2} = f_{2}(\sigma, u)$
	%
	%
	%
	$\sigma$
	(P3)
	$\sigma$
	%
	$\sigma'_{1} = f_{3}(\sigma_{1}, v, \hat{\sigma}_{1})$
	$\hat{\sigma}_{1} = (\{ u_{1} \}, \varnothing, u_{1})$
	%
	%
	$u$
	$\sigma'_{2} = f_{3}(\sigma_{2}, u, \hat{\sigma}_{2})$
	$\hat{\sigma}_{2} = (\{ u_{2} \}, \varnothing, u_{2})$
	%
	$\sigma'_{2} = \sigma_{2}$
	%
	%
	$v$
	$v$
	%
	$\sigma'_{1}$($\sigma'_{2}$)
	%
	%
	$D_{OPT}(\sigma_{1}) \cup D_{OPT}(\sigma_{2}) = D_{OPT}(\sigma)$
	%
	$D_{OPT}(\sigma'_{1}) = D_{OPT}(\sigma_{1})$
	$D_{OPT}(\sigma'_{2}) = D_{OPT}(\sigma_{2})$
	%
	%
	\begin{equation} \label{LMA:rand_up.p4:eq.2}
		D_{OPT}(\sigma'_{1}) + D_{OPT}(\sigma'_{2}) = D_{OPT}(\sigma)
	\end{equation}
	%
	%
	$\sigma$
	$\sigma'_{1}$
	%
	%
	$(v, u_{1})$
	%
	$u_{2}$
	$(u, u_{2})$
	%
	%
	$\sigma'_{1}$($\sigma'_{2}$)
	%
	%
	$\sigma'_{1}$($\sigma'_{2}$)
	%
	%
	$\sigma'_{1}$
	$\sigma$
	$\sigma'_{2}$
	%
	%
	%
	$\sigma'_{1}$($\sigma'_{2}$)
	%
	%
	$\sigma'_{1}$($\sigma'_{2}$)
	%
	
	%
	\begin{equation} \label{LMA:rand_up.p4:eq.3}
		{\mathbb E}[C_{RA}(\sigma)] 
			\leq {\mathbb E}[C_{RA}(\sigma'_{1})] + {\mathbb E}[C_{RA}(\sigma'_{2})]
	\end{equation}
	$C_{OPT}(\sigma'_{1}) + C_{OPT}(\sigma'_{2}) = C_{OPT}(\sigma)$
	(3)
	%
	%
	$RA$
	%
	$v$
	$v$
	$u$
	%
	%
	$v$
	%
	%
	%
	$u$
	%
	%
	%
	
	%
	\fi
	\ifnum \count11 > 0
	%
	%
	We prove this lemma similarly to the proofs of the above lemmas. 
	Suppose that an input $\sigma$ satisfies the conditions (i) and (ii) in the statement. 
	Let $(v, u)$ be a free edge of the graph in $\sigma$. 
	Suppose that $OPT$ selects $v$. 
	To prove the lemma, 
	we will prove that 
	we can construct an input $\sigma''$ satisfying the following three properties from $\sigma$: 
	(1) 
	the number of such free edges $(v, u)$ in $\sigma''$ is less than that in $\sigma$, 
	(2)
	(P1), (P2) and (P3) hold, 
	and 
	(3) 
	$
		\frac{ {\mathbb E}[C_{RA}(\sigma)] }{ C_{OPT}(\sigma) } 
			\leq \frac{ {\mathbb E}[C_{RA}(\sigma'')] }{ C_{OPT}(\sigma'') }$. 
	We can construct an input satisfying the properties (a), (b) and (c) in the statement recursively by proving the existence of $\sigma''$. 
	Define 
	$\sigma_{1} = f_{1}(\sigma, u)$ and 
	$\sigma_{2} = f_{2}(\sigma, u)$. 
	Furthermore, 
	we add a few vertices to these inputs in the same way as the proof of Lemma~\ref{LMA:rand_up.p3}. 
	Since $\sigma$ satisfies (P3) by the condition (ii), 
	the degree of $v$ is three in $\sigma$. 
	Then, 
	define $\sigma'_{1} = f_{3}(\sigma_{1}, v, \hat{\sigma}_{1})$, 
	in which 
	$\hat{\sigma}_{1} = (\{ u_{1} \}, \varnothing, u_{1})$. 
	If $OPT$ selects $u$, 
	define $\sigma'_{2} = f_{3}(\sigma_{2}, u, \hat{\sigma}_{2})$, 
	in which 
	$\hat{\sigma}_{2} = (\{ u_{2} \}, \varnothing, u_{2})$. 
	Otherwise, 
	define $\sigma'_{2} = \sigma_{2}$. 
	After the input transformation, 
	the degrees of $v$ and $u$ are at most three and 
	the degrees of the other vertices do not change. 
	Also, 
	$\sigma$ satisfies (P2) by the condition (ii). 
	Thus, 
	$\sigma'_{1}$ ($\sigma'_{2}$) also satisfies (P2). 
	$D_{OPT}(\sigma_{1}) \cup D_{OPT}(\sigma_{2}) = D_{OPT}(\sigma)$ 
	by Lemma~\ref{LMA:rand_up.input_div_opt}. 
	Both 
	$D_{OPT}(\sigma'_{1}) = D_{OPT}(\sigma_{1})$ 
	and 
	$D_{OPT}(\sigma'_{2}) = D_{OPT}(\sigma_{2})$ 
	by Lemma~\ref{LMA:rand_up.input_cnct_opt}. 
	By the above equalities, 
	\begin{equation} \label{LMA:rand_up.p4:eq.2}
		D_{OPT}(\sigma'_{1}) + D_{OPT}(\sigma'_{2}) = D_{OPT}(\sigma). 
	\end{equation}
	By this equality, 
	each of the free edges except $(v, u)$ is contained in either $\sigma'_{1}$  or $\sigma'_{2}$. 
	Whether these edges are fixed does not change. 
	Also, 
	neither $(v, u_{1})$ nor $(u, u_{2})$ (if any) is free. 
	Thus, 
	$\sigma'_{1}$ ($\sigma'_{2}$) satisfies the above property (1). 
	These facts indicate that 
	$\sigma'_{1}$ ($\sigma'_{2}$) satisfies (P1)
	because $\sigma$ satisfies (P1) by the condition (ii). 
	The degree of $v$ in $\sigma'_{1}$ is three. 
	The degree of $u$ in $\sigma'_{2}$ is also three
	if $OPT$ selects $u$ in $\sigma$. 
	The degrees of the other vertices selected by $OPT$ do not change. 
	Hence, 
	$\sigma'_{1}$ ($\sigma'_{2}$) satisfies (P3) 
	because 
	$\sigma$ satisfies (P3) by the condition (ii). 
	By the above argument, 
	$\sigma'_{1}$ ($\sigma'_{2}$) satisfies the property (2). 
	In what follows, 
	we show 
	\begin{equation} \label{LMA:rand_up.p4:eq.3}
		{\mathbb E}[C_{RA}(\sigma)] 
			\leq {\mathbb E}[C_{RA}(\sigma'_{1})] + {\mathbb E}[C_{RA}(\sigma'_{2})]. 
	\end{equation}
	This inequality together with $C_{OPT}(\sigma'_{1}) + C_{OPT}(\sigma'_{2}) = C_{OPT}(\sigma)$ from Eq.~(\ref{LMA:rand_up.p4:eq.2}) shows that 
	the property (3) is satisfied. 
	We evaluate the cost of $RA$ using Lemma~\ref{LMA:rand_up.vtx_cost}. 
	The input transformations by $f_{1}$ and $f_{2}$ do not affect the vertices except for $v$, $u$ and adjacent vertices to either $v$ or $u$. 
	The degree of $v$ does not change. 
	$u$ becomes $v_{1}$ in $\sigma'_{2}$, 
	for which the expected cost is one. 
	Thus, 
	the expected cost for $u$ does not decrease. 
	The expected cost for an adjacent vertex to $v$ does not decrease 
	by the same argument as $v'$ in the proof of Lemma~\ref{LMA:rand_up.p1}. 
	%
	%
	%
	The expected cost for an adjacent vertex to  $u$ does not decrease, either. 
	Now 
	we have shown Eq.~(\ref{LMA:rand_up.p4:eq.3}). 
	\fi
\end{proof}
\ifnum \count10 > 0
%
%

%

%
\fi
\ifnum \count11 > 0
%
%
\begin{LMA} \label{LMA:rand_up.p5}
	\ifnum \count10 > 0
	%
	%
	(i) 
	%
	(ii) 
	(P1), (P2), (P3), (P4)
	%
	
	%
	%
	(a) 
	%
	(b)
	(P1),(P2),(P3), (P4)
	(c) 
	$\frac{ {\mathbb E}[C_{RA}(\sigma)] }{ C_{OPT}(\sigma) } 
	\leq \frac{ {\mathbb E}[C_{RA}(\sigma')] }{ C_{OPT}(\sigma') }$
	%
	
	%
	\fi
	\ifnum \count11 > 0
	%
	%
	Suppose that an input $\sigma$ satisfies the following conditions: 
	(i) 
	the graph in $\sigma$ contains at least one good vertex triplet, 
	and 
	(ii) 
	the properties from (P1) to (P4) inclusive hold. 
	Then, 
	there exists an input $\sigma'$ such that 
	(a) 
	the graph in $\sigma'$ contains no good vertex triplets, 
	that is, 
	(P5) holds, 
	(b)
	the properties from (P1) to (P4) inclusive hold, 
	and 
	(c) 
	$\frac{ {\mathbb E}[C_{RA}(\sigma)] }{ C_{OPT}(\sigma) } 
	\leq \frac{ {\mathbb E}[C_{RA}(\sigma')] }{ C_{OPT}(\sigma') }$. 
	\fi
\end{LMA}
\begin{proof}
	\ifnum \count10 > 0
	%
	%
	%
	(i)
	%
	%
	%
	(1) 
	$\sigma''$
	$\sigma$
	(2)
	(P1),(P2),(P3),(P4)
	(3) 
	$	\frac{ {\mathbb E}[C_{RA}(\sigma)] }{ C_{OPT}(\sigma) } 
			\leq \frac{ {\mathbb E}[C_{RA}(\sigma'')] }{ C_{OPT}(\sigma'') }
	$
	%
	%
	%
	
	%
	$\sigma$
	%
	$u_{1}$
	$OPT$
	%
	%
	$u_{1}$
	$u_{2}$
	$OPT$
	%
	$u_{2}$
	%
	%
	$\sigma_{1} = f_{1}(\sigma, u_{2})$%
	$\sigma_{2} = f_{2}(\sigma, u_{2})$%
	$\sigma'_{2} = f_{1}(\sigma_{2}, u_{3})$%
	$\sigma''_{2} = f_{1}(\sigma_{2}, u_{1})$
	%
	%
	$u (\ne u_{1}, u_{3})$
	$u$
	$(u_{2}, u)$
	%
	\begin{equation} \label{LMA:rand_up.p5:eq.1}
		D_{OPT}(\sigma_{1}) \cup D_{OPT}(\sigma_{2}) = D_{OPT}(\sigma)
	\end{equation}
	%
	%
	$OPT$
	$(u_{2}, u_{1})$
	%
	$D_{OPT}(\sigma'_{2}) \cup D_{OPT}(f_{2}(\sigma_{2}, u_{3})) = D_{OPT}(\sigma_{2})$
	$D_{OPT}(\sigma''_{2}) \cup D_{OPT}(f_{2}(\sigma_{2}, u_{1})) = D_{OPT}(\sigma_{2})$
	%
	%
	$\sigma$
	$OPT$
	%
	$D_{OPT}(\sigma'_{2}) \cap D_{OPT}(\sigma''_{2}) = \varnothing$
	$D_{OPT}(f_{2}(\sigma_{2}, u_{3})) = D_{OPT}(\sigma''_{2})$
	%
	\[
		D_{OPT}(\sigma'_{2}) \cup D_{OPT}(\sigma''_{2}) = D_{OPT}(\sigma_{2}) 
	\]
	%
	%
	\begin{equation} \label{LMA:rand_up.p5:eq.0}
		D_{OPT}(\sigma_{1}) \cup D_{OPT}(\sigma'_{2}) \cup D_{OPT}(\sigma''_{2}) 
			= D_{OPT}(\sigma)
	\end{equation}
	%
	%
	$(u_{2}, u_{1})$
	$\sigma_{1}, \sigma'_{2}, \sigma''_{2}$
	$\sigma$
	%
	%
	%
	$\sigma$
	$\sigma_{1}, \sigma'_{2}, \sigma''_{2}$
	%
	%
	
	%
	$\sigma_{1}, \sigma'_{2}, \sigma''_{2}$
	%
	\begin{equation} \label{LMA:rand_up.p5:eq.2}
		C_{OPT}(\sigma_{1}) + C_{OPT}(\sigma'_{2}) + C_{OPT}(\sigma''_{2}) 
			= C_{OPT}(\sigma)
	\end{equation}
	%
	%
	$\sigma_{1}, \sigma'_{2}, \sigma''_{2}$
	$RA$
	$U_{1}$
	$U_{3}$
	$U_{2}$
	%
	%
	$\sigma$
	%
	\begin{equation} \label{LMA:rand_up.p5:eq.3}
		{\mathbb E}[C_{RA}(\sigma)] = c_{1} + c_{2} + c_{3}
	\end{equation}
	%
	%
	$\sigma_{1}$
	$U_{2} \setminus \{ u_{2} \}$
	$u$
	$\sigma$
	%
	\begin{equation} \label{LMA:rand_up.p5:eq.4}
		{\mathbb E}[C_{RA}(\sigma_{1})] \geq c_{2} - 2
	\end{equation}
	%
	%
	$u_{2}$
	$\sigma'_{2}$
	$\sigma$
	%
	%
	$\sigma'_{2}$
	$u_{2}$
	%
	%
	$U_{1}$
	\begin{equation} \label{LMA:rand_up.p5:eq.5}
		{\mathbb E}[C_{RA}(\sigma'_{2})] \geq c_{1} + 1
	\end{equation}
	%
	\begin{equation} \label{LMA:rand_up.p5:eq.6}
		{\mathbb E}[C_{RA}(\sigma''_{2})] \geq c_{3} + 1
	\end{equation}
	%
	%
	\[
		{\mathbb E}[C_{RA}(\sigma_{1})] + {\mathbb E}[C_{RA}(\sigma'_{2})] + {\mathbb E}[C_{RA}(\sigma''_{2})] 
			\geq c_{1} + c_{2} + c_{3} 
			= {\mathbb E}[C_{RA}(\sigma)]
	\]
	%
	(3)
	\fi
	\ifnum \count11 > 0
	%
	%
	We prove this lemma similarly to the proofs of the above lemmas. 
	Suppose that an input $\sigma$ satisfies the conditions (i) and (ii) in the statement. 
	To prove the lemma, 
	we will prove that 
	we can construct an input $\sigma''$ satisfying the following three properties from $\sigma$: 
	(1) 
	the number of good vertex triplets in the graph given in $\sigma''$ is less than that in $\sigma$. 
	(2)
	(P1), (P2), (P3) and (P4) hold, 
	and 
	(3) 
	$	\frac{ {\mathbb E}[C_{RA}(\sigma)] }{ C_{OPT}(\sigma) } 
			\leq \frac{ {\mathbb E}[C_{RA}(\sigma'')] }{ C_{OPT}(\sigma'') }$. 
	We can construct an input satisfying the properties (a), (b) and (c) in the statement recursively by proving the existence of $\sigma''$. 
	The graph in $\sigma$ contains a good vertex triplet of $(u_{1}, u_{2}, u_{3})$. 
	That is, 
	both $u_{1}$ and $u_{3}$ are adjacent to $u_{2}$, 
	and 
	$OPT$ selects $u_{1}$ and $u_{3}$. 
	Without loss of generality, 
	suppose that $u_{1}$ is revealed before $u_{3}$. 
	If $u_{2}$ is revealed after $u_{1}$, 
	then the edge $(u_{1}, u_{2})$ is free, 
	which contradicts that $\sigma$ satisfies (P4) by the condition (ii) 
	because $OPT$ selects $u_{1}$. 
	Thus, 
	$u_{2}$ is revealed before $u_{1}$. 
	Now, 
	define 
	$\sigma_{1} = f_{1}(\sigma, u_{2})$, 
	$\sigma_{2} = f_{2}(\sigma, u_{2})$, 
	$\sigma'_{2} = f_{1}(\sigma_{2}, u_{3})$, and 
	$\sigma''_{2} = f_{1}(\sigma_{2}, u_{1})$ 
	(see Fig.~\ref{fig:fig_p5}). 
	Let $u (\ne u_{1}, u_{3})$ be the vertex adjacent to $u_{2}$. 
	Since $u$ is dominated by a vertex except for $u_{1}$ and $u_{3}$, 
	either $(u_{2}, u)$ or $(u, u_{2})$ is free. 
	Thus, 
	\begin{equation} \label{LMA:rand_up.p5:eq.1}
		D_{OPT}(\sigma_{1}) \cup D_{OPT}(\sigma_{2}) = D_{OPT}(\sigma) 
	\end{equation}
	by Lemma~\ref{LMA:rand_up.input_div_opt}. 
	Furthermore, 
	$(u_{2}, u_{1})$ and $(u_{2}, u_{3})$ are free 
	because $OPT$ selects both $u_{1}$ and $u_{3}$. 
	Hence, 
	$D_{OPT}(\sigma'_{2}) \cup D_{OPT}(f_{2}(\sigma_{2}, u_{3})) = D_{OPT}(\sigma_{2})$
	and 
	$D_{OPT}(\sigma''_{2}) \cup D_{OPT}(f_{2}(\sigma_{2}, u_{1})) = D_{OPT}(\sigma_{2})$
	by Lemma~\ref{LMA:rand_up.input_div_opt}. 
	Since $OPT$ does not select $u_{2}$ in $\sigma$ by definition, 
	$D_{OPT}(\sigma'_{2}) \cap D_{OPT}(\sigma''_{2}) = \varnothing$. 
	Thus, 
	$D_{OPT}(f_{2}(\sigma_{2}, u_{3})) = D_{OPT}(\sigma''_{2})$. 
	By these equalities, 
	\[
		D_{OPT}(\sigma'_{2}) \cup D_{OPT}(\sigma''_{2}) = D_{OPT}(\sigma_{2}). 
	\]
	By this equality and Eq.~(\ref{LMA:rand_up.p5:eq.1}), 
	\begin{equation} \label{LMA:rand_up.p5:eq.0}
		D_{OPT}(\sigma_{1}) \cup D_{OPT}(\sigma'_{2}) \cup D_{OPT}(\sigma''_{2}) 
			= D_{OPT}(\sigma), 
	\end{equation}
	which indicates that 
	the above input transformations do not affect free edges except for $(u_{2}, u_{1})$, $(u_{2}, u_{3})$ and $(u_{2}, u)$ (or $(u, u_{2})$). 
	Also, 
	the number of good vertex triplets in $\sigma_{1}$ ($\sigma'_{2}, \sigma''_{2}$) is less than that in $\sigma$, 
	and hence these inputs satisfy the property (1).  
	Since $\sigma$ satisfies the properties from (P1) through (P4) by the condition (ii), 
	$\sigma_{1}$, $\sigma'_{2}$ and $\sigma''_{2}$ also satisfy these properties,  
	that is, 
	the property (2). 
	Finally, we show that at least one input of $\sigma_{1}$, $\sigma'_{2}$ and $\sigma''_{2}$ satisfies the property (3). 
	By Eq.~(\ref{LMA:rand_up.p5:eq.0}), 
	\begin{equation} \label{LMA:rand_up.p5:eq.2}
		C_{OPT}(\sigma_{1}) + C_{OPT}(\sigma'_{2}) + C_{OPT}(\sigma''_{2}) 
			= C_{OPT}(\sigma). 
	\end{equation}
	We evaluate the expected costs for $\sigma_{1}, \sigma'_{2}$ and $\sigma''_{2}$ 
	using Lemma~\ref{LMA:rand_up.vtx_cost}. 
	For $i = 1$ and $3$, 
	let $U_{i}$ be the set of vertices consisting of $u_{i}$ and all the descendants of $u_{i}$. 
	Let $U_{2}$ be the set of all the vertices except for $U_{1} \cup U_{3}$ in $\sigma$. 
	For each $j \in \{ 1, 2, 3 \}$, 
	let $c_{j}$ be the expected cost of $RA$ for $U_{j}$ in $\sigma$. 
	By definition,  
	\begin{equation} \label{LMA:rand_up.p5:eq.3}
		{\mathbb E}[C_{RA}(\sigma)] = c_{1} + c_{2} + c_{3}. 
	\end{equation}
	Note that 
	the set of all the vertices of the graph given in $\sigma_{1}$ is $U_{2} \setminus \{ u_{2} \}$. 
	The degrees of these vertices except for $u$ do not change before and after the input transformation. 
	Also, 
	the expected cost for $u_{2}$ ($u$) in $\sigma$ is one (at most one). 
	Thus, 
	\begin{equation} \label{LMA:rand_up.p5:eq.4}
		{\mathbb E}[C_{RA}(\sigma_{1})] \geq c_{2} - 2. 
	\end{equation}
	The degrees of all the vertices except for $u_{2}$ in $\sigma'_{2}$ are equal to those in $\sigma$, and 
	hence 
	their expected costs do not change. 
	Since $u_{2}$ is the first revealed vertex in $\sigma_{2}$,  
	its expected cost does not decrease. 
	Since $U_{1}$ does not contain $u_{2}$, 
	\begin{equation} \label{LMA:rand_up.p5:eq.5}
		{\mathbb E}[C_{RA}(\sigma'_{2})] \geq c_{1} + 1. 
	\end{equation}
	In the same way, 
	\begin{equation} \label{LMA:rand_up.p5:eq.6}
		{\mathbb E}[C_{RA}(\sigma''_{2})] \geq c_{3} + 1. 
	\end{equation}
	By Eqs.~(\ref{LMA:rand_up.p5:eq.3}), (\ref{LMA:rand_up.p5:eq.4}), (\ref{LMA:rand_up.p5:eq.5}) and (\ref{LMA:rand_up.p5:eq.6}), 
	\[
		{\mathbb E}[C_{RA}(\sigma_{1})] + {\mathbb E}[C_{RA}(\sigma'_{2})] + {\mathbb E}[C_{RA}(\sigma''_{2})] 
			\geq c_{1} + c_{2} + c_{3} 
			= {\mathbb E}[C_{RA}(\sigma)]. 
	\]
	The property (3) holds 
	by this inequality and Eq.~(\ref{LMA:rand_up.p5:eq.2}). 
	\fi
\end{proof}
\ifnum \count10 > 0
%
%

%

%
\fi
\ifnum \count11 > 0
%
%

%
\fi
\ifnum \count12 > 0
\begin{figure*}[ht]
	 \begin{center}
	  \includegraphics[width=130mm]{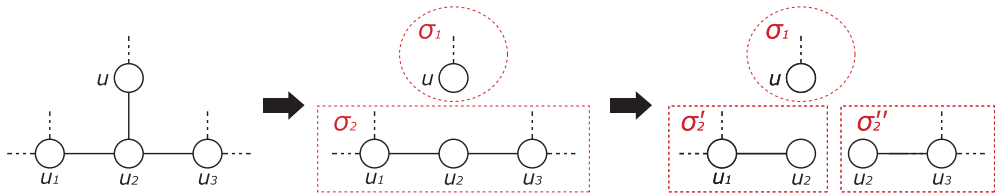}
	 \end{center}
	 \caption{
\ifnum \count10 > 0
%
%
$\sigma$
$\sigma'_{2}$
\fi
\ifnum \count11 > 0
%
%
Construction of $\sigma_{1}, \sigma'_{2}$ and $\sigma''_{2}$ from $\sigma$. 
Note that both $\sigma'_{2}$ and $\sigma''_{2}$ contain $u_{2}$. 
\fi
			}
	\label{fig:fig_p5}
\end{figure*}
\fi
\begin{LMA} \label{LMA:rand_up.degnot2}
	\ifnum \count10 > 0
	%
	%
	(P1)
	%
	
	%
	\fi
	\ifnum \count11 > 0
	%
	%
	Consider the graph in an input satisfying the properties from (P1) to (P6) inclusive. 
	Then, 
	the degree of any vertex is one or three. 
	That is, 
	(P7) holds. 
	\fi
\end{LMA}
\begin{proof}
	\ifnum \count10 > 0
	%
	%
	(P1)
	(P2)
	$v$
	$v$
	%
	%
	%
	
	%
	$u$
	(P6)
	$(v, u)$
	(P3)
	$OPT$
	$OPT$
	$v' (\ne u)$
	$OPT$
	$(v, u)$
	$OPT$
	%
	%
	%
	$v = v_{1}$
	$(v, v')$
	%
	%
	$v \ne v_{1}$
	$(v', v)$
	$v$
	%
	(P1)
	%
	%
	$v$
	\fi
	\ifnum \count11 > 0
	%
	%
	Consider an input satisfying the properties from (P1) to (P6) inclusive. 
	It suffices to prove that 
	the degree of a vertex $v$ in the input is not two 
	because the degree of $v$ is at most three by (P2). 
	Then, 
	we prove it by contradiction, and 
	assume that the degree of $v$ is two. 
	Let $u$ be a vertex arriving at $v$. 
	The edge $(v, u)$ is not free by (P6). 
	$OPT$ does not select $v$ 
	because the degree of any vertex selected by $OPT$ is three by (P3). 
	Let $v' (\ne u)$ be the other vertex adjacent to $v$. 
	If $OPT$ selects $v'$, 
	then $(v, u)$ is free. 
	Thus, $OPT$ does not select $v'$, 
	which means that 
	$u$ dominates $v$. 
	If $v = v_{1}$, 
	then the edge $(v, v')$ is free. 
	This edge is not fixed, 
	which contradicts that the input satisfies (P1). 
	If $v \ne v_{1}$, 
	then the edge $(v', v)$ is free. 
	This edge is not fixed 
	since ${deg}(v) = 2$ by definition, 
	which contradicts (P1). 
	By the above argument, 
	we have shown that the degree of $v$ is not two. 
	\fi
\end{proof}
\ifnum \count10 > 0
%
%

%

%
\fi
\ifnum \count11 > 0
%
%

%
\fi
%
\begin{LMA} \label{LMA:rand_up.p6}
	\ifnum \count10 > 0
	%
	%
	(i) 
	${deg}(v) = 2$
	(ii) 
	(P1), (P2), (P3), (P4), (P5)
	%
	
	%
	%
	(a) 
	${deg}(v) = 3$
	(b)
	(P1),(P2),(P3), (P4), (P5)
	(c) 
	$\frac{ {\mathbb E}[C_{RA}(\sigma)] }{ C_{OPT}(\sigma) } 
	\leq \frac{ {\mathbb E}[C_{RA}(\sigma')] }{ C_{OPT}(\sigma') }$
	%
	
	%
	\fi
	\ifnum \count11 > 0
	%
	%
	Suppose that an input $\sigma$ satisfies the following conditions: 
	(i) 
	there exists at least one free edge $(v, u)$ such that ${deg}(v) = 2$, 
	and 
	(ii) 
	the properties from (P1) to (P5) inclusive hold. 
	Then, 
	there exists an input $\sigma'$ such that 
	(a) 
	for any free edge $(v, u)$, 
	${deg}(v) = 3$, 
	(b)
	the properties from (P1) to (P5) inclusive hold, 
	and 
	(c) 
	$\frac{ {\mathbb E}[C_{RA}(\sigma)] }{ C_{OPT}(\sigma) } 
	\leq \frac{ {\mathbb E}[C_{RA}(\sigma')] }{ C_{OPT}(\sigma') }$. 
	\fi
\end{LMA}
\begin{proof}
	\ifnum \count10 > 0
	%
	%
	%
	(i)
	$\sigma$
	$v$
	%
	$v$
	$v$
	$(v, u_{1})$
	%
	%
	%
	(1) 
	$\sigma''$
	$(v, u_{1})$
	$\sigma$
	(2)
	(P1),(P2),(P3),(P4),(P5)
	(3) 
	$	\frac{ {\mathbb E}[C_{RA}(\sigma)] }{ C_{OPT}(\sigma) } 
			\leq \frac{ {\mathbb E}[C_{RA}(\sigma'')] }{ C_{OPT}(\sigma'') }
	$
	%
	%
	%
	
	%
	$\sigma_{1} = f_{1}(\sigma, u_{1})$
	$\frac{ {\mathbb E}[C_{RA}(\sigma_{1})] }{ C_{OPT}(\sigma_{1}) } 
			\geq \frac{ {\mathbb E}[C_{RA}(\sigma)] }{ C_{OPT}(\sigma) }$
	%
	$\sigma_{1}$
	%
	%
	$D_{OPT}(\sigma_{1})$
	$\sigma_{1}$
	$(v, u_{1})$
	$\sigma$
	(1)
	%
	%
	$\sigma_{1}$
	(2)
	\begin{equation} \label{LMA:rand_up.p6:eq.2}
		\frac{ {\mathbb E}[C_{RA}(\sigma_{1})] }{ C_{OPT}(\sigma_{1}) } 
			< \frac{ {\mathbb E}[C_{RA}(\sigma)] }{ C_{OPT}(\sigma) } 
	\end{equation}
	%
	%
	$\sigma_{2} = f_{2}(\sigma, u_{1})$
	$\sigma'_{2}$
	$\sigma'_{2} = (U', F', R')$
	$\sigma_{2} = (U, F, R)$
	$U = \{ u_{i} \mid i \in [1, k] \}$
	$R = (u_{1}, \ldots, u_{k})$
	%
	$U' = \{ u'_{i} \mid i \in [1, k] \}$
	$R' = (u'_{1}, \ldots, u'_{k})$
	$F' = \{ \{ u'_{i}, u'_{j} \} \mid \{ u_{i}, u_{j} \} \in F \}$
	%
	%
	$\sigma_{3} = f_{3}(\sigma, v, \sigma'_{2})$
	%
	%
	$\sigma_{3}$
	$\sigma_{3}$
	$\sigma$
	$u$
	$u'$
	$u$
	%
	$\sigma$
	%
	$v$
	$U$
	%
	$\sigma$
	%
	%
	$U$
	$\sigma$
	%
	%
	$\sigma$
	%
	(P3)
	%
	%
	$U$
	%
	%
	$\sigma_{3}$
	$U$
	$\sigma$
	$U'$
	$OPT$
	$OPT$
	%
	$\sigma_{3}$
	$v$
	$U''$
	$U$
	%
	%
	%
	$\sigma_{3}$
	$U \cup U''$
	$\sigma$
	%
	$\sigma_{3}$
	$U''$
	$\sigma$
	%
	%
	$\sigma_{3}$
	(1)
	%
	(2)
	%
	
	%
	%
	$\sigma$
	$U'' \setminus \{ v \}$
	$U'' \setminus \{ v \}$
	%
	%
	$\sigma$
	$\sigma_{3}$
	%
	$\sigma$
	(P3)
	%
	$\sigma_{3}$
	%
	%
	\begin{equation} \label{LMA:rand_up.p6:eq.3}
		\frac{ {\mathbb E}[C_{RA}(\sigma)] }{ C_{OPT}(\sigma) } 
			= \frac{ a + 1/2 + b }{ a' + b' }
	\end{equation}
	\begin{equation} \label{LMA:rand_up.p6:eq.4}
		\frac{ {\mathbb E}[C_{RA}(\sigma_{3})] }{ C_{OPT}(\sigma_{3}) } 
			= \frac{ a + 1 + 2 b }{ a' + 2 b' }
	\end{equation}
	%
	%
	$\sigma_{1}$
	$U'' \setminus \{ v \}$
	$\sigma$
	${\mathbb E}[C_{RA}(\sigma_{1})] \geq a$
	%
	%
	$(v, u_{1})$
	$C_{OPT}(\sigma_{1}) = a'$
	%
	%
	$\frac{ a }{ a' } < \frac{ a + 1/2 + b }{ a' + b' }$
	$\frac{ a }{ a' } < \frac{1/2 + b }{ b' }$
	%
	\[
		\frac{ {\mathbb E}[C_{RA}(\sigma)] }{ C_{OPT}(\sigma) } 
			< \frac{ {\mathbb E}[C_{RA}(\sigma_{3})] }{ C_{OPT}(\sigma_{3}) } 
	\]
	$\sigma_{3}$
	\fi
	\ifnum \count11 > 0
	%
	%
	We prove this lemma similarly to the proofs of the above lemmas. 
	Suppose that an input $\sigma$ satisfies the conditions (i) and (ii) in the statement. 
	Let $(v, u_{1})$ be a free edge of the graph in $\sigma$ such that the degree of $v$ is two and 
	there does not exist a descendant $v'$ of $v$ such that 
	$(v', v'')$ is free, in which a vertex $v''$ arrives at $v'$ and 
	the degree of $v'$ is two. 
	To prove the lemma, 
	we will prove that 
	we can construct an input $\sigma''$ satisfying the following three properties from $\sigma$: 
	(1) 
	the number of such free edges $(v, u_{1})$ of the graph given in $\sigma''$ is less than that in $\sigma$. 
	(2)
	the properties from (P1) to (P5) inclusive hold, 
	and 
	(3) 
	$	\frac{ {\mathbb E}[C_{RA}(\sigma)] }{ C_{OPT}(\sigma) } 
			\leq \frac{ {\mathbb E}[C_{RA}(\sigma'')] }{ C_{OPT}(\sigma'') }$. 
	We can construct an input satisfying the properties (a), (b) and (c) in the statement recursively by proving the existence of $\sigma''$. 
	Define $\sigma_{1} = f_{1}(\sigma, u_{1})$. 
	We fist consider the case in which 
	$\frac{ {\mathbb E}[C_{RA}(\sigma_{1})] }{ C_{OPT}(\sigma_{1}) } 
			\geq \frac{ {\mathbb E}[C_{RA}(\sigma)] }{ C_{OPT}(\sigma) }$. 
	Clearly $\sigma_{1}$ satisfies the above property (3). 
	$D_{OPT}(\sigma_{1})$ can be a subset of $D_{OPT}(\sigma)$ by Lemma~\ref{LMA:rand_up.input_div_opt}, and thus 
	the number of such edges $(v, u_{1})$ in $\sigma_{1}$ is less than that in $\sigma$ by one, 
	which indicates that the property (1) is satisfied. 
	Furthermore, 
	$\sigma_{1}$ satisfies the properties from (P1) through (P5) by the condition (ii), 
	and hence satisfies the property (2). 
	Next we consider the case in which
	\begin{equation} \label{LMA:rand_up.p6:eq.2}
		\frac{ {\mathbb E}[C_{RA}(\sigma_{1})] }{ C_{OPT}(\sigma_{1}) } 
			< \frac{ {\mathbb E}[C_{RA}(\sigma)] }{ C_{OPT}(\sigma) }. 
	\end{equation}
	Define $\sigma_{2} = f_{2}(\sigma, u_{1})$, and 
	let $\sigma'_{2}$ be a copy of $\sigma_{2}$. 
	Specifically,  
	define $\sigma'_{2} = (U', F', R')$, 
	in which 
	$\sigma_{2} = (U, F, R)$, 
	$k$ is the number of vertices in $U$, 
	$U = \{ u_{i} \mid i \in [1, k] \}$, 
	$R = (u_{1}, \ldots, u_{k})$, 
	$U' = \{ u'_{i} \mid i \in [1, k] \}$, 
	$R' = (u'_{1}, \ldots, u'_{k})$ and 
	$F' = \{ \{ u'_{i}, u'_{j} \} \mid \{ u_{i}, u_{j} \} \in F \}$. 
	Then, 
	define $\sigma_{3} = f_{3}(\sigma, v, \sigma'_{2})$. 
	First, 
	we discuss vertices selected by $OPT$ of the graph in $\sigma_{3}$ 
	to prove that $\sigma_{3}$ satisfies the properties (1) and (2). 
	Suppose that for vertices $u, u', u'' \in U$, 
	$u$ is adjacent to $u'$, 
	$u'$ is adjacent to $u''$, and 
	$OPT$ selects $u$ in $\sigma$. 
	Then, 
	$OPT$ does not select $u'$ in $\sigma$ 
	because $\sigma$ satisfies (P4) by the condition (ii). 
	Also, 	
	the degree of any vertex in $U$ is one or three 
	by the definition of $v$ together with Lemma~\ref{LMA:rand_up.degnot2}. 
	Hence, 
	$OPT$ does not select $u''$ in $\sigma$ 
	because $\sigma$ satisfies (P5) by the condition (ii). 
	The length of a path between two different vertices selected by $OPT$ in $U$ is at least three 
	when dealing with $\sigma$ by these facts. 
	Thus, 
	any vertex in $U$ is dominated by only one vertex in $\sigma$. 
	Moreover, 
	a vertex selected by $OPT$ dominates four vertices by (P3). 
	Thus, 
	the way for $OPT$ to select vertices in $U$ is unique. 
	That is, 
	the vertices selected by $OPT$ in $U$ in $\sigma_{3}$ are the same as those in $\sigma$ (called Fact~(a)). 
	There exists $OPT$ such that 
	$OPT$ selects $u_{i}$ if and only if $OPT$ selects $u'_{i}$ by the definition of $U'$, and we use such $OPT$ (Fact~(b)). 
	Thus, 
	$v$ is not dominated by only $u'_{1}$ in $\sigma_{3}$ 
	(Fact~(c)). 
	Let $U''$ be the set of vertices except for $U$ of the graph in $\sigma$, 
	that is, 
	all the vertices in $f_{1}(\sigma, u_{1})$. 
	The set of vertices selected by $OPT$ in $U \cup U''$ for $\sigma_{3}$ is a dominating set of the graph in $\sigma$ by Fact~(c), 
	which together with Fact~(a) and the optimality of $OPT$ for $\sigma$ shows that the vertices selected by $OPT$ in $U''$ for $\sigma_{3}$ can be the same as those for $\sigma$ (Fact~(d)). 
	Hence, 
	$\sigma_{3}$ satisfies the property (1) 
	because the degree of $v$ in $\sigma_{3}$ is three. 
	Also, 
	$\sigma_{3}$ satisfies the properties from (P1) through (P5) 
	because $\sigma$ satisfies the properties by the condition (ii) 
	and $\sigma_{3}$ satisfies the property (2). 
	Finally, we discuss the property (3). 
	Let $a$ and $b$ be the expected costs of $RA$ for $U'' \setminus \{ v \}$ and $U$, respectively, in $\sigma$. 
	Let $a'$ and $b'$ be the costs of $OPT$ for $U'' \setminus \{ v \}$ and $U$, respectively, in $\sigma$. 
	Since the degree of $v$ in $\sigma$ ($\sigma_{3}$) is two (three), 
	the expected cost for $v$ in $\sigma$ ($\sigma_{3}$) is $1/2$ (one) by Lemma~\ref{LMA:rand_up.vtx_cost}. 
	$OPT$ does not select $v$ in $\sigma$ by (P3)
	because the degree of $v$ is two, 
	and hence by Fact~(d), 
	$OPT$ does not select $v$ in $\sigma_{3}$, either. 
	By these costs, 
	\begin{equation} \label{LMA:rand_up.p6:eq.3}
		\frac{ {\mathbb E}[C_{RA}(\sigma)] }{ C_{OPT}(\sigma) } 
			= \frac{ a + 1/2 + b }{ a' + b' },  
	\end{equation}
	and 
	\begin{equation} \label{LMA:rand_up.p6:eq.4}
		\frac{ {\mathbb E}[C_{RA}(\sigma_{3})] }{ C_{OPT}(\sigma_{3}) } 
			= \frac{ a + 1 + 2 b }{ a' + 2 b' }
	\end{equation}
	by Fact~(b). 
	The expected cost of $RA$ for $U'' \setminus \{ v \}$ in $\sigma_{1}$ is equal to that in $\sigma$ by Lemma~\ref{LMA:rand_up.vtx_cost}, 
	and thus 
	${\mathbb E}[C_{RA}(\sigma_{1})] \geq a$. 
	In addition, 
	$C_{OPT}(\sigma_{1}) = a'$ 
	by Lemma~\ref{LMA:rand_up.input_div_opt} 
	since $(v, u_{1})$ is free. 
	These costs for $\sigma_{1}$ and Eqs.~(\ref{LMA:rand_up.p6:eq.2}) and (\ref{LMA:rand_up.p6:eq.3}) show that 
	$\frac{ a }{ a' } < \frac{ a + 1/2 + b }{ a' + b' }$, 
	which is rearranged to 
	$\frac{ a }{ a' } < \frac{1/2 + b }{ b' }$. 
	This inequality together with Eqs.~(\ref{LMA:rand_up.p6:eq.3}) and  (\ref{LMA:rand_up.p6:eq.4}) yields 
	\[
		\frac{ {\mathbb E}[C_{RA}(\sigma)] }{ C_{OPT}(\sigma) } 
			< \frac{ {\mathbb E}[C_{RA}(\sigma_{3})] }{ C_{OPT}(\sigma_{3}) },  
	\]
	which means that $\sigma_{3}$ satisfies the property (3). 
	\fi
\end{proof}
\ifnum \count10 > 0
%
%
%
%
(P1)
\fi
\ifnum \count11 > 0
%
%
Now we can show Lemma~\ref{LMA:rand_up.basic_p1_7}
using Lemmas~\ref{LMA:rand_up.degnot2} and~\ref{LMA:rand_up.p6}. 
In the next section, 
we analyze only inputs satisfying the properties from (P1) through (P7). 
\fi
%
\subsection{Analysis of $RA$}\label{sec:analysis}
\ifnum \count10 > 0
%
%
$RA$
%
%
%
%
%
$N(v)$
$N(v) = \{ u \mid \{ v, u\} \in E \}$
%

%
%
\fi
\ifnum \count11 > 0
%
%
We assign a positive integer to each vertex of a given tree according to the below routine. 
We call the set of all the vertices with the same assigned value a {\em block}. 
%
All the vertices in a block are on a path of at most length three. 
%
We obtain the competitive ratio of $RA$ by evaluating the costs of $RA$ and $OPT$ for each block. 
For a vertex $v$, 
$N(v)$ denotes the set of vertices adjacent to $v$. 
That is, 
$N(v) = \{ u \mid \{ v, u\} \in E \}$, 
in which 
$E$ is the set of all the edges of a given graph. 
\fi
\ifnum \count10 > 0
%
%
\noindent\vspace{-1mm}\rule{\textwidth}{0.5mm} 
\vspace{-3mm}
{{\sc BlockRoutine}}\\
\rule{\textwidth}{0.1mm}
%
	%
	%
	{\bf\boldmath Step~1:} 
		%
		$\ell := 0$, 
		$U := \{ v_{i} \mid i = 1, \ldots, n \}$, 
		%
%
	{\bf\boldmath Step~2:}
		%
		$\ell := \ell + 1$. 
		$U = \varnothing$
		$i_{1} := \min \{ i \mid v_{i} \in U \}$
		$U := U \setminus \{ v_{i_{1}} \}$
%
%
	{\bf\boldmath Step~3:}
		%
%
		$U \cap N(v_{i_{1}})= \varnothing$
		$i_{2} := \min \{ i \mid v_{i} \in U \cap N(v_{i_{1}}) \}$
		$U := U \setminus \{ v_{i_{2}} \}$
		%
%
		{\bf\boldmath Step~4:}
		%
%
		$U' := U \cap \{ N(v_{i_{1}}) \cup N(v_{i_{2}}) \}$
		$U' = \varnothing$
		$i_{3} := \min \{ i \mid v_{i} \in U' \}$
		$U := U \setminus \{ v_{i_{3}} \}$
		Step~2
\noindent\vspace{-1mm}\rule{\textwidth}{0.5mm} 
\fi
\ifnum \count11 > 0
%
%
\noindent\vspace{-1mm}\rule{\textwidth}{0.5mm} 
\vspace{-3mm}
{{\sc BlockRoutine}}\\
\rule{\textwidth}{0.1mm}
%
	%
	%
	{\bf\boldmath Step~1:} 
		%
		$\ell := 0$ and 
		$U := \{ v_{i} \mid i = 1, \ldots, n \}$, 
		in which $n$ is the number of all the vertices of a given graph. \\
	{\bf\boldmath Step~2:}
		%
		$\ell := \ell + 1$. 
		If $U = \varnothing$, then finish. 
		Otherwise, 
		$i_{1} := \min \{ i \mid v_{i} \in U \}$, 
		$U := U \setminus \{ v_{i_{1}} \}$ and 
		assign $\ell$ to the vertex $v_{i_{1}}$. \\
	{\bf\boldmath Step~3:}
		%
%
		If $U \cap N(v_{i_{1}})= \varnothing$, then go to Step~2. 
		Otherwise, 
		$i_{2} := \min \{ i \mid v_{i} \in U \cap N(v_{i_{1}}) \}$, 
		$U := U \setminus \{ v_{i_{2}} \}$ and 
		assign $\ell$ to the vertex $v_{i_{2}}$. \\
		{\bf\boldmath Step~4:}
		%
%
		$U' := U \cap \{ N(v_{i_{1}}) \cup N(v_{i_{2}}) \}$. 
		If $U' = \varnothing$, then go to Step~2. 
		Otherwise, 
		$i_{3} := \min \{ i \mid v_{i} \in U' \}$, 
		$U := U \setminus \{ v_{i_{3}} \}$, 
		assign $\ell$ to the vertex $v_{i_{3}}$ and 
		go to Step~2. \\
\noindent\vspace{-1mm}\rule{\textwidth}{0.5mm} 
\fi
\begin{LMA} \label{LMA:rand_up.n4}
	\ifnum \count10 > 0
	%
	%
	%
	
	%
	\fi
	\ifnum \count11 > 0
	%
	%
	The number of all the vertices of a given tree is at least four. 
	\fi
\end{LMA}
\begin{proof}
	\ifnum \count10 > 0
	%
	%
	$OPT$
	(P3)
	%
	
	%
	\fi
	\ifnum \count11 > 0
	%
	%
	Clearly $OPT$ selects at least one vertex. 
	The degree of a vertex selected by $OPT$ is three by (P3), 
	and thus the number of all the vertices of a given tree is at least four. 
	\fi
\end{proof}
\ifnum \count10 > 0
%
%
{\sc BlockRoutine}
1
%
%
%
%
{\em \boldmath $B_{1}$-
${deg}(u_{1}) = 1$
%
%
3
$u_{1}$
%
%
{\em \boldmath $B_{2}$-
%
%
{\em \boldmath $B_{3}$-
%
%
{\em \boldmath $B_{4}$-
%
%

%
%

%
\fi
\ifnum \count11 > 0
%
%
By the definition of {\sc BlockRoutine}, 
this lemma leads to the fact that 
at least one vertex in a block is adjacent to a vertex in another block. 
Also, by (P7) in the previous section, 
the degree of a vertex is one or three, 
and hence blocks which a given graph can contain are classified into the following four categories 
(Fig.~\ref{fig:block}): 
%
%
A {\em \boldmath $B_{1}$-block} is a set consisting of one vertex $u_{1}$ 
such that ${deg}(u_{1}) = 1$. 
The following three blocks are sets consisting of three vertices $u_{1},u_{2}$ and $u_{3}$. 
Suppose that both $u_{1}$ and $u_{3}$ are adjacent to $u_{2}$. 
%
%
A {\em \boldmath $B_{2}$-block} is a set 
such that ${deg}(u_{1}) = 3$, ${deg}(u_{2}) = 3$ and ${deg}(u_{3}) = 1$, 
%
%
a {\em \boldmath $B_{3}$-block} is a set 
such that ${deg}(u_{1}) = 1$, ${deg}(u_{2}) = 3$ and ${deg}(u_{3}) = 1$, and 
%
%
a {\em \boldmath $B_{4}$-block} is a set 
such that ${deg}(u_{1}) = 3$, ${deg}(u_{2}) = 3$ and ${deg}(u_{3}) = 3$. 
%
%

%
%
\fi
\ifnum \count12 > 0
\begin{figure*}[ht]
	 \begin{center}
	  \includegraphics[width=150mm]{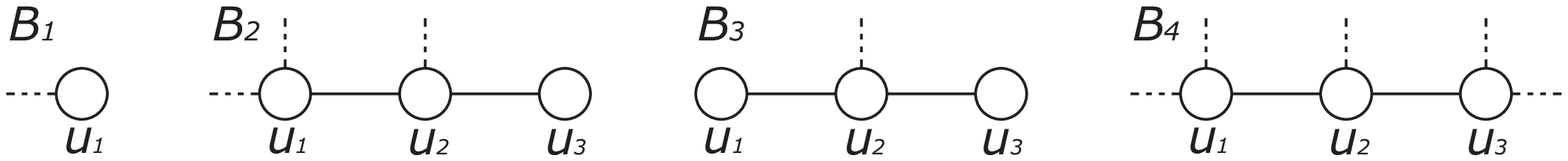}
	 \end{center}
	 \caption{Blocks}
	\label{fig:block}
\end{figure*}
\fi
\ifnum \count10 > 0
%
%
$OPT$
%
$B_{1},B_{2},B_{3},B_{4}$
$OPT$
%
%
%
%

%
$B_{1}$-
{\bf \boldmath $B_{1}^{0}$-
{\bf \boldmath $B_{1}^{1}$-
$B_{2}$-
{\bf \boldmath $B_{2}^{010}$-
{\bf \boldmath $B_{2}^{110}$-
$B_{3}$-
$OPT$
$B_{4}$-
{\bf \boldmath $B_{4}^{000}$-
{\bf \boldmath $B_{4}^{100}$-
{\bf \boldmath $B_{4}^{010}$-
{\bf \boldmath $B_{4}^{110}$-
{\bf \boldmath $B_{4}^{101}$-
{\bf \boldmath $B_{4}^{111}$-
\fi
\ifnum \count11 > 0
%
%
For each block, we discuss vertices selected by $OPT$ and 
classify $B_{1},B_{2},B_{3}$ and $B_{4}$ into the following eleven categories
(Fig.~\ref{fig:block_opt}). 
Then the next lemma shows that we only have to consider six categories. 
%
%
$u_{1},u_{2}$ and $u_{3}$ to classify $B_{i}$ are used 
in the same definitions as those of $u_{1},u_{2}$ and $u_{3}$ to define $B_{i}$. 
%

%
$B_{1}$-blocks are classified into two categories: 
A {\em \boldmath $B_{1}^{0}$-block} in which $OPT$ does not select any vertex, and 
a {\em \boldmath $B_{1}^{1}$-block} in which $OPT$ selects only $u_{1}$. 
$B_{2}$-blocks are classified into two categories: 
A {\em \boldmath $B_{2}^{010}$-block} in which $OPT$ selects only $u_{2}$, and 
a {\em \boldmath $B_{2}^{110}$-block} in which $OPT$ selects only $u_{1}$ and $u_{2}$. 
$B_{3}$-blocks are not classified. 
$OPT$ selects only $u_{2}$ in a $B_{3}$-block. 
$B_{4}$-blocks are classified into six categories: 
A {\em \boldmath $B_{4}^{000}$-block} in which $OPT$ selects no vertices, 
a {\em \boldmath $B_{4}^{100}$-block } in which $OPT$ selects only $u_{1}$, 
a {\em \boldmath $B_{4}^{010}$-block} in which $OPT$ selects only $u_{2}$, 
a {\em \boldmath $B_{4}^{110}$-block} in which $OPT$ selects only $u_{1}$ and $u_{2}$, 
a {\em \boldmath $B_{4}^{101}$-block} in which $OPT$ selects only $u_{1}$ and $u_{3}$, and 
a {\em \boldmath $B_{4}^{111}$-block} in which $OPT$ selects all the vertices. 
\fi
\ifnum \count12 > 0
\begin{figure*}[ht]
	 \begin{center}
	  \includegraphics[width=150mm]{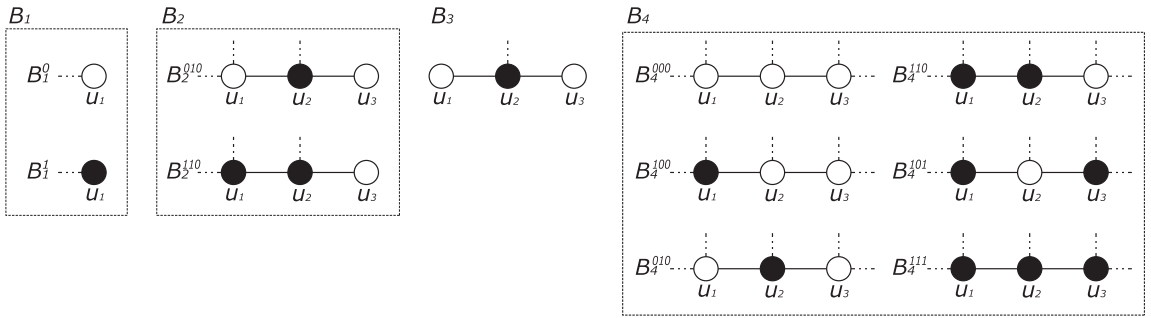}
	 \end{center}
	 \caption{
				\ifnum \count10 > 0
				%
				%
				%
				\fi
				\ifnum \count11 > 0
				%
				%
				Classified blocks. 
				$OPT$ selects black vertices. 
				\fi
			}
	\label{fig:block_opt}
\end{figure*}
\fi
\ifnum \count10 > 0
%
%

%

%
\fi
\ifnum \count11 > 0
%
%

%
\fi
%

%
\begin{LMA} \label{LMA:rand_up.spblock_opt}
	\ifnum \count10 > 0
	%
	%
	$B_{1}^{0},B_{2},B_{3},B_{4}^{000},B_{4}^{100},B_{4}^{010}$
	%
	
	%
	\fi
	\ifnum \count11 > 0
	%
	%
	An input can contain at most six kinds of blocks: 
	$B_{1}^{0}, B_{2}, B_{3}, B_{4}^{000}, B_{4}^{100}$ and $B_{4}^{010}$. 
	\fi
\end{LMA}
\begin{proof}
	\ifnum \count10 > 0
	%
	%
	$B_{1}^{1}$
	$B_{2}^{110}, B_{4}^{110}, B_{4}^{111}$
	$B_{4}^{101}$
	(P5)
	\fi
	\ifnum \count11 > 0
	%
	%
	A $B_{1}^{1}$-block does not satisfy (P3). 
	A $B_{2}^{110}$-block, a $B_{4}^{110}$-block and a $B_{4}^{111}$-block do not satisfy (P4). 
	A $B_{4}^{101}$-block composes a good vertex triplet, 
	which does not satisfy (P5). 
	\fi
\end{proof}
\ifnum \count10 > 0
%
%
$B_{1}$-
%
%
$B_{1,0}$
$B_{1,1}$
\fi
\ifnum \count11 > 0
%
%
A $B_{1}$-block consists of one vertex $v$ of degree one, and 
(4) in Lemma~\ref{LMA:rand_up.vtx_cost} shows that 
the expected cost of $RA$ for $v$ depends on the adjacent vertex $u$. 
Then, 
we classify $B_{1}$-blocks into the following two categories in terms of $RA$: 
A $B_{1,0}$-block of $v$ such that ${deg}_{v}(u) = 3$ and 
a $B_{1,1}$-block of $v$ such that ${deg}_{v}(u) \leq 2$. 
\fi
%

%
\begin{LMA} \label{LMA:rand_up.block_cost}
	\ifnum \count10 > 0
	%
	%
	%
	(i) $B_{1,0}$-
	(ii) $B_{1,1}$-
	(iii) $B_{2}$-
	(iv) $B_{3}$-
	(v) $B_{4}$-
	\fi
	\ifnum \count11 > 0
	%
	%
	Consider a block without $v_{1}$ and then 
	the expected costs of $RA$ are as follows: 
	(i) zero for a $B_{1,0}$-block, 
	(ii) $1/2$ for a $B_{1,1}$-block, 
	(iii) at most $5/2$ for a $B_{2}$-block, 
	(iv) at most $3/2$ for a $B_{3}$-block and 
	(v) at most three for a $B_{4}$-block. 
	\fi
\end{LMA}
\begin{proof}
	\ifnum \count10 > 0
	%
	%
	{\bf\boldmath $B_{1}$-block:}
	$B_{1}$-
	%
	$B_{1,0}$-
	${deg}_{u_{1}}(u) = 3$
	$B_{1,1}$-
	${deg}_{u_{1}}(u) \leq 2$
	%
	%
	
	%
	\noindent
	{\bf\boldmath $B_{2}$-
	$B_{2}$-
	${deg}(u_{1}) = {deg}(u_{2}) = 3$
	$u_{2}$
	$u_{1}$
	%
	$u_{3}$
	%
	$B_{2}$-
	\noindent
	{\bf\boldmath $B_{3}$-
	$B_{3}$-
	${deg}(u_{1}) = {deg}(u_{3}) = 1$
	%
	%
	$u_{2}$
	%
	${deg}_{u_{1}}(u_{2}) = 2$
	${deg}_{u_{3}}(u_{2}) = 3$
	%
	$u_{1}$
	%
	%
	%
	$B_{3}$-
	\noindent
	{\bf\boldmath $B_{4}$-
	%
	$B_{4}$-
	%
	$B_{4}$-
	\fi
	\ifnum \count11 > 0
	%
	%
	Suppose that a $B_{1}$-block consists of a vertex $u_{1}$, 
	and $u_{1}$ is adjacent to a vertex $u$. 
	By definition, 
	for a $B_{1,0}$-block, 
	${deg}_{u_{1}}(u) = 3$ holds 
	while for a $B_{1,1}$-block, 
	${deg}_{u_{1}}(u) \leq 2$ holds. 
	By (4) in Lemma~\ref{LMA:rand_up.vtx_cost}, 
	the expected costs for them are zero and $1/2$, respectively
	(Fig.~\ref{fig:block_cost}). 
	\noindent
	{\bf\boldmath $B_{2}$-block:}
	Suppose that a $B_{2}$-block consists of vertices $u_{1},u_{2}$ and $u_{3}$, 
	${deg}(u_{1}) = {deg}(u_{2}) = 3$, ${deg}(u_{3}) = 1$, and 
	$u_{2}$ is adjacent to both $u_{1}$ and $u_{3}$. 
	The expected cost for $u_{1}$ ($u_{2}$) is clearly at most one. 
	Also, 
	by (4) in Lemma~\ref{LMA:rand_up.vtx_cost}, 
	the expected cost for $u_{3}$ is at most $1/2$. 
	Thus, 
	the expected cost for a $B_{2}$-block is at most $5/2$. 
	\noindent
	{\bf\boldmath $B_{3}$-block:}
	Suppose that a $B_{3}$-block consists of vertices $u_{1},u_{2}$ and $u_{3}$, 
	${deg}(u_{1}) = {deg}(u_{3}) = 1$ and ${deg}(u_{2}) = 3$. 
	Without loss of generality, 
	suppose that $u_{1}$ is revealed before $u_{3}$. 
	Since this block does not contain $v_{1}$ by the assumption of the lemma, 
	$u_{2}$ is revealed before $u_{1}$. 
	Hence, 
	${deg}_{u_{1}}(u_{2}) = 2$
	and 
	${deg}_{u_{3}}(u_{2}) = 3$. 
	By (4-1) and (4-2) in Lemma~\ref{LMA:rand_up.vtx_cost}, 
	the expected costs for $u_{1}$ and $u_{3}$ are $1/2$ and zero, respectively. 
	Also, 
	the expected cost for $u_{2}$ is clearly at most one. 
	By the above argument, 
	the expected cost for a $B_{3}$-block is at most $3/2$. 
	\noindent
	{\bf\boldmath $B_{4}$-block:}
	A $B_{4}$-block contains three vertices. 
	The expected cost for each vertex is at most one and 
	thus the expected cost for a $B_{4}$-block is at most three. 
	\fi
\end{proof}

\ifnum \count12 > 0
\begin{figure*}[ht]
	 \begin{center}
	  \includegraphics[width=150mm]{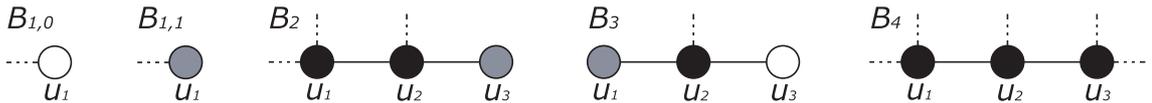}
	 \end{center}
	 \caption{
\ifnum \count10 > 0
%
%
%
\fi
\ifnum \count11 > 0
%
%
Costs for blocks without the first revealed vertex. 
Neither $A$ nor $B$ selects white vertices, 
both $A$ and $B$ select black vertices, and 
only the one of $A$ and $B$ selects gray vertices. 
%
\fi
			}
	\label{fig:block_cost}
\end{figure*}
\fi
\ifnum \count10 > 0
%
%
$B_{1}$-
\fi
\ifnum \count11 > 0
%
%
Next, 
we evaluate the expected cost of $RA$ for each block with $v_{1}$. 
Since the number of all the vertices of a given graph is at least four, 
no $B_{1}$-block contains $v_{1}$ by the definition of {\sc BlockRoutine}. 
\fi
%

%
\begin{LMA} \label{LMA:rand_up.block_cost_init}
	\ifnum \count10 > 0
	%
	%
	%
	(i) $B_{2}$-
	(ii) $B_{3}$-
	(iii) $B_{4}$-
	\fi
	\ifnum \count11 > 0
	%
	%
	Consider a block with $v_{1}$ 
	and then the expected costs of $RA$ are as follows: 
	(i) at most three for a $B_{2}$-block, 
	(ii) at most $5/2$ for a $B_{3}$-block and 
	(iii) at most three for a $B_{4}$-block. 
	\fi
\end{LMA}
\begin{proof}
	\ifnum \count10 > 0
	%
	%
	\noindent
	{\bf\boldmath $B_{2}$-
	$B_{2}$
	1
	%
	%
	%
	
	%
	\noindent
	{\bf\boldmath $B_{3}$-
	$B_{3}$-
	${deg}(u_{1}) = {deg}(u_{3}) = 1$
	%
	%
	%
	$u_{1}$
	%
	%
	${deg}_{u_{3}}(u_{2}) = 2$
	%
	$u_{3}$
	%
	$B_{3}$-
	$u_{2} = v_{1}$
	$u_{2}$
	%
	$u_{2}$
	$u_{3}$
	%
	${deg}_{u_{1}}(u_{2}) = 1$
	${deg}_{u_{3}}(u_{2}) = 2$
	%
	$u_{1}$
	%
	$B_{3}$-
	%
	%
	$B_{3}$-
	\fi
	\ifnum \count11 > 0
	%
	%
	\noindent
	{\bf\boldmath $B_{2}$-block, $B_{4}$-block:}
	$B_{2}$-blocks and $B_{4}$-blocks consist of three vertices each. 
	The expected cost for any vertex is at most one. 
	Hence, 
	the expected costs for these blocks are at most three each
	(Fig.~\ref{fig:block_cost_init}). 
	\noindent
	{\bf\boldmath $B_{3}$-block:}
	Suppose that a $B_{3}$-block consists of three vertices $u_{1},u_{2}$ and $u_{3}$, 
	${deg}(u_{1}) = {deg}(u_{3}) = 1$ and ${deg}(u_{2}) = 3$. 
	Without loss of generality, 
	suppose that $u_{1}$ is revealed before $u_{3}$. 
	First, we consider the case in which $u_{1} = v_{1}$. 
	Clearly, 
	the expected costs for $u_{1}$ and $u_{2}$ are at most one each. 
	Since ${deg}_{u_{3}}(u_{2}) = 2$ 
	by the definition of {\sc BlockRoutine}, 
	the expected cost for $u_{3}$ is $1/2$ 
	by (4-1) in Lemma~\ref{LMA:rand_up.vtx_cost}. 
	By the above argument, 
	the expected cost for a $B_{3}$-block is at most $5/2$ in this case. 
	Next, we consider the case in which $u_{2} = v_{1}$. 
	The expected cost for $u_{2}$ is clearly at most one. 
	The vertex adjacent to $u_{2}$ which is neither $u_{1}$ nor $u_{3}$ is revealed after $u_{3}$ 
	by the definition of {\sc BlockRoutine}. 
	Thus, 
	${deg}_{u_{1}}(u_{2}) = 1$
	and 
	${deg}_{u_{3}}(u_{2}) = 2$. 
	Hence, 
	the expected costs for $u_{1}$ and $u_{3}$ are $1/2$ each 
	by (4-1) in Lemma~\ref{LMA:rand_up.vtx_cost}. 
	By summing up the above costs, 
	the expected cost for a $B_{3}$-block is two in this case. 
	Therefore, 
	the expected cost for a $B_{3}$-block is at most $5/2$. 
	\fi
\end{proof}

\ifnum \count12 > 0
\begin{figure*}[ht]
	 \begin{center}
	  \includegraphics[width=100mm]{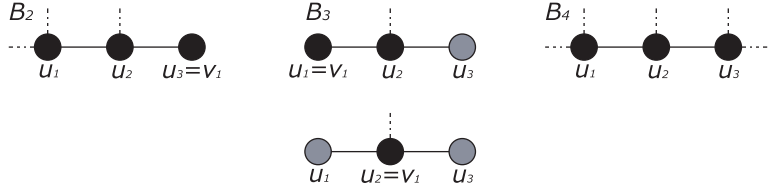}
	 \end{center}
	 \caption{
				\ifnum \count10 > 0
				%
				%
				%
				\fi
				\ifnum \count11 > 0
				%
				%
				Costs for blocks with $v_{1}$. 
				Both $A$ and $B$ select black vertices and 
				only the one of $A$ and $B$ selects gray vertices. 
				%
				\fi
			}
	\label{fig:block_cost_init}
\end{figure*}
\fi
\ifnum \count10 > 0
%
%
$B_{1,0},B_{1,1},B_{2},B_{3},B_{4}^{000},B_{4}^{100},B_{4}^{010}$
$b_{1,0}, b_{1,1}, b_{2}, b_{3}, b_{4}^{000}, b_{4}^{100}, b_{4}^{010}$
%
$b_{4} = b_{4}^{000} + b_{4}^{100} + b_{4}^{010}$
%

%
\fi
\ifnum \count11 > 0
%
%
Let 
$b_{1,0}, b_{1,1}, b_{2}, b_{3},b_{4}^{000}, b_{4}^{100}$ and $b_{4}^{010}$
denote the numbers of 
$B_{1,0}$-blocks, $B_{1,1}$-blocks, $B_{2}$-blocks, $B_{3}$-blocks, $B_{4}^{000}$-blocks, $B_{4}^{100}$-blocks and  $B_{4}^{010}$-blocks, respectively. 
We define 
$b_{4} = b_{4}^{000} + b_{4}^{100} + b_{4}^{010}$. 
\fi
%
\begin{LMA} \label{LMA:rand_up.b1cnd}
	\ifnum \count10 > 0
	%
	%
	\begin{equation} \label{LMA:rand_up.b1cnd.b10}
		b_{1, 0} \leq b_{2} + b_{4}^{100} + b_{4}^{010}
	\end{equation}
	\begin{equation} \label{LMA:rand_up.b1cnd.b11}
		b_{1, 1} \leq b_{4}^{100}
	\end{equation}
	%
	
	%
	\fi
	\ifnum \count11 > 0
	%
	%
	If the number of all the vertices of a given graph is at least five, 
	\begin{equation} \label{LMA:rand_up.b1cnd.b10}
		b_{1, 0} \leq b_{2} + b_{4}^{100} + b_{4}^{010}
	\end{equation}
	and 
	\begin{equation} \label{LMA:rand_up.b1cnd.b11}
		b_{1, 1} \leq b_{4}^{100}. 
	\end{equation}
	\fi
\end{LMA}
\begin{proof}
	\ifnum \count10 > 0
	%
	%
	%
	%
	$B_{1}$
	%
	$B_{2}^{010}, B_{4}^{000}, B_{4}^{100}, B_{4}^{010}$
	%
	%
	3
	$u_{1}$
	\noindent
	{\bf \boldmath $B_{2}^{010}$-
	${deg}(u_{1}) = {deg}(u_{2}) = 3$
	%
	$u (\ne u_{2})$
	$OPT$
	$u$
	$u$
	(P3)
	$u$
	$u$
	%
	%
	%
	$B_{1}$-
	$u$
	%
	%
	$u' (\ne u_{1}, u_{3})$
	{\sc BlockRoutine}
	$u'$
	$u_{1}$
	%
	$u'$
	${deg}_{u'}(u_{2}) = 3$
	%
	$u'$
	\noindent
	{\bf \boldmath $B_{4}^{000}$-
	%
	$OPT$
	%
	(P3)
	$u_{1}$
	$B_{2}^{010}$-
	$u''$
	%
	$B_{4}^{000}$-
	\noindent
	{\bf \boldmath $B_{4}^{100}$-
	$OPT$
	$u_{3}$
	$B_{4}^{000}$-
	%
	%
	$u_{2}$
	$B_{2}^{010}$-
	%
	%
	$u_{1}$
	$\hat{u}$
	%
	$u_{1}$
	$u_{1} = v_{1}$
	{\sc BlockRoutine}
	${deg}_{\hat{u}}(u_{1}) = 2$
	${deg}_{\hat{u}'}(u_{1}) = 3$
	%
	$B_{1,0}$
	$u_{1}$
	%
	%
	$u_{1} \ne v_{1}$
	$\hat{u}$
	${deg}_{\hat{u}'}(u_{1}) = 3$
	%
	$u_{2}$
	${deg}_{\hat{u}}(u_{1}) = 2$
	${deg}_{\hat{u}'}(u_{1}) = 3$
	%
	%
	$B_{4}^{100}$
	%
	
	%
	\noindent
	{\bf \boldmath $B_{4}^{010}$-
	%
	$OPT$
	%
	$u_{1}$
	%
	$B_{1,0}$-
	%
	
	%
	$B_{2}^{010}, B_{4}^{100}, B_{4}^{010}$
	%
	$B_{4}^{100}$
	$B_{1,1}$-
	%
	
	%
	\fi
	\ifnum \count11 > 0
	%
	%
	We discuss whether vertices in each block are adjacent to the vertex in a $B_{1}$-block. 
	If the vertex in a $B_{3}$-block is adjacent to the vertex in another $B_{1}$-block, 
	the number of vertices of a given tree is at most four, 
	which contradicts the assumption of this lemma. 
	This discussion is also true with a $B_{1}$-block. 
	In what follows, 
	we consider 
	$B_{2}^{010}$-blocks, $B_{4}^{000}$-blocks, $B_{4}^{100}$-blocks and $B_{4}^{010}$-blocks. 
	Since these blocks consist of three vertices each, 
	suppose that they are $u_{1},u_{2}$ and $u_{3}$, and 
	both $u_{1}$ and $u_{3}$ are adjacent to $u_{2}$. 
	\noindent
	{\bf \boldmath $B_{2}^{010}$-block:}
	Suppose that ${deg}(u_{1}) = {deg}(u_{2}) = 3$ and ${deg}(u_{3}) = 1$. 
	$OPT$ selects $u_{2}$ by definition. 
	Let $u (\ne u_{2})$ be a vertex adjacent to $u_{1}$. 
	$OPT$ must select $u$ or a vertex adjacent to $u$ to dominate $u$. 
	If $OPT$ selects $u$, 
	the degree of $u$ is three by (P3). 
	Otherwise, 
	it is at least two because there exists an adjacent vertex to dominate $u$
	(i.e., by the assumption of an input, it is exactly three). 
	Since the degree of the vertex in a $B_{1}$-block is one, 
	$u$ is not contained in a $B_{1}$-block. 
	Next, 
	let $u' (\ne u_{1}, u_{3})$ be the vertex adjacent to $u_{2}$. 
	By the definition of {\sc BlockRoutine}, 
	$u'$ is revealed either before $u_{2}$ or 
	after both $u_{1}$ and $u_{3}$. 
	In the former case, 
	$u'$ is contained in a block of three vertices, 
	which is not a $B_{1}$-block. 
	In the latter case, 
	${deg}_{u'}(u_{2}) = 3$. 
	Therefore, 
	$u'$ can be contained only in a $B_{1,0}$-block. 
	\noindent
	{\bf \boldmath $B_{4}^{000}$-block:}
	$OPT$ does not select any vertex in this block by definition, 
	and hence 
	each vertex is adjacent to a vertex selected by $OPT$ in another block. 
	The degree of the dominating vertex is three by (P3). 
	Both $u_{1}$ and $u_{3}$ can be adjacent to a vertex $u''$ which is not selected by $OPT$. 
	In a similar way to $u$ in the case of a $B_{2}^{010}$-block, 
	the degree of $u''$ is three. 
	Thus, 
	no vertices in a $B_{4}^{000}$-block are adjacent to the vertex in a $B_{1}$-block. 
	\noindent
	{\bf \boldmath $B_{4}^{100}$-block:}
	Suppose that $OPT$ selects $u_{1}$. 
	The degree of a vertex adjacent to $u_{3}$ in this block is three 
	similarly to $u_{3}$ in a $B_{4}^{000}$-block. 
	The degree of a vertex adjacent to $u_{2}$ is also three 
	similarly to $u$ in the case of a $B_{2}^{010}$-block. 
	We consider the two vertices $\hat{u}$ and $\hat{u}'$ which are adjacent to $u_{1}$, in which  
	$\hat{u}$ is revealed before $\hat{u}'$. 
	First, 
	suppose that 
	$u_{1}$ is revealed before $u_{2}$. 
	If $u_{1} = v_{1}$, 
	then ${deg}_{\hat{u}}(u_{1}) = 2$
	and 
	${deg}_{\hat{u}'}(u_{1}) = 3$ 
	by the definition of {\sc BlockRoutine}. 
	Hence, 
	$\hat{u}$ and $\hat{u}'$ can be the vertices in a $B_{1,1}$-block and a $B_{1,0}$-block, respectively, 
	by the definitions of $B_{1,0}$-blocks and $B_{1,1}$-blocks. 
	If $u_{1} \ne v_{1}$, 
	then $\hat{u}$ is revealed before $u_{1}$ 
	and ${deg}_{\hat{u}'}(u_{1}) = 3$. 
	Next, 
	suppose that $u_{2}$ is revealed before $u_{1}$. 
	Then, 
	${deg}_{\hat{u}}(u_{1}) = 2$
	and 
	${deg}_{\hat{u}'}(u_{1}) = 3$. 
	By the above argument, 
	at most one vertex in a $B_{4}^{100}$-block can be adjacent to the vertex in a $B_{1,0}$-block (a $B_{1,1}$-block). 
	\noindent
	{\bf \boldmath $B_{4}^{010}$-block:}
	$OPT$ selects $u_{2}$ by definition. 
	In a similar way to the above argument, 
	the degree of a vertex adjacent to either $u_{1}$ or $u_{3}$ is three. 
	Also, 
	the vertex adjacent to $u_{2}$ in another block can be contained in a $B_{1,0}$-block. 
	Therefore, 
	at most one vertex in a $B_{2}^{010}$-block, a $B_{4}^{100}$-block, or a $B_{4}^{010}$-block can be adjacent to the vertex in a $B_{1,0}$-block. 
	That is, 
	Eq.~(\ref{LMA:rand_up.b1cnd.b10}) holds. 
	At most one vertex in a $B_{4}^{100}$-block can be adjacent to the vertex in a $B_{1,1}$-block. 
	That is, 
	Eq.~(\ref{LMA:rand_up.b1cnd.b11}) holds. 
	\fi
\end{proof}
\ifnum \count10 > 0
%
%

%

%
\fi
\ifnum \count11 > 0
%
%

%
\fi
\begin{LMA} \label{LMA:rand_up.leafcnd}
	\ifnum \count10 > 0
	%
	%
	$b_{1,0} + b_{1,1} + b_{3} = b_{2} + 3 b_{4} + 2$
	%
	
	%
	\fi
	\ifnum \count11 > 0
	%
	%
	$b_{1,0} + b_{1,1} + b_{3} = b_{2} + 3 b_{4} + 2$. 
	\fi
\end{LMA}
\begin{proof}
	\ifnum \count10 > 0
	%
	%
	%
	$B_{1}$-
	%
	%
	%
	$B_{1}$-
	%
	%
	\begin{equation} \label{LMA:rand_up.leafcnd.eq1}
		b_{1, 0} + b_{1, 1} + b_{3} = o_{2} + o_{4}
	\end{equation}
	%
	%
	1
	\begin{equation} \label{LMA:rand_up.leafcnd.eq2}
		3 b_{2} = i_{2} + o_{2}
	\end{equation}
	%
	1
	\begin{equation} \label{LMA:rand_up.leafcnd.eq3}
		5 b_{4} = i_{4} + o_{4}
	\end{equation}
	%
	
	%
	%
	$x = y + 1$
	%
	%
	\begin{equation} \label{LMA:rand_up.leafcnd.eq4}
		b_{2} + b_{4} = (i_{2} + i_{4}) / 2 + 1
	\end{equation}
	%
	\begin{eqnarray*}
		3 b_{2} + 5 b_{4} 
			&=& i_{2} + o_{2} + i_{4} + o_{4}\\
			&=& b_{1, 0} + b_{1, 1} + b_{3} + 2 b_{2} + 2 b_{4} - 2 \mbox{\hspace*{10mm} (
	\end{eqnarray*}
	%
	
	%
	\fi
	\ifnum \count11 > 0
	%
	%
	Let $i_{2}$ denote the number of edges 
	each of which is an edge between a vertex in a $B_{2}$-block and a vertex in either another $B_{2}$-block or a $B_{4}$-block. 
	Let $o_{2}$ denote the number of edges 
	each of which is an edge between a vertex in a $B_{2}$-block and a vertex in either a $B_{1}$-block or a $B_{3}$-block. 
	Similarly, 
	let $i_{4}$ denote the number of edges 
	each of which is an edge between a vertex in a $B_{4}$-block and a vertex in either another $B_{4}$-block or a $B_{2}$-block. 
	Let $o_{4}$ denote the number of edges 
	each of which is an edge between a vertex in a $B_{4}$-block and a vertex in either a $B_{1}$-block or a $B_{3}$-block. 
	By definition, 
	\begin{equation} \label{LMA:rand_up.leafcnd.eq1}
		b_{1, 0} + b_{1, 1} + b_{3} = o_{2} + o_{4}. 
	\end{equation}
	Since there exist three edges each of which is between a vertex in a $B_{2}$-block and one in another block, 
	\begin{equation} \label{LMA:rand_up.leafcnd.eq2}
		3 b_{2} = i_{2} + o_{2}. 
	\end{equation}
	Since there exist five edges each of which is between a vertex in a $B_{4}$-block and one in another block, 
	\begin{equation} \label{LMA:rand_up.leafcnd.eq3}
		5 b_{4} = i_{4} + o_{4}. 
	\end{equation}
	Let us call a vertex of degree at least two 
	an {\em internal vertex}. 
	For a tree, 
	$x = y + 1$ holds, 
	in which 
	$x$ is the number of internal vertices and 
	$y$ is the number of edges between internal vertices. 
	If we regard a block as one vertex, 
	then the number of internal vertices is $b_{2} + b_{4}$ and 
	the number of edges between internal vertices is $(i_{2} + i_{4}) / 2$. 
	Thus, 
	\begin{equation} \label{LMA:rand_up.leafcnd.eq4}
		b_{2} + b_{4} = (i_{2} + i_{4}) / 2 + 1. 
	\end{equation}
	By summing  up Eqs.~(\ref{LMA:rand_up.leafcnd.eq2}) and (\ref{LMA:rand_up.leafcnd.eq3}), 
	\begin{eqnarray*}
		3 b_{2} + 5 b_{4} 
			&=& i_{2} + o_{2} + i_{4} + o_{4}\\
			&=& b_{1, 0} + b_{1, 1} + b_{3} + 2 b_{2} + 2 b_{4} - 2.  \mbox{\hspace*{10mm} (by Eqs.~(\ref{LMA:rand_up.leafcnd.eq1}) and (\ref{LMA:rand_up.leafcnd.eq4}))}
	\end{eqnarray*}
	We complete the proof by rearranging this equality. 
	\fi
\end{proof}
\ifnum \count10 > 0
%
%

%

%
\fi
\ifnum \count11 > 0
%
%
%
\fi
\begin{theorem}\label{thm:rand_up}
	\ifnum \count10 > 0
	%
	%
	$RA$
	\fi
	\ifnum \count11 > 0
	%
	%
	The competitive ratio of $RA$ is at most $5/2$. 
	\fi
\end{theorem}
\begin{proof}
	\ifnum \count10 > 0
	%
	%
	%
	$B_{1}$-
	%
	%
	%
	{\sc BlockRoutine}
	$C_{1}$
	$C_{3}$
	%
	$C_{3}$
	$C_{1}$
	$v$
	$B_{3}$-
	${deg}_{v}(u) = 3$
	$C_{1}$
	%
	%
	$C_{1}$
	%
	%
	${\mathbb E}[C_{RA}(\sigma)] = 5/2$
	%
	%
	$OPT$
	%
	$C_{OPT}(\sigma) \geq 1$
	%
	%
	%
	
	%
	%
	%
	$v_{1}$
	$v_{1}$
	$B_{4}$
	%
	$v_{1}$
	$b'_{2}, b'_{3} \in \{ 0, 1 \}$
	$b'_{2} + b'_{3} \leq 1$
	%
	\[
		{\mathbb E}[C_{RA}(\sigma)] 
			\leq b_{1,1}/2  + 5 b_{2} / 2 + 3 b_{3} / 2 + 3 b_{4} + b'_{2} / 2 + b'_{3}
			\leq b_{1,1}/2  + 5 b_{2} / 2 + 3 b_{3} / 2 + 3 b_{4} + 1
	\]
	%
	\[
		C_{OPT}(\sigma)
			=  b_{2} + b_{3} + b_{4}^{100} + b_{4}^{010}
	\]
	%
	\begin{eqnarray*}
		&& \frac{{\mathbb E}[C_{RA}(\sigma)]}{C_{OPT}(\sigma)}
			\leq 
				\frac{ b_{1,1}/2  + 5 b_{2} / 2 + 3 b_{3} / 2 + 3 b_{4} + 1 }{ b_{2} + b_{3} + b_{4}^{100} + b_{4}^{010} }\\
			&=& 
				\frac{ - 3 b_{1,0} / 2 - b_{1,1} + 4 b_{2} + 15 b_{4} / 2 + 4 }{ - b_{1,0} - b_{1,1} + 2 b_{2} + 3 b_{4} + b_{4}^{100} + b_{4}^{010} + 2 } 
					\mbox{\hspace*{7mm} (
			&\leq& 
				\frac{ - 3 b_{1,0} / 2 + 4 b_{2} + 15 b_{4} / 2 - b_{4}^{100} + 4 }{ - b_{1,0} + 2 b_{2} + 3 b_{4} + b_{4}^{010} + 2 } 
					\mbox{\hspace*{7mm} (
			&\leq& 
				\frac{ - 5 b_{1,0} / 2 + 5 b_{2} + 15 b_{4} / 2 + b_{4}^{010} + 4 }{ - b_{1,0} + 2 b_{2} + 3 b_{4} + b_{4}^{010} + 2 } 
					\mbox{\hspace*{7mm} (
			&=& 
				\frac{5}{2} \cdot
				\frac{ - b_{1,0} + 2 b_{2} + 3 b_{4} + 2 b_{4}^{010} / 5 + 8/5 }{ - b_{1,0} + 2 b_{2} + 3 b_{4} + b_{4}^{010} + 2 } 
			< \frac{5}{2}
	\end{eqnarray*}

	\fi
	\ifnum \count11 > 0
	%
	%
	First, 
	we consider an input $\sigma$ of which 
	the number of vertices of a given tree is four. 
	A combination of blocks composing a tree with four vertices consists of one $B_{1}$-block $C_{1}$ and one $B_{3}$-block $C_{3}$ 
	by Lemma~\ref{LMA:rand_up.spblock_opt}. 
	%
	%
	Since $C_{1}$ does not contain $v_{1}$ by the definition of {\sc BlockRoutine}, 
	$C_{3}$ contains $v_{1}$. 
	Thus, 
	the expected of $RA$ for $C_{3}$ is at most $5/2$
	by Lemma~\ref{LMA:rand_up.block_cost_init}. 
	Let $v$ denote the vertex in $C_{1}$, and 
	let $u$ denote the vertex adjacent to $v$ in $C_{3}$. 
	${deg}_{v}(u) = 3$ 
	by the definition of $B_{3}$-blocks, 
	which means that 
	$C_{1}$ is a $B_{1,0}$-block. 
	Thus, 
	the expected cost for $C_{1}$ is zero by Lemma~\ref{LMA:rand_up.block_cost}. 
	By the above argument, 
	${\mathbb E}[C_{RA}(\sigma)] = 5/2$. 
	On the other hand, 
	$OPT$ clearly selects at least one vertex, 
	that is, 
	$C_{OPT}(\sigma) \geq 1$. 
	Therefore, 
	we have shown the statement of the theorem in the case of a tree with four vertices. 
	Next, 
	we consider the case in which of a graph with at least five vertices. 
	The expected costs of $RA$ for a $B_{2}$-block and a $B_{3}$-block with $v_{1}$ are greater than those for a $B_{2}$-block and a $B_{3}$-block without $v_{1}$ by $1/2$ and one, respectively,  
	by Lemmas~\ref{LMA:rand_up.block_cost} and  \ref{LMA:rand_up.block_cost_init}. 
	Also, 
	$v_{1}$ does not affect the expected cost for $B_{4}$-blocks. 
	Thus, 
	let $b'_{2}$ and $b'_{3}$ denote the numbers of $B_{2}$-blocks and $B_{3}$-blocks with $v_{1}$, respectively. 
	By definition, 
	$b'_{2}, b'_{3} \in \{ 0, 1 \}$ and 
	$b'_{2} + b'_{3} \leq 1$. 
	Then, 
	using Lemma~\ref{LMA:rand_up.block_cost}, 
	we have 
	\[
		{\mathbb E}[C_{RA}(\sigma)] 
			\leq b_{1,1}/2  + 5 b_{2} / 2 + 3 b_{3} / 2 + 3 b_{4} + b'_{2} / 2 + b'_{3}
			\leq b_{1,1}/2  + 5 b_{2} / 2 + 3 b_{3} / 2 + 3 b_{4} + 1. 
	\]
	By the definitions of blocks, we have 
	\[
		C_{OPT}(\sigma)
			=  b_{2} + b_{3} + b_{4}^{100} + b_{4}^{010}. 
	\]
	By the inequality and the equality, 
	we have 
	\begin{eqnarray*}
		&& \frac{{\mathbb E}[C_{RA}(\sigma)]}{C_{OPT}(\sigma)}
			\leq
				\frac{ b_{1,1}/2  + 5 b_{2} / 2 + 3 b_{3} / 2 + 3 b_{4} + 1 }{ b_{2} + b_{3} + b_{4}^{100} + b_{4}^{010} }\\
			&=& 
				\frac{ - 3 b_{1,0} / 2 - b_{1,1} + 4 b_{2} + 15 b_{4} / 2 + 4 }{ - b_{1,0} - b_{1,1} + 2 b_{2} + 3 b_{4} + b_{4}^{100} + b_{4}^{010} + 2 } 
					\mbox{\hspace*{7mm} (by the substitution for $b_{3}$ by Lemma~\ref{LMA:rand_up.leafcnd})}\\
			&\leq& 
				\frac{ - 3 b_{1,0} / 2 + 4 b_{2} + 15 b_{4} / 2 - b_{4}^{100} + 4 }{ - b_{1,0} + 2 b_{2} + 3 b_{4} + b_{4}^{010} + 2 } 
					\mbox{\hspace*{7mm} (by the substitution for $b_{1,1}$ by Eq.~(\ref{LMA:rand_up.b1cnd.b11}))}\\
			&\leq& 
				\frac{ - 5 b_{1,0} / 2 + 5 b_{2} + 15 b_{4} / 2 + b_{4}^{010} + 4 }{ - b_{1,0} + 2 b_{2} + 3 b_{4} + b_{4}^{010} + 2 } 
					\mbox{\hspace*{7mm} (by the substitution for $b_{4}^{100}$ by Eq.~(\ref{LMA:rand_up.b1cnd.b10}))}\\
			&=& 
				\frac{5}{2} \cdot
				\frac{ - b_{1,0} + 2 b_{2} + 3 b_{4} + 2 b_{4}^{010} / 5 + 8/5 }{ - b_{1,0} + 2 b_{2} + 3 b_{4} + b_{4}^{010} + 2 } 
			< \frac{5}{2}. 
	\end{eqnarray*}
	\fi
\end{proof}
\ifnum \count10 > 0
%
%
%

%
\fi
\ifnum \count11 > 0
%
%
Our analysis of $RA$ is exact by the following theorem. 
\fi
\begin{theorem}\label{thm:rand_low_ra}
	\ifnum \count10 > 0
	%
	%
	$RA$
	\fi
	\ifnum \count11 > 0
	%
	%
	The competitive ratio of $RA$ is at least $5/2$. 
	\fi
\end{theorem}
\begin{proof}
	\ifnum \count10 > 0
	%
	%
	%
	$v_{2}$
	%
	%
	$A$
	$B$
	$OPT$
	%
	$\frac{ {\mathbb E}[C_{RA}(\sigma)] }{ C_{OPT}(\sigma) } = 5/2$
	%
	
	%
	\fi
	\ifnum \count11 > 0
	%
	%
	Consider the following input $\sigma$. 
	A vertex $v_{1}$ is first revealed, and 
	the second vertex $v_{2}$ arrives at $v_{1}$. 
	The third vertex $v_{3}$ and the fourth vertex $v_{4}$ arrive at $v_{2}$. 
	Then, 
	$A$ selects $v_{1}$ (Case~1) and $v_{2}$ (Case~2.2.2). 
	$B$ selects $v_{1}$ (Case~1), $v_{2}$ (Case~2.1) and $v_{3}$ (Case~2.2.1). 
	$OPT$ selects $v_{2}$. 
	Therefore, 
	${\mathbb E}[C_{RA}(\sigma)] / C_{OPT}(\sigma) = 5/2$. 
	\fi
\end{proof}

\section{Randomized Lower Bound} \label{sec:rand_low}
\ifnum \count10 > 0
%
%

%

%
\fi
\ifnum \count11 > 0
%
%

%
\fi
\begin{LMA} \label{LMA:rand_low.1}
	\ifnum \count10 > 0
	%
	%
	$p_u (p_v)$
	%
	$p_u + p_v \geq 1$
	%
	
	%
	\fi
	\ifnum \count11 > 0
	%
	%
	Consider a randomized online algorithm $RON$. 
	Suppose that a vertex $v$ arrives at a vertex $u$. 
	Let $p_u (p_v)$ denote the probability that $u \in D_{RON}(v) (v \in D_{RON}(v))$. 
	Then, 
	$p_u + p_v \geq 1$. 
	\fi
\end{LMA}
\begin{proof}
	\ifnum \count10 > 0
	%
	%
	$u \in D_{RON}(u)$
	%
	$u \in D_{RON}(u)$
	$p'_u$
	(
	%
	%
	$u \notin D_{RON}(u)$
	$v$
	$RON$
	%
	$u \in D_{RON}(v)$%
	%
	%
	$p'_u$
	$u \notin D_{RON}(u)$
	%
	$u \notin D_{RON}(u)$
	$u \in D_{RON}(v)$%
	$1-p'_u$
	%
	%
	$p_u + p_v \geq p'_u + 1-p'_u =1$
	%
	
	%
	\fi
	\ifnum \count11 > 0
	%
	%
	Let $p'_u$ be the probability that $u \in D_{RON}(u)$. 
	Since $RON$'s selection is irrevocable, 
	the probability that $u \in D_{RON}(u)$ and $u \in D_{RON}(v)$ is greater than or equal to $p'_u$ 
	(Fact~(a)). 
	Next, 
	we consider the case in which $u \notin D_{RON}(u)$. 
	$RON$ must select $u$ or $v$ to construct a dominating set immediately after $v$ is revealed. 
	Thus, 
	the probability that 
	either $u \in D_{RON}(v)$ or $v \in D_{RON}(v)$ is one. 
	By the definition of $p'_u$, 
	the probability that $u \notin D_{RON}(u)$ is $1 - p'_u$. 
	Hence, 
	the probability that 
	both $u \notin D_{RON}(u)$ and 
	either $u \in D_{RON}(v)$ or $v \in D_{RON}(v)$ is at least $1-p'_u$. 
	This probability together with Fact~(a) shows that 
	$p_u + p_v \geq p'_u + 1-p'_u =1$. 
	\fi
\end{proof}
\ifnum \count10 > 0
%
%

%

%
\fi
\ifnum \count11 > 0
%
%

%
\fi
\begin{theorem}\label{thm:rand_low}
	\ifnum \count10 > 0
	%
	%
	$4/3$
	\fi
	\ifnum \count11 > 0
	%
	%
	The competitive ratio of any randomized online algorithm is at least $4/3$. 
	\fi
\end{theorem}
\begin{proof}
	\ifnum \count10 > 0
	%
	%
	%
	%
	$m$
	%
	%
	%
	%
	%
	$RON$
	%
	
	%
	$i$
	%
	%
	$p_j$
	%
	%
	$v_{2m}$
	%
	
	%
	\noindent
	{\bf\boldmath Case~1 ($\min \{ p_{2 \ell - 1}, p_{2 \ell} \} \geq 2/3$): }
	%
	%
	$v_{2 \ell - 1}$
	%
	
	%
	\noindent
	{\bf\boldmath Case~2 ($\min \{ p_{2 \ell - 1}, p_{2 \ell} \} < 2/3$):}
	%
	$\ell_{1} = 2 \ell - 1$
	$\ell_{2} = 2 \ell - 1$
	%
	%
	$v_{\ell_1}$
	%
	%
	$v_{\ell_1} \in D_{RON}(u_{\ell_1})$%
	$u_{\ell_1} \in D_{RON}(u_{\ell_1})$%
	(
	%
	%
	$p_{2 \ell - 1} + p_{2 \ell} = p_{\ell_1}  + p_{\ell_2} \geq 1$
	%
	Case~2
	$p_{\ell_1} < 2/3$
	$p_{\ell_2} \geq 1/3$
	%
	$v_{\ell_1}$
	$1 + 1/3 > 4/3$
	Case~1
	Case~2
	${\mathbb E}[C_{RON}(\sigma)] \geq 4x/3 + 4(m-x)/3 = 4m/3$
	%
	
	%
	%
	%
	$OFF$
	$OFF$
	%
	%
	$OFF$
	%
	%
	$C_{OPT}(\sigma) \leq C_{OFF}(\sigma) = m$
	%
	%
	${\mathbb E}[C_{RON}(\sigma)] / C_{OPT}(\sigma) \geq 4/3$
	%
	
	%
	\fi
	\ifnum \count11 > 0
	%
	%
	Consider a randomized online algorithm $RON$ for the following input $\sigma$. 
	Let $m$ be any positive integer. 
	We sketch an adversary constructing $\sigma$. 
	First, 
	the adversary gives a line of $2m$ vertices to $ON$. 
	Then, 
	for every two consecutive vertices on the line, 
	the adversary determines whether an additional vertex will arrive at one of the two vertices. 
	Specifically, 
	if the probability that $RON$ selects at least one of the two vertices is low, 
	the adversary makes a new vertex arrive at the vertex. 
	For each $i = 1, 2, \ldots, 2m$, 
	the $i$-th vertex $v_{i}$ arrives at $v_{i-1}$. 
	For each $j = 1, 2, \ldots, 2m$, 
	let $p_j$ be the probability that $v_j \in D_{RON}(v_{2m})$. 
	Next, 
	for each $\ell = 1, 2, \ldots, m$, 
	vertices are revealed after the revelation of $v_{2m}$ in the following two cases. 
	\noindent
	{\bf\boldmath Case~1 ($\min \{ p_{2 \ell - 1}, p_{2 \ell} \} \geq 2/3$):}
	A new vertex does not arrive at either $v_{2 \ell - 1}$ or $v_{2 \ell}$. 
	Since $RON$'s selection is irrevocable, 
	the expected cost of $RON$ for $v_{2 \ell - 1}$ and $v_{2 \ell}$ is at least $2 \cdot 2/3 = 4/3$. 
	\noindent
	{\bf\boldmath Case~2 ($\min \{ p_{2 \ell - 1}, p_{2 \ell} \} < 2/3$):}
	If $p_{2 \ell - 1} \leq p_{2 \ell}$, 
	then define $\ell_{1} = 2 \ell - 1$ and $\ell_{2} = 2 \ell$. 
	Otherwise, 
	define $\ell_{2} = 2 \ell - 1$ and $\ell_{1} = 2 \ell$. 
	Then, 
	a vertex $u_{\ell_1}$ arrives at $v_{\ell_1}$. 
	By Lemma~\ref{LMA:rand_low.1}, 
	the probability that 
	$v_{\ell_1} \in D_{RON}(u_{\ell_1})$ or $u_{\ell_1} \in D_{RON}(u_{\ell_1})$ is at least one
	(Fact (a)). 
	Moreover, 
	$p_{2 \ell - 1} + p_{2 \ell} = p_{\ell_1}  + p_{\ell_2} \geq 1$ also holds. 
	Since $p_{\ell_1} < 2/3$ by the condition of Case~2, 
	$p_{\ell_2} \geq 1/3$. 
	Hence, 
	the expected cost for $v_{\ell_1}$, $v_{\ell_2}$ and $u_{\ell_1}$ is at least $1 + 1/3 = 4/3$. 
	Let $x$ be the number of such $\ell$ with applying Case~1. 
	Thus, 
	the number of such $\ell$ with applying Case~2 is $m - x$. 
	${\mathbb E}[C_{RON}(\sigma)] \geq 4x/3 + 4(m-x)/3 = 4m/3$  
	by the above argument. 
	Next, 
	we consider an offline algorithm $OFF$ to give an upper bound on the cost of $OPT$. 
	For such $\ell$ with applying Case~1, 
	$OFF$ selects $v_{2 \ell - 1}$ and 
	for such $\ell$ with applying Case~2, selects $v_{\ell_{1}}$. 
	Thus, 
	$OFF$ selects $m$ vertices, and 
	the set of the $m$ selected vertices is clearly a dominating set. 
	By the optimality of $OPT$, 
	$C_{OPT}(\sigma) \leq C_{OFF}(\sigma) = m$. 
	Therefore, 
	${\mathbb E}[C_{RON}(\sigma)] / C_{OPT}(\sigma) \geq 4/3$. 
	\fi
\end{proof}

\section{Conclusions}
\ifnum \count10 > 0
%
%
%
%
%
%

%
(i)
\cite{ES2002,BEFKL2016}
%
%
(ii)
%
%
(iii)
%

%
\fi
\ifnum \count11 > 0
%
%
In this paper, 
we have conducted research on algorithms for an online variant of the minimum dominating set problem on trees and obtained the following results: 
First, 
we have shown that the competitive ratio of any deterministic algorithm is at least 3, 
which matches the upper bound shown in \cite{ES2002,BEFKL2016}. 
Then, 
we have designed an algorithm whose competitive ratio is exactly $5/2$ using randomization. 
Furthermore, 
we have shown that the competitive ratio of any randomized algorithm is at least $4/3$. 
We conclude this paper by providing open questions: 
(i)
Online algorithms for dominating sets on several graph classes have been discussed in \cite{ES2002,BEFKL2016} 
and optimal online algorithms have not yet known on some classes. 
Then, 
it is interesting to consider online algorithms on other classes in addition to them. 
(ii)
Our algorithm $RA$ is the first randomized algorithm for the online dominating set problem on trees
and can achieve a competitive ratio smaller than that of any deterministic algorithm. 
Can we also obtain a better ratio on other classes using randomization?
(iii)
The gap between the randomized bounds shown in this paper is still large and thus, 
it is an obvious open problem to close the gap. 
%
\fi
%

%


\end{document}